\documentclass[11pt,a4paper]{article}
\usepackage{amsmath, amsfonts, amsthm, amsxtra, amssymb, bm, graphicx, paralist, url, verbatim, xcolor}
\usepackage{subfigure}
\usepackage{parskip}
\theoremstyle{definition}
\newtheorem{prob}{}

\usepackage{esint}
\usepackage[top=2.5cm,bottom=2.5cm,left=2cm,right=2cm]{geometry}
\usepackage[nolists]{endfloat}

\theoremstyle{definition}

\newtheorem{eg}{Example}[section]

\theoremstyle{plain}
\newtheorem{thm}{Theorem}[section]
\newtheorem{lemma}[thm]{Lemma}

\theoremstyle{remark}
\newtheorem*{rem}{Remark}

\newcommand{\R}{\mathbb{R}}

\newcommand{\Exp}{\mathbb{E}}
\newcommand{\Prob}{\mathbb{P}}

\newcommand{\abs}[1]{\left| #1 \right|}
\newcommand{\vecnorm}[1]{\left|\left|#1\right|\right|}

\newcommand{\rd}{\mathrm{d}}
\newcommand{\pd}{\partial}

\numberwithin{equation}{section}

\graphicspath{{Figures/}}

\usepackage[numbers]{natbib}

\title{{\LARGE \bfseries Biomembranes report} \\ {\small MASDOC, RSG} \\ }
\author{\itshape{Don Praveen Amarasinghe} \\ \itshape{Andrew Aylwin} \\ \itshape{Pravin Madhavan} \\ \itshape{Chris Pettitt}}
\date{May 13th 2011}

\begin{document}

	\maketitle
	\tableofcontents

	\newpage

	\section{Introduction} \label{sec:Intro}
		Chemotaxis --- the response of cells to chemical changes in their surroundings --- plays a crucial role in the study of organisms. An example of chemotaxis in action is depicted in the remarkable video of a neutrophil (white blood cell) chasing a \emph{Staphylococcus aureus} bacterium (available at \url{http://www.biochemweb.org/neutrophil.shtml}). The goal of this project is to develop a model and simulation for the movement of the neutrophil in this situation.

		The model we consider is a so-called `signal centred' approach. Receptors studded throughout the neutrophil's cell membrane detect the levels of chemoattractants in the surroundings. This includes chemicals secreted by the bacterium nearby. The resulting chemoattractant concentration gradient in the neutrophil's surroundings stimulate the cell's movement. Signalling pathways induce protrusions in the membrane in the direction of increasing chemoattractant levels; hence towards the bacterium. The rest of the cell then follows in that direction through the resulting tension in the cell membrane. It should be noted that the neutrophil may respond to the presence of other cells (such as red blood cells) as obstructions. However, we shall ignore this factor for the purposes of this model.

		The chemotaxis model of Neilson, Mackenzie, Insall and Webb in \cite{Neil10} acts as our starting point. The reaction-diffusion equations for this setup are reproduced below.
		\begin{subequations} \label{eq:Neilson}
			\begin{align}
				\label{eq:Neilson-a} \pd^{\bullet}_{t} a + a \nabla_{\Gamma(t)} \cdot v &= D_{a} \Delta_{\Gamma(t)} a + \frac{s(\frac{a^{2}}{b} + b_{a})}{(s_{c} + c)(1 + a^{2}s_{a})} - r_{a}a.
				\\
				\label{eq:Neilson-b} \pd^{\bullet}_{t} b + b \nabla_{\Gamma(t)} \cdot v &= D_{b} \Delta_{\Gamma(t)} b - r_{b}b + r_{b} \fint_{\Gamma(t)} a\, \rd x.
				\\
				\label{eq:Neilson-c} \pd^{\bullet}_{t} c + c \nabla_{\Gamma(t)} \cdot v &= D_{c} \Delta_{\Gamma(t)} c - r_{c}c + b_{c}a.
			\end{align}
		\end{subequations}
		where
		\begin{itemize}
			\item The cell boundary, $\Gamma(t)$, is a compact, smooth, connected and oriented curve in $\R^{2}$ for each $t \in [0,T]$, and moves with velocity $v = V_{f} \nu$, where $\nu$ is the outward normal to $\Gamma(t)$.
			\item $a$ denotes a local autocatalytic activator (or attractant), with associated decay rate $r_{a}$ and diffusion coefficient $D_{a}$
			\item $b$ denotes a rapidly distributed global inhibitor, with associated decay rate $r_{b}$ and diffusion coefficient $D_{b}$
			\item $c$ denotes a local inhibitor, with associated decay rate $r_{c}$ and diffusion coefficient $D_{c}$
			\item In equation \eqref{eq:Neilson-a}, $s_{a}$ is a saturation coefficient, $s_{c}$ the Michaelis-Menten constant and $b_{a}$ a basal production rate of the activator
			\item $b_{c}$ in equation \eqref{eq:Neilson-c} determines the growth of the local inhibitor $c$ in the presence of the activator $a$
			\item $s$ incorporates the effect of the external signal and random fluctuations. In \cite{Neil10}, $s$ is defined as
			\begin{align} \label{eq:StochasticTerm}
				s(x, t) = r_{a}[ (1 + dr \mathrm{RND}) + R_{0}( 1 + dr \mathrm{RND})],
			\end{align}
			where $dr$ is a positive parameter, $\mathrm{RND} \in (0,1)$ is a uniformly distributed random variable, and $R_{0}$ is a non-constant number defined in \cite{Neil10} related to the local concentration of chemoattractant.
		\end{itemize}
		In addition, the paper couples this system of reaction-diffusion equations to a PDE controlling the movement of the cell membrane. This PDE is given by
		\begin{align} \label{eq:MembraneMoveOriginal}
			u_{t} \cdot \nu = V_{f} - \lambda \kappa
		\end{align}
		where $u$ is the velocity field, $\nu$ the outward unit normal to the cell surface $\Gamma(t)$, $V_{f} = K_{\mathrm{prot}} a$ with $K_{\mathrm{prot}}$ a positive parameter, $\kappa$ the cortical tension term, and $\lambda$ a term to ensure that the area of the cell is controlled: in particular, it is a solution to the non-linear ODE
		\begin{align} \label{eq:MembraneMoveLambdaODE}
			\frac{\rd \lambda}{\rd t} = \frac{\lambda_{0} \lambda \left( A - A_{0} + \frac{\rd A}{\rd t} \right)}{A_{0} \left( \lambda + \lambda_{0} \right)} - \beta \lambda
		\end{align}
		with $\lambda_{0}$ and $\beta$ positive constants, $A(t)$ the area of the cell and $A_{0}$ the initial cell area.
		
		This project proposes a modified version of the above model by incorporating a new model for the neutrophil membrane movement. In Sections~\ref{sec:Model} and~\ref{sec:Numerics}, we formulate the mathematical details of the normalised chemotaxis model and describe numerical methods to simulate the neutrophil movement. Section~\ref{sec:Prob} will employ the use of SDEs in establishing escape probabilities of a bacterium being chased by a neutrophil. Finally in Section~\ref{sec:ModelComparison}, we consider an empirical model by Li \emph{et al.} in \cite{Li2008} to compare the predictions of our model against experimental observations and consider how the neutrophil's search strategy in the absence of a chemoattractant signal improves its chances of finding a bacterium.

	\section{Model Formulation} \label{sec:Model}
		\subsection{Normalisation and Reduction of the Chemotaxis Model} \label{subsec:NormRed}
			In Figure 3 of \cite{Neil10}, simulations of the chemotaxis model suggest that, whilst the concentration of local activator $a$ varies along the neutrophil cell's membrane, the concentration of global inhibitor $b$ remains constant. Furthermore, it appears that the concentration of local inhibitor $c$ follows a similar pattern to that of $a$, but is always a fraction of the value. We want to investigate if the chemotaxis model gives rise to this behaviour. Furthermore, these observations suggest that equation \eqref{eq:Neilson-a} alone determines the dynamics of the model, and we would like to ascertain whether or not this is the case. Our first task is to rescale the reaction-diffusion equations \eqref{eq:Neilson} of the Neilson \emph{et al.} model. By doing this, we will be able to reduce the system down to just one key equation. We will then compare our reduced system with the data obtained through simulations (as described in \cite{Neil10}). We shall use normalised variables $x'$, $t'$, $a'$, $b'$ and $c'$, where
			\begin{align*}
				x = Lx', \quad t = Tt', \quad a = Aa', \quad b = Ab', \quad \mbox{and} \quad c = Ac'
			\end{align*}
			and we denote $\Gamma'$ be a reparametrisation of $\Gamma$ under $t'$ and $x'$.The normalised form of \eqref{eq:Neilson} is then given by
			\begin{subequations} \label{eq:NormNeilson}
				\begin{align}
					\label{eq:NormNeilson-a} \pd^{\bullet}_{t'} a' + a' \nabla_{\Gamma '(t')} \cdot v'
					&= T \left( \frac{D_{a}}{L^{2}} \Delta_{\Gamma '(t')} a' + \frac{s(\frac{a'^{2}}{b'} + \frac{b_{a}}{A})}{(s_{c} + Ac')(1 + A^{2}(a')^{2}s_{a})} - r_{a}a' \right),
					\\
					\label{eq:NormNeilson-b} \pd^{\bullet}_{t'} b' + b' \nabla_{\Gamma '(t')} \cdot v'
					&= T \left( \frac{D_{b}}{L^{2}} \Delta_{\Gamma '(t')} b' - r_{b} \left( b' - \fint_{\Gamma '(t')} a'\, \rd x' \right) \right),
					\\
					\label{eq:NormNeilson-c} \pd^{\bullet}_{t'} c' + c' \nabla_{\Gamma '(t')} \cdot v'
					&= T \left( \frac{D_{c}}{L^{2}} \Delta_{\Gamma '(t')} c' - r_{c}c' + b_{c}a' \right)
				\end{align}
			\end{subequations}
			The constants in these equations have specific values in \cite{Neil10} --- these are reproduced in Table~\ref{tab:EqnConstants}.
%			THE INCLUDED TABLE IS GIVEN AT THE END OF THE DOCUMENT
%			ANY EDITS MADE MUST BE MADE TO BOTH COPIES!!!
%			\begin{table}
%				\centering
%				\begin{tabular}{|c|l|c|}
%					\hline
%					Quantity & Description & Value
%					\\
%					\hline
%					$r_{a}$ & Decay rate of activator & $2 \times 10^{-2}$
%					\\
%					$r_{b}$ & Decay rate of global inhibitor & $3 \times 10^{-2}$
%					\\
%					$r_{c}$ & Decay rate of local inhibitor & $1.3 \times 10^{-2}$
%					\\
%					$D_{a}$ & Diffusion Coefficient of activator & $4 \times 10^{-7}$
%					\\
%					$D_{b}$ & Diffusion Coefficient of global inhibitor & $4$
%					\\
%					$D_{c}$ & Diffusion Coefficient of local inhibitor & $2.8 \times 10^{-6}$
%					\\
%					$b_{a}$ & Basal activator production rate & $1 \times 10^{-1}$
%					\\
%					$b_{c}$ & Local inhibitor production rate & $5 \times 10^{-3}$
%					\\
%					$s_{a}$ & Saturation Coefficient & $5 \times 10^{-4}$
%					\\
%					$s_{c}$ & Michaelis-Menten constant & $2 \times 10^{-1}$
%					\\
%					\hline
%				\end{tabular}
%				\caption{Table of Reaction-Diffusion equation constants}
%				\label{tab:EqnConstants}
%			\end{table}
			We would like to choose suitable values for $A$, $L$ and $T$ to simplify this system so that the dynamics for $a'$, $b'$ and $c'$ can be summarised in one, non-linear PDE. Since \eqref{eq:NormNeilson-a} appears to be the most complicated, involving both a non-linear part and a random term (the value $s$), we shall try to analyse the other two equations in order to see if we can express $b'$ and $c'$ as functions of $a'$. We can choose:
			\begin{itemize}
				\item $A$ to be the maximum observed value of $a$, which leaves $a'$ roughly of order 1. According to the data of Figure 3 in \cite{Neil10}, we would choose $A = 25$.
				\item $L = \abs{\Gamma(0)} = \int_{\Gamma(0)} \rd x$, so that $\abs{\Gamma '(0)} = 1$, and then we can compare cell lengths on a scale basis. It is not clear from \cite{Neil10} what $L$ should be since the data provided gives details of the cell shape in the later stages of the simulation rather than the initial phase. We guess that a value of $L = 0.5$ will do.
				\item $T$ to scale with the effect of the Laplace-Beltrami operator in \eqref{eq:NormNeilson-a}. Specifically, we would like $\frac{D_{a}T}{L^{2}} = 1$, so $T = \frac{L^{2}}{D_{a}} \approx 6.25 \times 10^{5}$.
			\end{itemize}
			We would like to show the following:
			\begin{thm}[Reduced Reaction-Diffusion System] \label{thm:ReducedSys}
				With the above choices of $A$, $L$ and $T$, the system of equations \eqref{eq:NormNeilson} can be approximately reduced down to
				\begin{subequations} \label{eq:NormNeilsonReduced}
					\begin{align}
						\label{eq:NormNeilsonReduced-a} \pd^{\bullet}_{t'} a' + a' \nabla_{\Gamma '(t')} \cdot v'
						&= \Delta_{\Gamma '(t')} a' + \frac{Ts(\frac{a'^{2}}{b'} + \frac{b_{a}}{A})}{(s_{c} + Ac')(1 + A^{2}(a')^{2}s_{a})} - Tr_{a}a'
						\\
						\label{eq:NormNeilsonReduced-b} b'
						&= \fint_{\Gamma '(t')} a'\, \rd x',
						\\
						\label{eq:NormNeilsonReduced-c} c'
						&= \frac{\widetilde{b_{c}}}{\widetilde{r_{c}}} a' \approx 0.385 a'.
					\end{align}
				\end{subequations}
			\end{thm}
			We require a few results to prove this. The first of these is the Transport Identity \cite{Elli05}.
			\begin{align} \label{eq:Transport}
				\frac{\rd}{\rd t}\int_{\Gamma(t)} f = \int_{\Gamma(t)} \pd^{\bullet}_{t} f + f \nabla_{\Gamma(t)} \cdot v
			\end{align}
			The second is an approximation result.
			\begin{lemma} \label{lemma:bEqnApprox}
				For $\varepsilon > 0$, let $b(\cdot,t) , b_{\epsilon}(\cdot,t) : \Gamma(t) \rightarrow \R$ be $L^{2}$ functions satisfying
				\begin{subequations} \label{eq:NormNeilson-bDeriv1}
					\begin{align}
						\label{eq:NormNeilson-bDeriv1a} -\Delta_{\Gamma(t)} b_{\varepsilon} + C_{b} (b_{\varepsilon} - \overline{a}) 
						&= -\varepsilon \left( \pd^{\bullet}_{t} b_{\varepsilon} + b_{\varepsilon} \nabla_{\Gamma(t)} \cdot v \right)
						\\
						\label{eq:NormNeilson-bDeriv1b} -\Delta_{\Gamma(t)} b + C_{b} (b - \overline{a})
						&= 0
					\end{align}
				\end{subequations}
				Assume that there exist constants $C_{v}, C_{d}, C_{l} < \infty$ such that
				\begin{itemize}
					\item $\vecnorm{\nabla_{\Gamma} \cdot v}_{L^{\infty}} \leq C_{V}$
					\item $\vecnorm{\pd^{\bullet}_{t} b + b \nabla_{\Gamma(t)} \cdot v}_{L^{\infty}} \leq C_{D}$
					\item $\abs{\Gamma(t)} \leq C_{L}^{2}, \, \forall t \geq 0$
				\end{itemize}
				Then, for small enough $\varepsilon$, $b_{\varepsilon}$ tends to $b$ in $L^{2}$ as $\varepsilon \rightarrow 0$.
			\end{lemma}
			\begin{proof}
				We start by subtracting \eqref{eq:NormNeilson-bDeriv1b} away from \eqref{eq:NormNeilson-bDeriv1a} to get
				\begin{align} \label{eq:NormNeilson-bDeriv2}
					-\Delta_{\Gamma(t)} \left( b_{\varepsilon} - b \right) + C_{b} (b_{\varepsilon} - b) = -\varepsilon \left( \pd^{\bullet}_{t} b_{\varepsilon} + b_{\varepsilon} \nabla_{\Gamma(t)} \cdot v \right).
				\end{align}
				Now let $d := b_{\varepsilon} - b$, multiply \eqref{eq:NormNeilson-bDeriv2} by $d$ and integrate w.r.t. $x$ over $\Gamma(t)$ to obtain
				\begin{align} \label{eq:NormNeilson-bDeriv3}
					\int_{\Gamma(t)} \left( \left| \nabla_{\Gamma(t)} d \right|^{2} + C_{b} d^{2} + \varepsilon d \left( \pd^{\bullet}_{t} b_{\varepsilon} + b_{\varepsilon} \nabla_{\Gamma(t)} \cdot v \right) \right) \rd x = 0.
				\end{align}
				Adding and subtracting $\varepsilon d \left( \pd^{\bullet}_{t} b + b \nabla_{\Gamma(t)} \cdot v \right)$ gives
				\begin{align} \label{eq:NormNeilson-bDeriv4}
					\int_{\Gamma(t)} \left( \left| \nabla_{\Gamma(t)} d \right|^{2} + C_{b} d^{2} + \varepsilon d \left( \pd^{\bullet}_{t} d + d \nabla_{\Gamma(t)} \cdot v + \pd^{\bullet}_{t} b + b \nabla_{\Gamma(t)} \cdot v \right) \right) \rd x = 0.
				\end{align}
				We then use the Transport Identity \eqref{eq:Transport} with $f = \frac{1}{2} d^{2}$ to get
				\begin{align} \label{eq:NormNeilson-bDeriv5}
					\frac{\rd}{\rd t} \int_{\Gamma(t)} \frac{\varepsilon d^{2}}{2} \,\rd x + \int_{\Gamma(t)} \left( \left| \nabla_{\Gamma(t)} d \right|^{2} + C_{b} d^{2} + \frac{\varepsilon d^{2}}{2} \nabla_{\Gamma(t)} \cdot v + \varepsilon d \left( \pd^{\bullet}_{t} b + b \nabla_{\Gamma(t)} \cdot v \right) \right) \rd x = 0.
				\end{align}
				Thus,
				\begin{align*}
					\frac{\rd}{\rd t} \int_{\Gamma(t)} \frac{\varepsilon d^{2}}{2} \,\rd x
					&\leq \frac{\rd}{\rd t} \int_{\Gamma(t)} \frac{\varepsilon d^{2}}{2} \,\rd x + \int_{\Gamma(t)} \left| \nabla_{\Gamma(t)} d \right|^{2}
					\\
					&= -\int_{\Gamma(t)} \left( C_{b} d^{2} + \frac{\varepsilon d^{2}}{2} \nabla_{\Gamma(t)} \cdot v + \varepsilon d \left( \pd^{\bullet}_{t} b + b \nabla_{\Gamma(t)} \cdot v \right) \right) \rd x
					\\
					&\leq \left( -C_{b} + \frac{\varepsilon}{2} \vecnorm{\nabla_{\Gamma} \cdot v}_{L^{\infty}} \right) \int_{\Gamma} d^{2} \rd x + \varepsilon \vecnorm{\pd^{\bullet}_{t} b + b \nabla_{\Gamma(t)} \cdot v}_{L^{\infty}} \int_{\Gamma} \abs{d} \rd x
					\\
					&\leq \left( -C_{b} + \frac{\varepsilon}{2} C_{V} \right) \int_{\Gamma} d^{2} \rd x + \varepsilon C_{D} \int_{\Gamma} \abs{d} \rd x
					\\
					&\leq \left( -C_{b} + \frac{\varepsilon}{2} C_{V} \right) \int_{\Gamma} d^{2} \rd x + \varepsilon C_{D} \sqrt{\abs{\Gamma}} \int_{\Gamma} \abs{d}^{2} \rd x
					\\
					&\leq \left( -C_{b} + \frac{\varepsilon}{2} \left( C_{V} + 2 C_{D} C_{L} \right) \right) \int_{\Gamma} d^{2} \rd x
				\end{align*}
				where we have used H\"older's inequality to derive $\int_{\Gamma} \abs{d} \rd x \leq \sqrt{\abs{\Gamma}} \int_{\Gamma} \abs{d}^{2} \rd x$. Now choose $\varepsilon \leq \frac{C_{b}}{C_{V} + 2 C_{D} C_{L}}$. We then obtain
				\begin{align} \label{eq:NormNeilson-bDeriv6}
					\frac{\rd}{\rd t} \int_{\Gamma(t)} \frac{\varepsilon \abs{d}^{2}}{2} \,\rd x \leq \frac{-C_{b}}{\varepsilon} \int_{\Gamma(t)} \frac{\varepsilon \abs{d}^{2}}{2} \rd x
				\end{align}
				By Gr\"onwall's inequality \cite[\S 1.3]{Verhulst90}, we get
				\begin{align} \label{eq:NormNeilson-bDeriv7}
					\int_{\Gamma(t)} \frac{\varepsilon d^{2}}{2} \,\rd x \leq \left( \int_{\Gamma(0)} \left. \frac{\varepsilon d^{2}}{2} \right\vert_{t = 0} \,\rd x \right) \exp \left( \frac{-C_{b} t}{\varepsilon} \right)
				\end{align}
				Cancel the $\frac{\varepsilon}{2}$ factor in \eqref{eq:NormNeilson-bDeriv7} and then let $\varepsilon \rightarrow 0$ to obtain the result.
			\end{proof}
			We are now ready to show that the reaction-diffusion system can be reduced.
			\begin{proof}[Proof of theorem~\ref{thm:ReducedSys}]
				Rewrite \eqref{eq:NormNeilson-b} into the following form
				\begin{align} \label{eq:NormNeilson-b1}
					\pd^{\bullet}_{t'} b' + b' \nabla_{\Gamma '(t')} \cdot v'
					&= \widetilde{D_{b}} \Delta_{\Gamma '(t')}b' - \widetilde{r_{b}} \left( b' - \fint_{\Gamma '(t')} a'\, \rd x' \right)
				\end{align}
				where $\widetilde{D_{b}} = \frac{D_{b}T}{L^{2}} \approx 5 \times 10^{6}$ and $\widetilde{r_{b}} = T r_{b} \approx 1.875 \times 10^{4}$. Upon dividing by $\widetilde{D_{b}}$, we see that the LHS of these equations becomes very small in comparison to the rest of the terms, and can thus be taken to be zero. We then obtain
				\begin{align} \label{eq:NormNeilson-b2}
					\Delta_{\Gamma '(t')} b' - \frac{\widetilde{r_{b}}}{\widetilde{D_{b}}}\left( b' - \fint_{\Gamma '(t')} a'\, \rd x' \right) = 0.
				\end{align}
				The reasoning behind this is given by Lemma~\ref{lemma:bEqnApprox} --- equations \eqref{eq:NormNeilson-b1} and \eqref{eq:NormNeilson-b2} can be rewritten in the form of equations \eqref{eq:NormNeilson-bDeriv1a} and \eqref{eq:NormNeilson-bDeriv1b} where $b'$ in \eqref{eq:NormNeilson-b1} is taken to be $b_{\varepsilon}$ in \eqref{eq:NormNeilson-bDeriv1a}, $C_{b} = \frac{\widetilde{r_{b}}}{\widetilde{D_{b}}}$, $\varepsilon = \frac{1}{\widetilde{D_{b}}}$ and $\overline{a} = \fint_{\Gamma (t)} a\, \rd x$.

				Now rewrite \eqref{eq:NormNeilson-b2} in variational form to get
				\begin{align} \label{eq:NormNeilson-bVar}
					\int_{\Gamma '(t')} \nabla_{\Gamma '(t')} b' \nabla_{\Gamma '(t')} v + \frac{\widetilde{r_{b}}}{\widetilde{D_{b}}} \int_{\Gamma '(t')} b' v = \frac{\widetilde{r_{b}}}{\widetilde{D_{b}}} \left( \fint_{\Gamma '(t')} a'\, \rd x' \right) \int_{\Gamma '(t')} v
				\end{align}
				with test function $v \in H^{1} (\Gamma ' (t')) =: V$. The LHS of \eqref{eq:NormNeilson-bVar} is a bounded and coercive bilinear form on $V$ and the RHS defines a bounded linear functional on $V$. We can now apply the Lax-Milgram theorem to see that it has a unique solution. Upon choosing
				\begin{align} \label{eq:NormNeilson-bApprox}
					b' = \fint_{\Gamma(t')} a'\, \rd x',
				\end{align}
				we see that it is the unique solution to \eqref{eq:NormNeilson-bVar}. Thus it is an approximation of the value of $b'$ under the assumptions of Lemma~\ref{lemma:bEqnApprox}.

				The approach for \eqref{eq:NormNeilson-c} is similar. Rewrite the equation as
				\begin{align} \label{eq:NormNeilson-c1}
					\pd^{\bullet}_{t'} c' + c' \nabla_{\Gamma '(t')} \cdot v' -  \widetilde{D_{c}} \Delta_{\Gamma '(t')}c' = - \widetilde{r_{c}} c' + \widetilde{b_{c}} a'
				\end{align}
				where $\widetilde{D_{c}} = \frac{D_{c}T}{L^{2}} \approx 7$, $\widetilde{r_{c}} = T r_{c} \approx 8125$ and $\widetilde{b_{c}} = T b_{c} \approx 3125$. Upon dividing by $\widetilde{r_{c}}$, we see that the LHS of these equations becomes very small in comparison to the rest of the terms, and can thus be taken to be zero --- although we do not state an explicit result, we hope that a modified form of Lemma~\ref{lemma:bEqnApprox} to include a bound on the diffusion term $\Delta_{\Gamma '(t')}c'$ can justify this approach. We now have
				\begin{align} \label{eq:NormNeilson-cApprox}
					c' = \frac{\widetilde{b_{c}}}{\widetilde{r_{c}}} a' \approx 0.385 a'.
				\end{align}
			\end{proof}
			We have therefore shown that the model does give rise to the behaviour observed in the simulation results of \cite{Neil10}, and that the system dynamics are contained in the dynamics of $a$.

			\subsubsection*{Variational Formulation} \label{subsubsec:VarForm}
				We now derive a weak formulation of \eqref{eq:NormNeilsonReduced-a}. For notational simplicity, we lose the $'$ and rewrite \eqref{eq:NormNeilsonReduced-a} into the more compact form
				\begin{align} \label{eq:NonVarForm}
					\pd^{\bullet}_{t'} a + a \nabla_{\Gamma(t)} \cdot v = D \Delta_{\Gamma(t)} a + f(a)
				\end{align}
				where $f(a) = T \left( \frac{s(\frac{a^{2}}{b} + \frac{b_{a}}{A})}{(s_{c} + Ac)(1 + A^{2}(a)^{2}s_{a})} - r_{a}a \right)$. Using the Transport Identity \eqref{eq:Transport}, we obtain the following variational formulation:
				\renewcommand{\theprob}{($\mathbf{P}^{\mathrm{a}}_{\mathrm{wk}}$)}
				\begin{prob} \label{prob:1}
					Find $a(\cdot,t) \in V = H^{1}(\mathcal{G}_{T})$ such that for almost every $t \in (0,T)$,
					\begin{align} \label{eq:VarForm1}
						\frac{d}{dt}\int_{\Gamma(t)}a \phi + D \int_{\Gamma(t)} \nabla_{\Gamma(t)} a \nabla_{\Gamma(t)} \phi = \int_{\Gamma(t)}a \dot{\phi} + \int_{\Gamma(t)}f(a) \phi,
					\end{align}
					for every $\phi(\cdot,t) \in V$, where $\mathcal{G}_{T} = \cup_{t \in [0, T]} (\Gamma(t) \times \{ t \})$ and $f(a) = T \left( \frac{s(\frac{a^{2}}{b} + \frac{b_{a}}{A})}{(s_{c} + Ac)(1 + A^{2}(a)^{2}s_{a})} - r_{a}a \right)$.
				\end{prob}

		\subsection{Modelling the Neutrophil Membrane Movement} \label{subsec:ModelMembraneMove}
			Recall the neutrophil membrane movement model of \cite{Neil10} described in Section~\ref{sec:Intro}. One of the problems with this model is finding values of $\lambda$ which have to satisfy \eqref{eq:MembraneMoveLambdaODE}, a non-linear ODE. Therefore, we propose an alternative model that eliminates the $\lambda \kappa$ term and replaces the formula for $V_{f}$ with a mean curvature flow model, given by
			\begin{align} \label{eq:MembraneMoveNew}
				V_{f}(x) = -\varepsilon H(x) + \delta a(x) + \bar{\lambda}.
			\end{align}
			where $\epsilon, \delta$ are small, positive constants, $\bar{\lambda}$ is a Lagrange multiplier which constrains the area of the cell to remain constant and $H(x)$ is the mean curvature at point $x$, where
			\begin{align*}
				H = \nabla_{\Gamma(t)} \cdot \nu(x), \; x \in \Gamma(t)
			\end{align*}
			with the outward normal $\nu$ oriented as shown in Figure~\ref{fig:Orientation}.
			\begin{figure}[h!]
				\centering
				\includegraphics[scale=0.75]{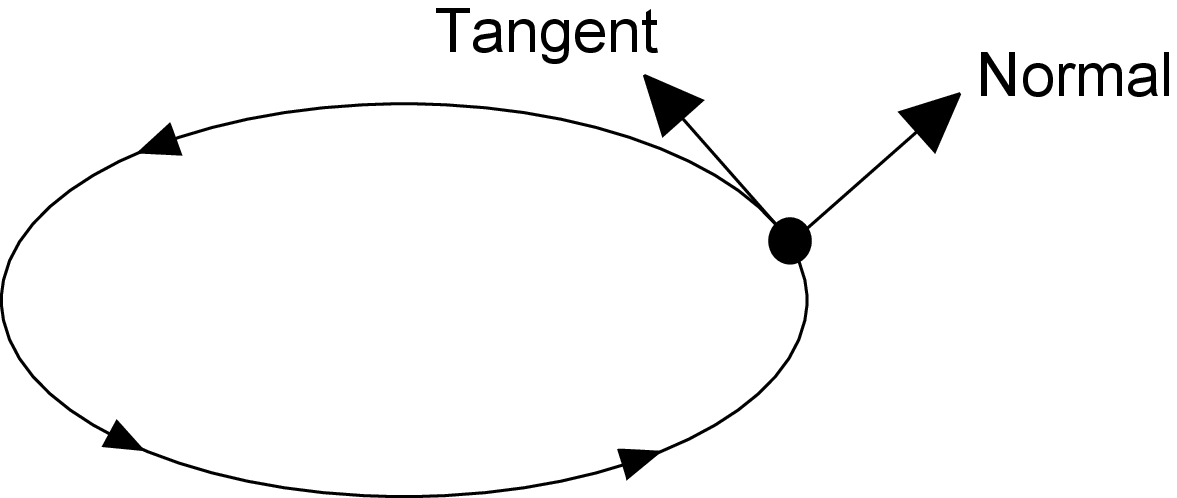}
				\caption{Orientation of curve parametrisation and normal vector.}
				\label{fig:Orientation}
			\end{figure}
			To ensure that the results obtained using this model can be compared to those of \cite{Neil10}, we choose $\delta$ to be of order $K_{prot} \frac{TA}{L}$, in keeping with the normalisation of the reaction-diffusion system. The $\varepsilon H$ term also replaces the cortical tension control, represented by the $\lambda \kappa$ term, in the original mode. We can derive an explicit formula for $\bar{\lambda}$.
			\begin{lemma}[Formula for $\bar{\lambda}$] \label{lemma:LambdaBar}
				\begin{align} \label{eq:LambdaBar}
					\bar{\lambda}(t) = \varepsilon \frac{2 \pi}{\abs{\Gamma(t)}} - \delta b
				\end{align}
				where $b$ is as given in equation \eqref{eq:NormNeilsonReduced-b} without primes.
			\end{lemma}
			\begin{proof}
				Since we want the area to remain constant over time, observe that
				\begin{align} \label{eq:LambdaBarDeriv1}
					0 = \frac{\rd}{\rd t} \abs{\Omega(t)} = \frac{\rd}{\rd t} \int_{\Omega(t)} 1 = \int_{\Gamma(t)} V_{f} = \int_{\Gamma(t)} -\varepsilon H(x) + \delta a(x) + \bar{\lambda}.
				\end{align}
				Thus $\bar{\lambda} \abs{\Gamma(t)} = \varepsilon \int_{\Gamma(t)} H(x) + \delta \int_{\Gamma(t)} a(x)$. However, the Gauss-Bonnet theorem states that $\int_{\Gamma(t)} H(x) = 2 \pi$ when $\Gamma(t)$ is a closed curve which is not self-intersecting. Substituting this and \eqref{eq:NormNeilsonReduced-b} and rearranging the equation gives the desired result.
			\end{proof}
			Now let $\mathbf{X} \in C^{2}\left( \R \times [0,T], \R^{2} \right)$ be a parametrisation of $\Gamma(t)$. Since $\mathbf{X}_{t} = V_{f} \nu$, we get
			\begin{align*}
				\mathbf{X}_{t} = -\varepsilon H(x) \nu + \delta a(x) \nu + \bar{\lambda} \nu.
			\end{align*}
			Note that $\nu = \left( \frac{\mathbf{X}_{p}}{\abs{\mathbf{X}_{p}}} \right)^{\perp} = \frac{\mathbf{X}_{p}^{\perp}}{\abs{\mathbf{X}_{p}}}$. Furthermore, using equation 2.15 of \cite{Elli05}, we can deduce that
			\begin{align*}
				H \nu = \frac{-1}{\abs{\mathbf{X}_{p}}} \frac{\pd}{\pd p} \left( \frac{\mathbf{X}_{p}}{\abs{\mathbf{X}_{p}}} \right).
			\end{align*}
			We now have
			\begin{subequations} \label{eq:MembraneMoveStrong}
				\begin{align}
					\label{eq:MembraneMoveStrong1} \mathbf{X}_{t} \abs{\mathbf{X}_{p}}
					&= \varepsilon \frac{\pd}{\pd p} \left( \frac{\mathbf{X}_{p}}{\abs{\mathbf{X}_{p}}} \right) + (\delta a + \bar{\lambda}) \mathbf{X}_{p}^{\perp} \quad \mbox{in} \, [0,1] \, \times (0,T)
					\\
					\label{eq:MembraneMoveStrong2} \mathbf{X}(\cdot,0)
					&= \mathbf{X}_{0} \quad \mbox{in} \; [0,1].
				\end{align}
			\end{subequations}
			We also require that $\mathbf{X}$ satisfies the periodicity condition
			\begin{align} \label{eq:MembraneMovePeriod}
				\mathbf{X}(p,t) = \mathbf{X}(p + 1,t), \; p \in \R, \; t \in [0,T].
			\end{align}
			Suppose that $\mathbf{X} : \R \times [0,T] \rightarrow \R^{2}$ is a smooth solution of \eqref{eq:MembraneMoveStrong} and \eqref{eq:MembraneMovePeriod} --- in particular, $\abs{\mathbf{X}_{p}} > 0$ in $[0,1] \times [0,T]$. Taking the dot product  with a test function $\phi \in H^{1}_{per}([0,1]; \R^{2}) := \{ \phi \in H^{1}([0,1]; \R^{2}) \vert \phi(0) = \phi(1) \}$ and integrating with respect to $p$ yields the following parametric variational form for \eqref{eq:MembraneMoveStrong}.
			\renewcommand{\theprob}{($\mathbf{P}^{\mathrm{m}}_{\mathrm{wk}}$)}
			\begin{prob} \label{prob:2}
				Given $a \in H^{1}_{per}([0,1] \times [0,T]; \R)$, find $\mathbf{X} \in H^{1}_{per}([0,1] \times [0,T]; \R^{2})$ such that
				\begin{align} \label{eq:MembraneMoveWeakForm}
					\int_{0}^{1} \left[ \mathbf{X}_{t} \cdot \phi \right] \abs{\mathbf{X}_{p}} + \frac{\varepsilon \mathbf{X}_{p} \cdot \phi_{p}}{\abs{\mathbf{X}_{p}}} \, \rd p = \int_{0}^{1} \left( \delta a + \bar{\lambda} \right) \phi \cdot \mathbf{X}_{p}^{\perp} \, \rd p
				\end{align}
				subject to the area of the cell remaining constant, for all $\phi \in H^{1}_{per}([0,1]; \R^{2})$.
			\end{prob}

	\section{Model Numerics} \label{sec:Numerics}
		We now develop a finite element method to numerically solve the reduced system of reaction-diffusion equations and simulate the movement of the neutrophil using variational problems~\ref{prob:1} and~\ref{prob:2}. The full non-linear, stochastic system will then be solved to simulate the movement of the neutrophil \emph{in the absence of chemoattractant} (i.e. $R_{0} = 0$).
    	\subsection{Reaction-Diffusion PDE} \label{subsec:ReactDiffNumerics}
			\subsubsection{Finite Element Approximation} \label{subsubsec:RDFEM}
				The smooth, evolving surface $\Gamma(t)$ with $\pd \Gamma(t)= \emptyset$ is approximated by an evolving surface $\Gamma_{h}(t)$ with  $\pd \Gamma_{h}(t)= \emptyset$. $\Gamma_{h}(t)$ is a polyhedral surface whose vertices $\{\mathbf{X}_{j}(t)\}_{j = 1}^{N}$ are taken to sit on $\Gamma(t)$ so that $\Gamma_{h}(t)$ is an interpolation. Let $\mathbf{X}^{h}: \R \times [0,T] \rightarrow \mathbb{R}^{2}$ be a smooth parametrisation of $\Gamma_{h}(t)$ with $\abs{\mathbf{X}^{h}_{p}} > 0$, and periodicity condition $\mathbf{X}^{h}(p,t) = \mathbf{X}^{h}(p + 1,t), \, 0 < t \leq T, \, \forall p \in \R$. Then for $F : [0,1] \times [0,T] \rightarrow \mathbb{R}$ we have
				\begin{align*}
					\nabla_{\Gamma_{h}(t)} F(p,t) = \frac{F_{p}(p,t)}{\abs{\mathbf{X}^{h}_{p}(p,t)}} \frac{\mathbf{X}^{h}_{p}(p,t)}{\abs{\mathbf{X}^{h}_{p}(p,t)}}.
				\end{align*}
				This parametrisation will allow us to replace the surface integrals by integrals over the unit interval, thus reducing our numerical analysis to a 1D finite element method.
				Let $p_{j} = jh$, with $j = 0, \dots, N$, be a uniform grid with grid size $h = 1/N$ and define the finite element space
				\begin{align} \label{eq:FinEltSpace}
					V_{h} = \{ \phi \in C^{0}([0,1]; \R) \left\vert \phi \vert_{[p_{j - 1}, p_{j}]} \in P_{1}, j = 1, \dots, N; \phi(0) = \phi(1) \right. \}
				\end{align}    
				to be the space of piecewise linear, continuous functions. We can now define the semi-discrete problem:
				\renewcommand{\theprob}{($\mathbf{P}^{\mathrm{a}}_{\mathrm{h}}$)}
				\begin{prob} \label{prob:3}
					Find $a^{h}(\cdot,t) \in V_{h}$ such that for almost every $t \in (0,T)$,
					\begin{align} \label{eq:VarForm2}
						\frac{\rd}{\rd t} \int_{0}^{1} a^{h} \phi \abs{\mathbf{X}^{h}_{p}} \rd p + D \int_{0}^{1} \frac{a^{h}_{p} \phi_{p}}{\abs{\mathbf{X}^{h}_{p}}} \rd p = \int_{0}^{1} f(a^{h}) \phi \abs{\mathbf{X}^{h}_{p}} \rd p,
					\end{align}
					for every $\phi(\cdot,t) \in V_{h}$.
				\end{prob}
				Now, denoting the nodal basis functions by $\{\phi_{j}\}_{j = 1}^{N}$, and letting
				\begin{align} \label{eq:ahExpression}
					a^{h}(p,t) := a^{h}(\mathbf{X}^{h}(p,t),t) = \sum_{j = 1}^{N} A_{j}(t) \phi_{j}(p) \in V_{h} \subset V
				\end{align}
				where $\dim(V_{h}) = N < \infty$, we can write the finite element approximation of \eqref{eq:VarForm1} as follows.
				\begin{align*}
					\frac{\rd}{\rd t} \sum_{j = 1}^{N} A_{j} \int_{0}^{1} \phi_{j} \phi_{i} \abs{\mathbf{X}_{p}^{h}} \, \rd p
					&+ D \sum_{j = 1}^{N} A_{j} \int_{0}^{1} \frac{\phi_{j,p} \phi_{i,p}}{\abs{\mathbf{X}_{p}^{h}}} \, \rd p
					\\
					&= \int_{0}^{1} f(a^{h}) \phi_{i} \abs{\mathbf{X}_{p}^{h}} \rd p \, , \, i = 1, \dots, N.
				\end{align*}
				 We thus end up with the following system of nonlinear ODEs,
				\begin{align} \label{eq:ODESystem}
					\frac{d}{dt}\Big(\mathbf{M}(t)\mathbf{a}\Big) + D\mathbf{S}(t)\mathbf{a} = \mathbf{M}(t) \mathbf{F}(a^{h},t)
				\end{align}
				where
				\begin{itemize}
					\item $\mathbf{M}$ is the time-dependent mass matrix given by
					\begin{align*}
						\mathbf{M}_{i,j}(t) = \int_{0}^{1} \phi_{i} \phi_{j} \abs{\mathbf{X}_{p}^{h}} \, \rd p \, , \, i,j = 1,...,N.
					\end{align*}
					\item $\mathbf{S}$ is the time-dependent stiffness matrix
					\begin{align}
						\mathbf{S}_{i,j}(t) = \int_{0}^{1} \frac{\phi_{i,p} \phi_{j,p}}{|\mathbf{X}_{p}^{h}|} \, \rd p \, , \, i,j = 1,...,N.
					\end{align}
					\item $\mathbf{a}$ is the time-dependent solution vector
					\begin{align}
						\mathbf{a}(t) = (A_{1}(t),...,A_{N}(t)),
					\end{align}
					\item $\mathbf{F}(t)$ is the time-dependent right-hand side vector
					\begin{align}
						\mathbf{F}_{i}(t) = \int_{0}^{1}f(a^{h})\phi_{i} \abs{\mathbf{X}_{p}^{h}} \, \rd p \, , \, i = 1,...,N.
					\end{align}
				\end{itemize}
				We can solve \eqref{eq:ODESystem} numerically by discretising it in time via the implicit Euler method, as done in \cite{DzEl07}. Let $t_{m} = m \Delta t$. Then we get
				\begin{align}
					\label{eq:RDTimeDisc1} \frac{\mathbf{M}^{m + 1}\mathbf{a}^{m + 1} - \mathbf{M}^{m}\mathbf{a}^{m}}{\Delta t} + \mathbf{S}^{m + 1}\mathbf{a}^{m + 1}
					&= \mathbf{M}^{m} \mathbf{F}^{m}, \, m = 0, \dots, M
					\\
					\label{eq:RDTimeDisc2} \Rightarrow \left( \mathbf{M}^{m+1} + \Delta t \mathbf{S}^{m + 1} \right) a^{m + 1}
					&= \mathbf{M}^{m} \left( \Delta t \mathbf{F}^{m} + \mathbf{a}^{m} \right)
				\end{align}
				where $\mathbf{A}^{m} := \mathbf{A}(t_{m})$. Note that we take $\mathbf{F}$ explicitly in time.

			\subsubsection{Numerical tests} \label{subsubsec:RDNumTests}
				Before we attempt to solve \eqref{eq:NonVarForm} numerically, we test this method in a simple setting ; an oscillating ellipse of the form
				\begin{align}
					\Gamma(t) := \{ \mathbf{x} = (x_{1},x_{2})^{T} \in \mathbb{R}^{2} \, \vert \, \frac{x_{1}^2}{\Gamma(t)} + x_{2}^{2} = R^{2} \}.
				\end{align}
				We consider examples from \cite{EliStin10}.
				\begin{eg}
					On the stationary unit circle $S^{2}$, the function $a_{S}(x,t):=e^{-4t}x_{1}x_{2}$ is a solution to the surface heat equation $\pd_{t}a_{S} - \Delta_{S^{2}}a_{S} = 0$. Parametrising the stationary unit circle and deriving a variational formulation in terms of this parametrisation as was done in the previous section, we solve the resulting problem numerically using the P1 finite element method over the domain $\Omega = [-1,1]$, choosing the final time to be $T = 1$. Figure~\ref{fig:ModelProb1} shows the numerical solution for a spatial discretisation $h=0.1$ and timestep $\Delta t = 0.0102$. Note that for this model problem there is no restriction on our choice of timestep as the scheme is unconditionally stable. Stability and convergence analysis of this scheme can be found in \cite{Thomee06}.
					\begin{figure}[h!]
						\centering
						\includegraphics[height=90mm]{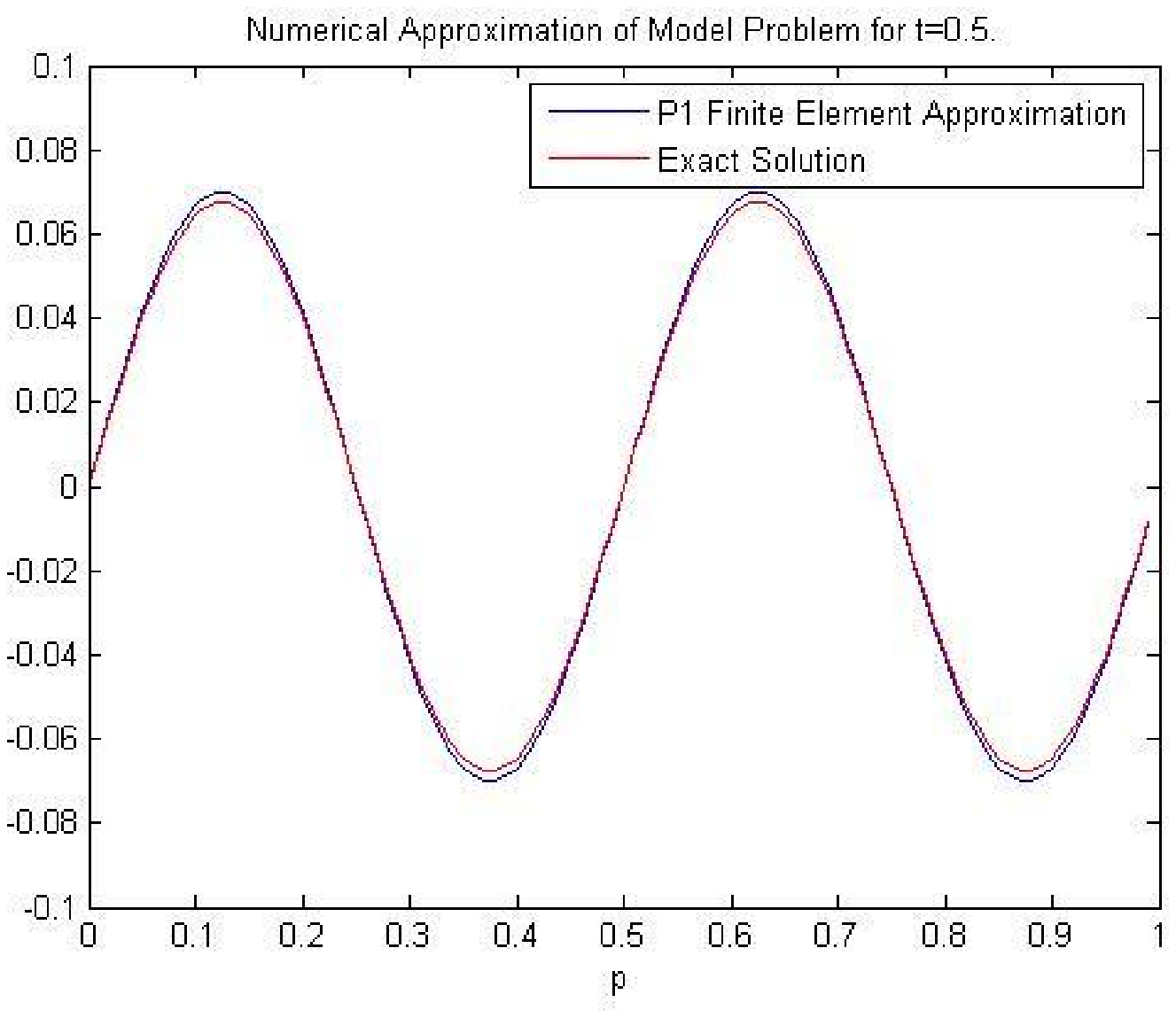}
						\caption{Finite element approximation and exact solution for Model Problem.}
						\label{fig:ModelProb1}
					\end{figure}
					Figure~\ref{fig:ModelProb2} shows the plot of the $L^{2}$ error between the exact solution and its finite element approximation at the final timestep $t = T$ for $h = 10^{-2}$ and $\Delta t = 10^{-4}$.
					\begin{figure}[h!]
						\centering
						\includegraphics[height=90mm]{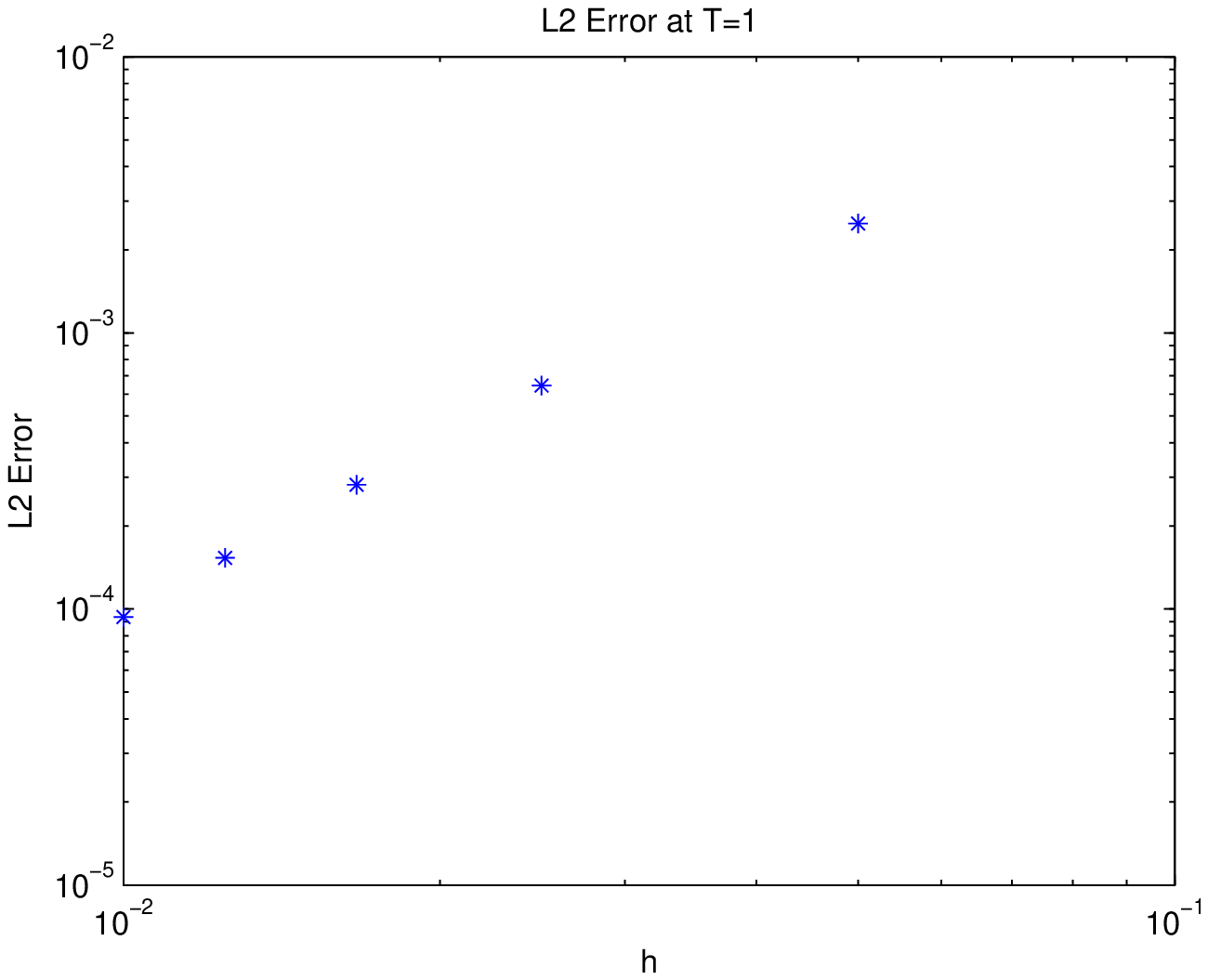}
						\caption{$L^{2}$ error at $t=T$ for Model Problem.}
						\label{fig:ModelProb2}
					\end{figure}
					From \cite{Thomee06}, the implicit Euler finite element approximation of the model problem satisfies the \emph{a priori} error estimate
					\begin{align}
						\max_{1\leq m\leq M} \vecnorm{e_{h}^{m}}_{L^{2}(\Omega)} \leq C(h^{2} + \Delta t).
					\end{align}
					By choosing $\Delta t = O(h^{2})$, as we have done by our choice of discretisation, we expect quadratic convergence in $h$ and this is confirmed by the plot in Figure~\ref{fig:ModelProb2}.
				\end{eg}
				\begin{eg}
					As another test example, we consider an expanding circle centred at the origin with radius $r(t) = 0.75 + 5t$. The velocity field is given by $v(x,t) = \frac{5x}{x}$ and, hence, is purely in the normal direction. The function $a_{E}(x,t):= \exp(\frac{4}{5r(t)})\frac{x_{1}x_{2}}{r(t)|x|^2}$ is a solution to $\pd_{t}a_{E} + v \cdot \nabla a_{E} - \Delta_{\Gamma}a_{E} = 0$.
As for the previous test example, the $L^{2}$ error at the final timestep converges quadratically in $h$.   
				\end{eg}

			\subsubsection{Results} \label{subsubsec:RDResults}
				We now test this numerical method on \eqref{eq:NormNeilson-a} by including the source term
				\begin{align} \label{eq:NumericsSourceTerm}
					f(a) = T \left( \frac{s(\frac{a^{2}}{b} + \frac{b_{a}}{A})}{(s_{c} + Ac)(1 + A^{2}(a)^{2}s_{a})} - r_{a}a \right).
				\end{align} 
				As mentioned earlier in Section~\ref{subsubsec:RDFEM}, we deal with this nonlinear term by taking it explicitly in time. The approximations of $b$ and $c$ (respectively \eqref{eq:NormNeilsonReduced-b} and \eqref{eq:NormNeilsonReduced-c}) derived in Section~\ref{sec:Model} as well as expression \eqref{eq:StochasticTerm} for the stochastic component $s(x,t)$ have been used to compute this source term. The chemical constants present in the source term are taken from Table~\ref{tab:EqnConstants}. Given that we have not yet discussed numerical methods for the curve evolution, we will restrict ourselves to solving this PDE over the stationary circle with radius $R = 1$.
				\begin{figure}[h!]
					\centering
					\includegraphics[height=90mm]{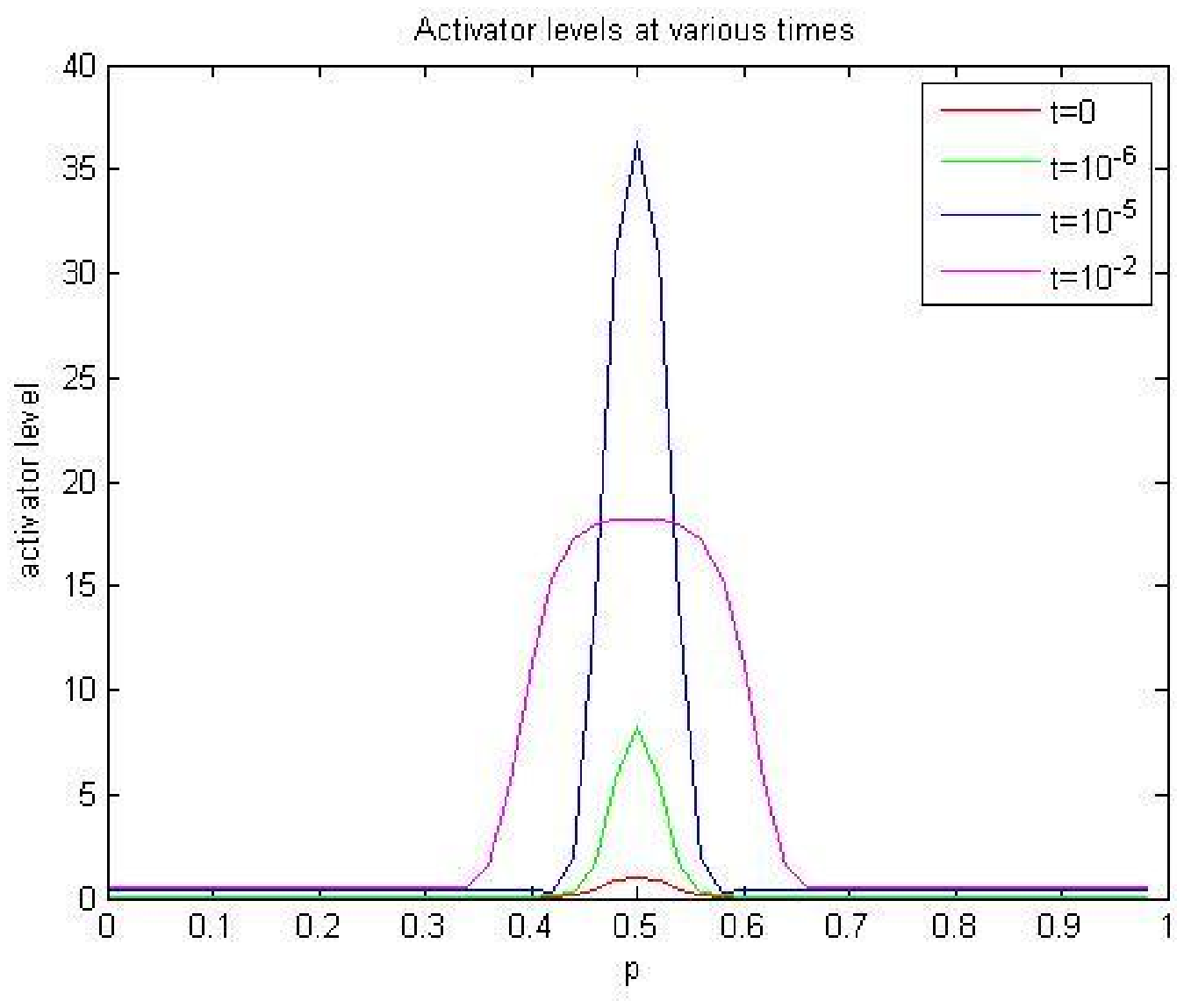}
					\caption{Local activator levels in the simulation of the reaction-diffusion equations}
					\label{fig:RDActivatorResults}
				\end{figure}

				Figure~\ref{fig:RDActivatorResults} plots the activator levels at different (rescaled) times, where we have chosen $h = 10^{-3}$ and $\Delta t = 400 h^{2}$ with final time $T = 10$. The initial condition for $a$ is chosen to be the `bump' shaped function $\exp \left( -\frac{(p - 0.5)^{2}}{0.002} \right)$. Now, given that the first term in \eqref{eq:NumericsSourceTerm} is of order $O \left( \frac{1}{a} \right)$ whilst the second term is of order $O(a)$, we expect that with our choice of initial condition the activator level would increase rapidly due to the effect of the first term. The second term has the opposite effect when $a$ is large enough, by countering the rapid growth induced by the first term. This is exactly the sort of behaviour we observe in Figure~\ref{fig:RDActivatorResults}.

		\subsection{Neutrophil Membrane Movement PDE} \label{subsec:MembraneMoveNumerics}
			\subsubsection{Finite Element Approximation} \label{subsubsec:MMFEM}
				We shall now apply a finite element method to solve \eqref{eq:MembraneMoveWeakForm} numerically. Once this is done, this will be coupled with the reaction-diffusion PDE discussed previously to simulate the entire system. Define the finite element space $\mathbf{V}_{h}$ in the same way as $V_{h}$ in \eqref{eq:FinEltSpace} but with values now in $\R^{2}$ rather than $\R$. Motivated by the approach of \cite[\S 4]{Elli05}, we derive the corresponding semi-discrete problem satisfied by $\mathbf{X}^{h}$. 
				\renewcommand{\theprob}{($\mathbf{P}^{\mathrm{m}}_{\mathrm{h}}$)}
				\begin{prob} \label{prob:4}
					Given $a^{h}(\cdot,t) \in V_{h}$, find $\mathbf{X}^{h} \in \mathbf{V}_{h}$ such that for almost every $t \in (0,T)$,
					\begin{align*}
						\int_{0}^{1} \left[ \mathbf{X}^{h}_{t} \cdot \varphi \right] \abs{\mathbf{X}^{h}_{p}} + \frac{\varepsilon \mathbf{X}^{h}_{p} \cdot \varphi_{p}}{\abs{\mathbf{X}^{h}_{p}}} \, \rd p = \int_{0}^{1} \left( \delta a^{h} + \bar{\lambda}^{h} \right) \varphi \cdot \left( \mathbf{X}^{h}_{p} \right)^{\perp} \, \rd p
					\end{align*}
					for every $\varphi(\cdot,t) \in \mathbf{V}_{h}$, subject to the area of the cell remaining constant, where $\bar{\lambda}^{h}$ is a discretised form of \eqref{eq:LambdaBar}.
				\end{prob}
				As before let $\{\phi_{j}\}_{j = 1}^{N}$  be the standard, piecewise linear, scalar nodal basis. Set 
                \begin{align*}				
				\mathbf{X}^{h}(p,t) = \sum_{j = 1}^{N} \mathbf{X}_{j}(t) \phi_{j}(p) 
				\end{align*}
				where $\mathbf{X}_{j}(t) \in \R^{2}$ and let $\varphi = \phi_{j}\mathbf{e}_{k}$, $k = 1, 2$, $j = 1, \dots, N$. Since $\sum_{j = 1}^{N} \phi_{j} = 1$,
				\begin{align*}
					\delta a^{h} + \bar{\lambda}^{h} = \sum_{j = 1}^{N} \left( \delta a^{h} + \bar{\lambda}^{h} \right) \phi_{j}\mathbf{e}_{k}, \, k = 1, 2.
				\end{align*}
				Hence, for $i = 1, \dots, N$,
				\begin{align*}
				\int_{0}^{1} \left( \left( \delta a^{h} + \bar{\lambda}^{h} \right) \phi_{i} \left( \mathbf{X}^{h}_{p} \right)^{\perp}\right) \cdot \mathbf{e}_{k} \, \rd p
				= \sum_{j = 1}^{N} \left( \delta a^{h} + \bar{\lambda}^{h} \right) \int_{0}^{1} \left( \phi_{i} \phi_{j} \left( \mathbf{X}^{h}_{p} \right)^{\perp}\right) \cdot \mathbf{e}_{k} \, \rd p, \, k = 1, 2.
				\end{align*}
				We then get
				\begin{align*}
					\sum_{j = 1}^{N} \int_{0}^{1} \left( \frac{\rd}{\rd t} \left[ \mathbf{X}_{j} \right] \phi_{j} \phi_{i} \abs{\mathbf{X}^{h}_{p}}\right) \cdot \mathbf{e}_{k} \, \rd p
					&+ \sum_{j = 1}^{N} \mathbf{X}_{j} \int_{0}^{1} \left( \frac{\varepsilon \phi_{j,p} \phi_{i,p}}{\abs{\mathbf{X}^{h}_{p}}} \right) \cdot \mathbf{e}_{k} \, \rd p
					\\
					&= \sum_{j = 1}^{N} \left( \delta a^{h} + \bar{\lambda}^{h} \right) \int_{0}^{1} \left( \phi_{i} \phi_{j} \left( \mathbf{X}^{h}_{p} \right)^{\perp}\right) \cdot \mathbf{e}_{k} \, \rd p,\ k = 1, 2.
				\end{align*}
				As before let $t_{m} = m \Delta t$, $m=0,\dots,M$. The following time discretisation is suggested to derive a numerical scheme for \eqref{eq:MembraneMoveWeakForm}: 
				\begin{align*}
					\sum_{j = 1}^{N} \left( \int_{0}^{1} \left( \frac{\left[ \mathbf{X}_{j}^{m + 1} - \mathbf{X}_{j}^{m} \right]}{\Delta t} \phi_{i} \phi_{j} \abs{\mathbf{X}^{h,m}_{p}} \right) \cdot \mathbf{e}_{k} \, \rd p + \int_{0}^{1} \left( \frac{\varepsilon \mathbf{X}_{j}^{m + 1}}{\abs{\mathbf{X}^{h,m}_{p}}} \left( \phi_{i} \right)_{p} \left( \phi_{j} \right)_{p} \right) \cdot \mathbf{e}_{k} \, \rd p \right)
					\\
					= \sum_{j = 1}^{N} \left( \delta a^{m} + \bar{\lambda}^{m} \right) \int_{0}^{1} \left( \phi_{i} \phi_{j} \left( \mathbf{X}_{p}^{h, m} \right)^{\perp} \right) \cdot \mathbf{e}_{k} \, \rd p, \, k = 1, 2.
				\end{align*}
				Note that we take the Lagrangian $\bar{\lambda}^{h}$ explicitly in time. Now set $\mathbf{X}_{j}(t_{m}) = \left( x_{j}(t_{m}), y_{j}(t_{m}) \right)$, and define, for $m = 1, \dots, M$, matrices $\widetilde{\mathbf{M}}_{\mathbf{x}}$ and $\widetilde{\mathbf{M}}_{\mathbf{y}}$ by
				\begin{align*}
					\left( \widetilde{\mathbf{M}}_{\mathbf{x}} \right)_{i,j}(t_{m})
					&= \int_{0}^{1} \phi_{i} \phi_{j} \left( \mathbf{X}_{p}^{h, m} \right)^{\perp} \cdot \mathbf{e}_{1} \, \rd p,
					\\
					\left( \widetilde{\mathbf{M}}_{\mathbf{y}} \right)_{i,j}(t_{m})
					&= \int_{0}^{1} \phi_{i} \phi_{j} \left( \mathbf{X}_{p}^{h, m} \right)^{\perp} \cdot \mathbf{e}_{2} \, \rd p,
				\end{align*}
				where $\mathbf{x}^{m} = (x_{1}(t_{m}), \dots, x_{N}(t_{m}))^{\mathrm{T}}$ and $\mathbf{y}^{m} = (y_{1}(t_{m}), \dots, y_{N}(t_{m}))^{\mathrm{T}}$. We thus obtain two decoupled systems of equations --- one for each component of $\mathbf{X}$.
				\begin{align}
					\label{eq:DiscreteCurveFEM1} \mathbf{M}^{m} \mathbf{x}^{m + 1} + \varepsilon \Delta t \mathbf{S}^{m} \mathbf{x}^{m + 1} = \mathbf{M}^{m} \mathbf{x}^{m} + \Delta t \widetilde{\mathbf{M}}_{x}^{m} \left( \delta a^{m} + \bar{\lambda}^{m} \mathbf{1} \right)
					\\
					\label{eq:DiscreteCurveFEM2} \mathbf{M}^{m} \mathbf{y}^{m + 1} + \varepsilon \Delta t \mathbf{S}^{m} \mathbf{y}^{m + 1} = \mathbf{M}^{m} \mathbf{y}^{m} + \Delta t \widetilde{\mathbf{M}}_{y}^{m} \left( \delta a^{m} + \bar{\lambda}^{m} \mathbf{1} \right)
				\end{align}
				where we denote $\mathbf{A}^{m}$ to mean $\mathbf{A}(t_{m})$ and $\mathbf{1} = (1, \dots, 1)^{\mathrm{T}}$.

			\subsubsection{Numerical tests} \label{subsubsec:MMNumTests}
				As was done for the reaction-diffusion PDE, we test this numerical method in a simpler setting before attempting to solve the original problem. 
\begin{eg} Consider a ball in $\mathbb{R}^2$ of radius $R(t)$ in a supercooled liquid of temperature $U < 0$ satisfying the following equation for its surface,
				\begin{align*}
					V = -\varepsilon H - U
				\end{align*}
				which is a \emph{Gibbs-Thomson} law (surface tension $-\varepsilon H$ with $\varepsilon>0$, and balancing temperature U) modified by kinetic undercooling. Then it can be shown that a ball with initial radius $R_{I}$ shrinks or grows depending on whether depending on whether its initial radius is respectively bigger or smaller than the critical radius $R_{c} = -\frac{\varepsilon}{U}$. Deriving a numerical method for the Gibbs-Thomson law is done in a very similar way to our original curve evolution model.
				\begin{figure}[h!]
					\centering
					\includegraphics[height=90mm]{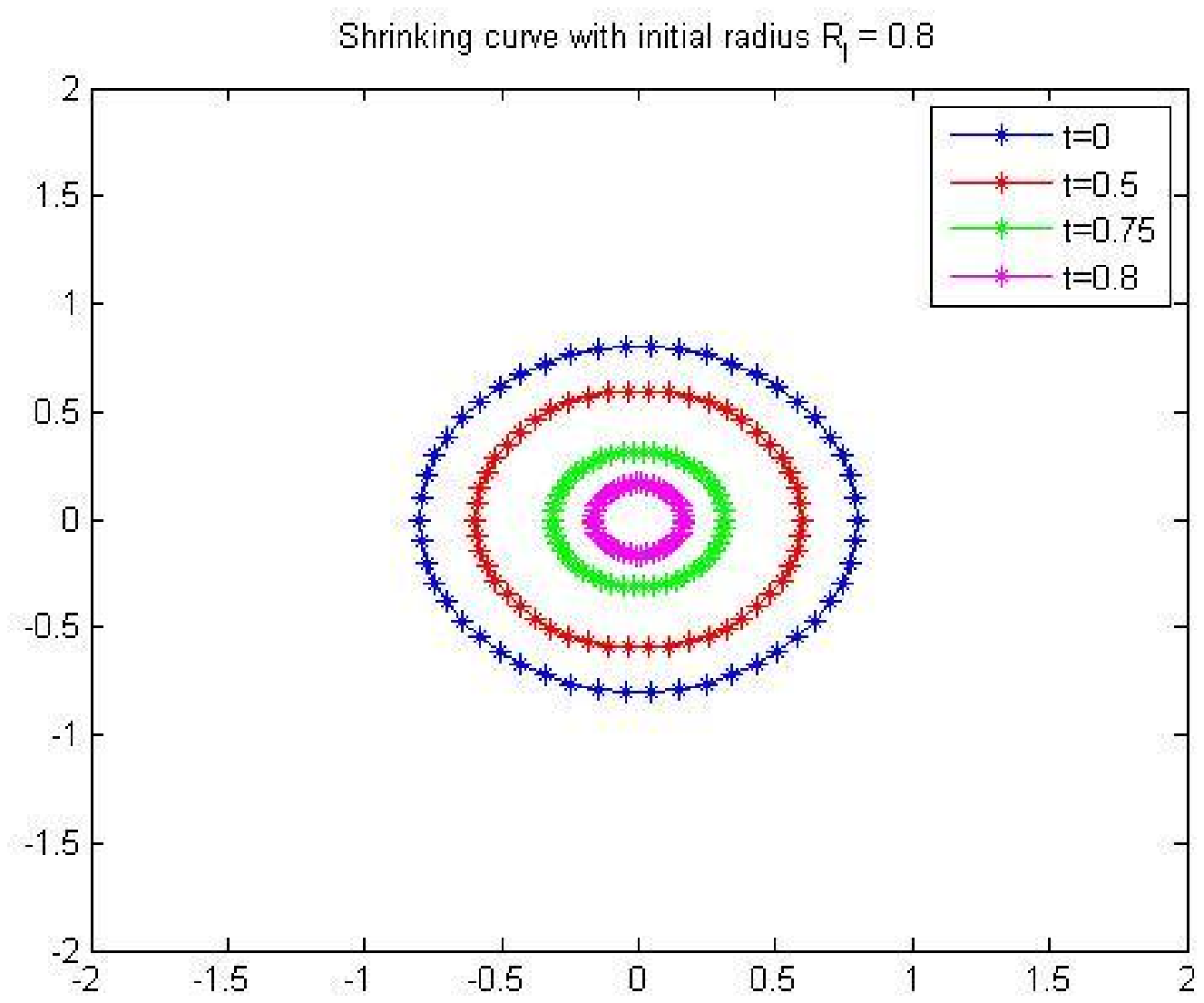}
					\caption{Curve evolution under Gibbs-Thomson law for $R_{I} = 0.8$.}
					\label{fig:curveShrink}
				\end{figure}
				\begin{figure}[h!]
					\centering
					\includegraphics[height=90mm]{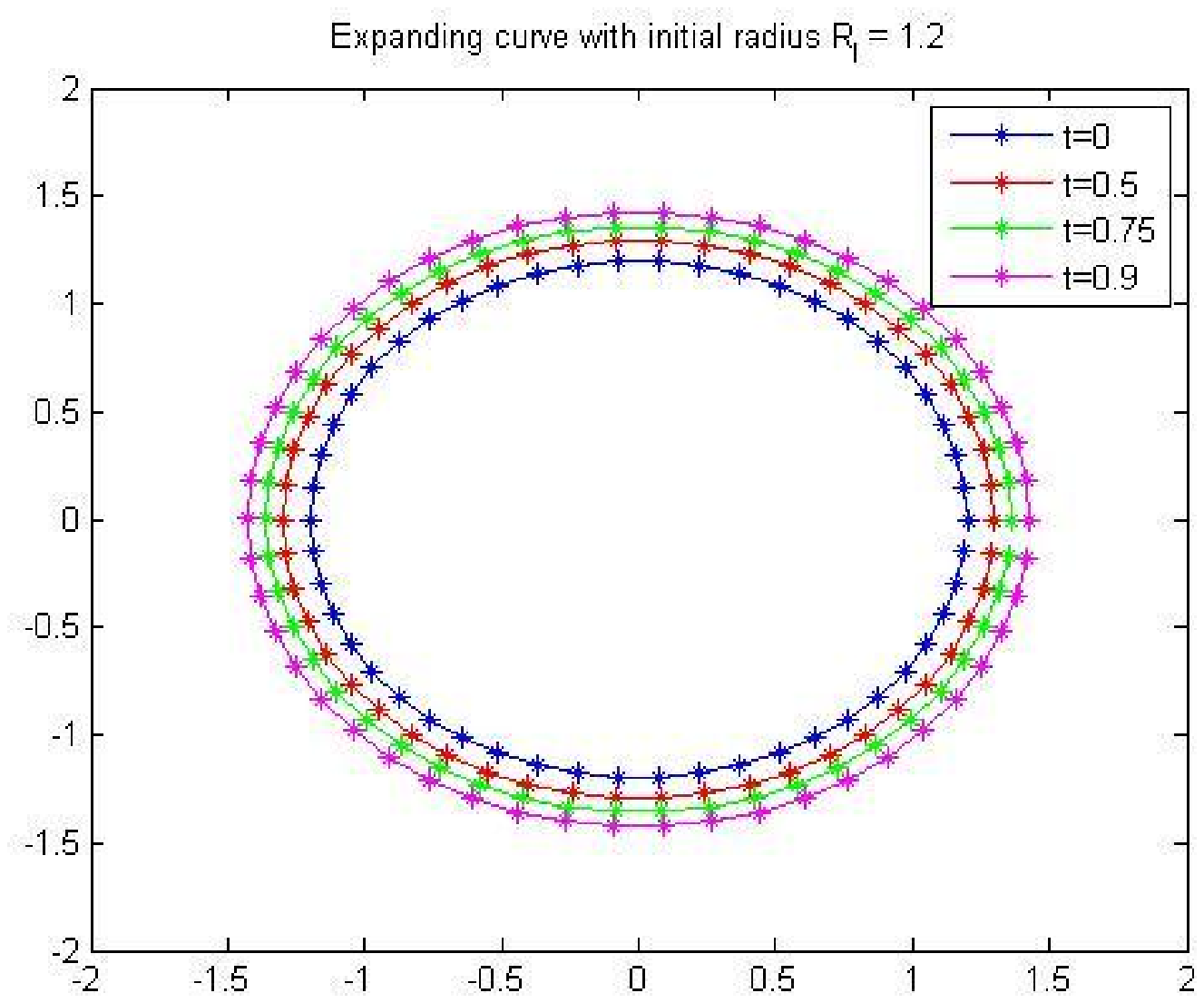}
					\caption{Curve evolution under Gibbs-Thomson law for $R_{I} = 1.2$.}
					\label{fig:curveExpand}
				\end{figure}

				Figure~\ref{fig:curveShrink} and Figure~\ref{fig:curveExpand} show the curve evolution under the Gibbs-Thomson law for $\epsilon = 1$ and $U=-1$ (and so $R_{c} = 1$) with initial radii $R_{I} = 0.8$ and $R_{I} = 1.2$ respectively. As expected, the curve shrinks for the former and expands for the latter.
				\end{eg}

			\subsubsection{Results} \label{subsubsec:MMResults}
				We now test this numerical method on \eqref{eq:MembraneMoveNew} with $a = 0$. The full system, including the coupling of the curve evolution with the reaction-diffusion PDE, will be dealt with in the next section. 
				\begin{figure}[h!]
					\centering
					\includegraphics[height=90mm]{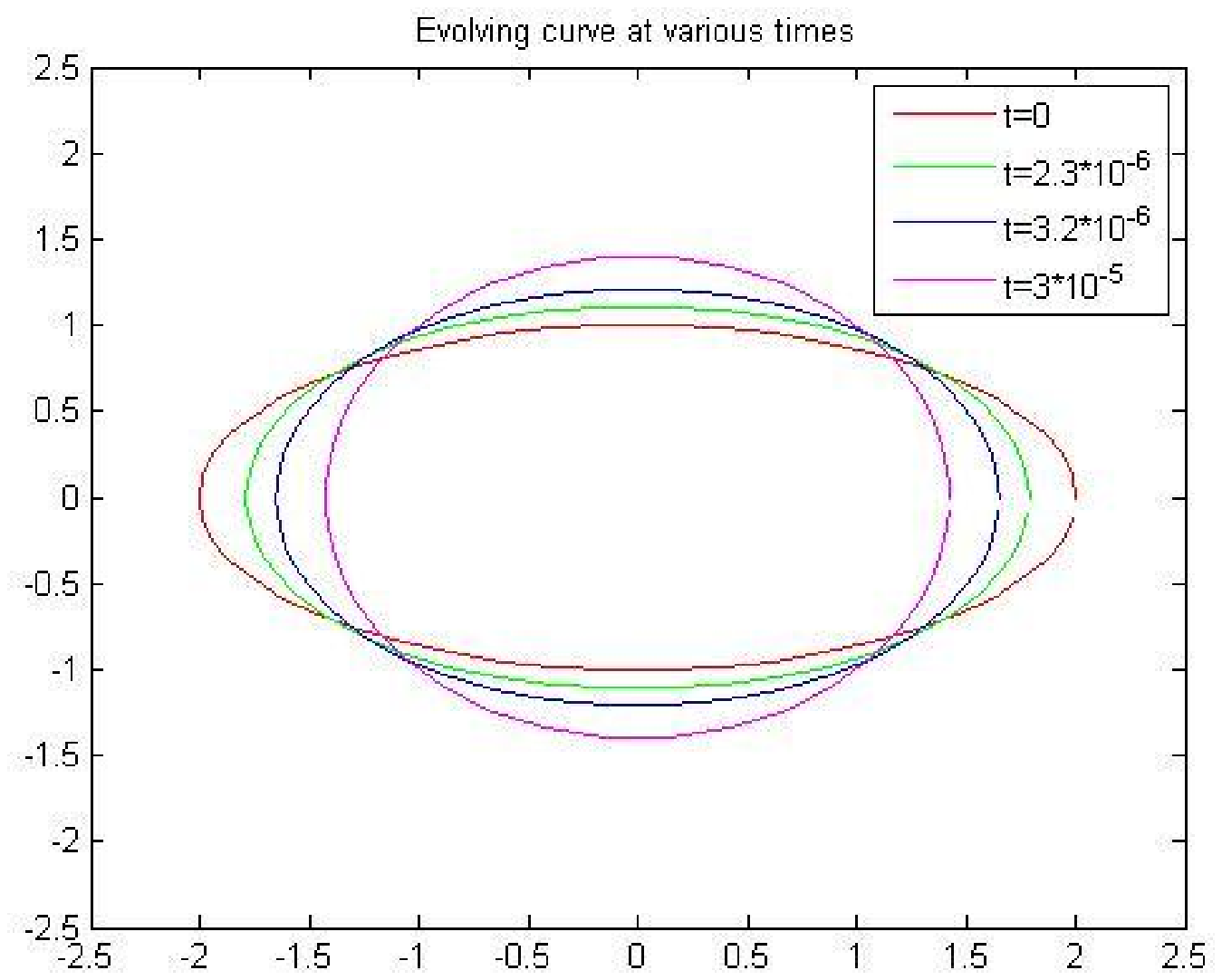}
					\caption{Curve evolution under \eqref{eq:MembraneMoveNew} given $a = 0$.}
					\label{fig:ellipse2Circle}
				\end{figure}
				Figure~\ref{fig:ellipse2Circle} plots the curve evolution at different (rescaled) times for an initial curve given by the ellipse
				\begin{align*}
					\Gamma^{0} := \{ \textbf{x} = (x_{1},x_{2})^{T} \in \mathbb{R}^2\ | \frac{x_{1}^2}{4} + x_{2}^2 = 1 \}
				\end{align*}
				where have chosen $h = 0.02$, $\Delta t = 100h^{2}$ and $\varepsilon = 1$ with final time $T = 10$. The initial ellipse should theoretically evolve into a circle with the same area. This provides a natural way to test the accuracy of the numerical method. We can calculate the area numerically as follows: denote by $\abs{\Omega_{h}^{m}}$ the area of the interior of the interpolated curve at timestep $m$, then we have
				\begin{align*}
					\abs{\Omega_{h}^{0}} = \abs{\Omega_{h}^{m}} = \int_{\Omega_{h}^{m}}1 = \int_{\Gamma_{h}^{m}} \frac{1}{2} \mathbf{X}^{h,m} \cdot \nu^{h,m}, \, m = 0, \dots, M,
				\end{align*}    
				where $\nu^{h,m} = \frac{(\mathbf{X}^{h,m})^{\perp}}{\abs{(\mathbf{X}^{h,m})^{\perp}}}$ is the corresponding outward normal. The area of the initial curve is given explicitly by $\abs{\Omega^{0}} = 2\pi \approx 6.2832$. The numerically computed areas are given as follows:   
				\begin{align*}
					\abs{\Omega_{h}^{0}} = 6.2832, \, \abs{\Omega_{h}^{8}} = 6.2999, \, \abs{\Omega_{h}^{16}} = 6.3117, \, \abs{\Omega_{h}^{75}} = 6.2941,
				\end{align*}
				where the superscript values are the values of $m$ corresponding to the times shown in Figure~\ref{fig:ellipse2Circle}. The variation in the computed numerical areas is due to the scheme we chose to consider, which took the area constraining Lagrangian term $\bar{\lambda}^{h}$ explicitly in time.

		\subsection{Full Reaction-Diffusion System} \label{subsec:FullResultsNumerics}
			Having described numerical methods for both the reaction-diffusion PDE and the curve evolution, we can now couple the two to simulate the entire system: the autocatalytic attractant $a_{h}$ is computed using \eqref{eq:RDTimeDisc2} and then taken explicitly in time in \eqref{eq:DiscreteCurveFEM1} and \eqref{eq:DiscreteCurveFEM2} to evolve $\Gamma_{h}(t)$.

			Figure~\ref{fig:NoRemesh} illustrates the progression of the neutrophil under the full system at various times (with corresponding activator levels on the right), where we have chosen $h = 5 \times 10^{-3}$, $\Delta t = O(h^{2})$, $\delta = 312.5$ and $\varepsilon = 10^{-6}$ with final time $T = 10$. The initial curve is chosen to be the unit circle and the initial condition for the activator is given by $a_{0} = 20 \exp \left( \frac{-(p - 0.5)^{2}}{0.0002} \right)$ (which acts as an initial impulse on the left-hand side of the neutrophil). The model constants are as in Table~\ref{tab:EqnConstants}.
			\begin{figure}[h!]
				\centering
				\includegraphics[height=110mm]{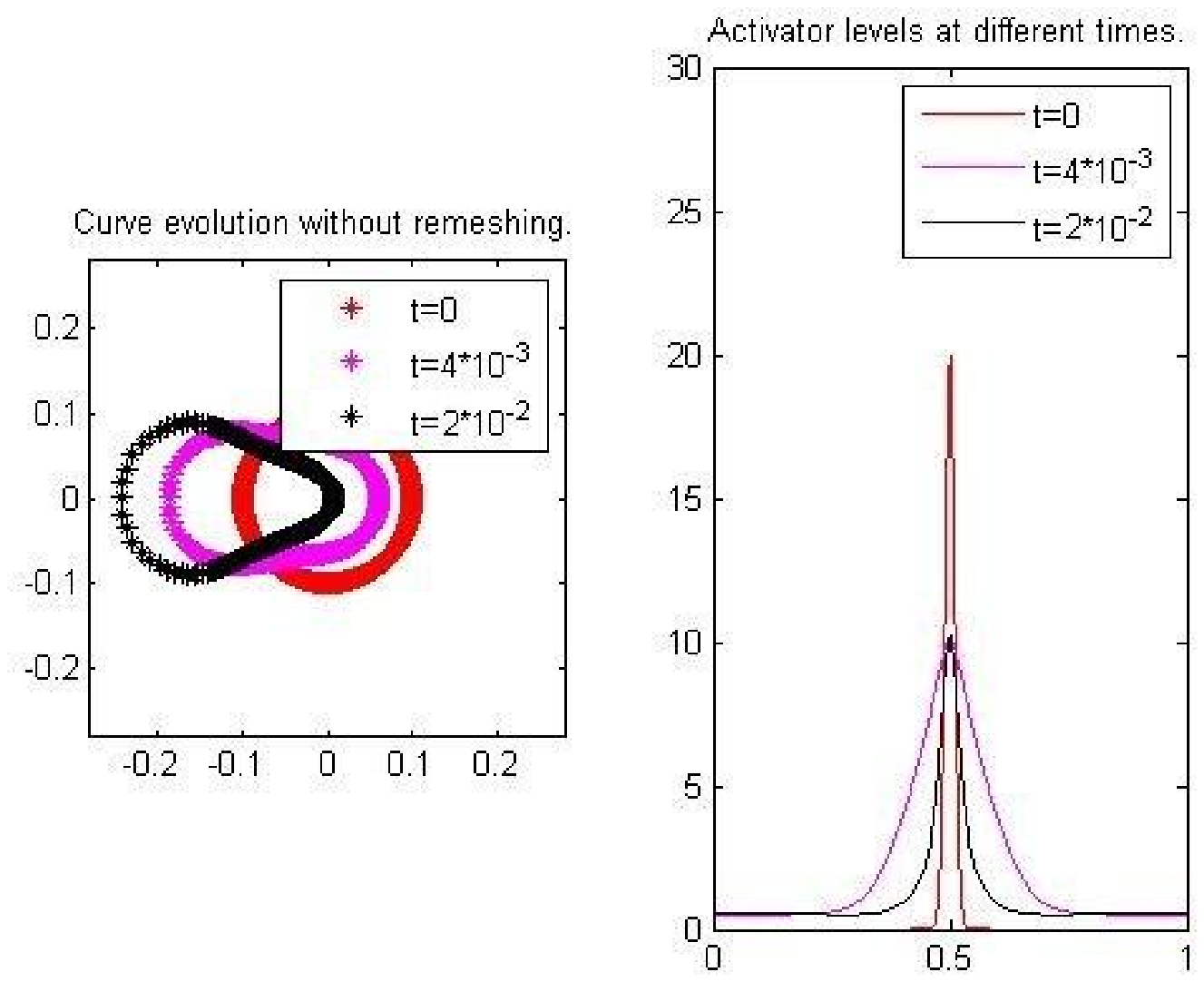}
				\caption{Membrane progression and local activator concentration levels simulated by the full system.}
				\label{fig:NoRemesh}
			\end{figure}

			The result is the formation of a protrusion in the direction of the initial impulse, driving the neutrophil to move in the same direction. However, as the figure shows, the advection of the neutrophil results in a progressive clustering of the points which will lead to an ill-conditioned problem. Our solver will thus have to include a remeshing step. 

			\subsubsection{Remeshing}
				The procedure for mapping $\mathbf{X}_{i}$ to the remeshed point $\mathbf{X}_{i}^{'}$ is given as follows: 
				\begin{align*}
					\mathbf{X}_{i}^{'} = \frac{\mathbf{X}_{i - 1}+\mathbf{X}_{i + 1}}{2} + ((\mathbf{X}_{i}-\mathbf{X}_{i - 1})\cdot \mathbf{n}_{i})\mathbf{n}_{i},\, i = 1, \dots, N,
				\end{align*}
				where $\mathbf{n}_{i}$ is the outward normal of the segment $[\mathbf{X}_{i - 1},\mathbf{X}_{i}]$. Figure~\ref{fig:remesh} illustrates this procedure. It is important to note that this remeshing step is area-preserving, which we require to ensure that the cell area remains constant.
				\begin{figure}[h!]
					\centering
					\includegraphics[height=90mm]{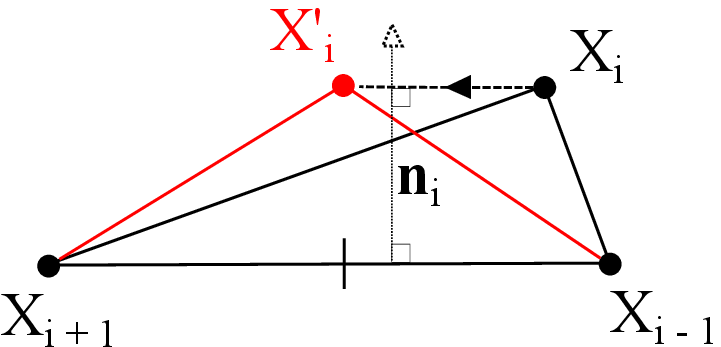}
					\caption{Illustration of the remeshing procedure.}
					\label{fig:remesh}
				\end{figure}

				This process alone creates an artificial, tangential transport component $v_{\tau}^{h}$ that has to be counter-balanced in the activator PDE with an advection term in the opposite direction, given by the matrix
				\begin{align*}
					\mathbf{B}_{i,j} = \int_{0}^{1} \Delta t (v')^{h}(p) \cdot \frac{\mathbf{X}_{p}^{'}}{\abs{\mathbf{X}_{p}^{'}}} \, \chi_{i,p} \, \chi_{j} \, \rd p, \; i,j = 1, \dots, N
				\end{align*}
				with $(v')^{h} = v^{h} + v_{\tau}^{h}$, where $v^{h}$ is the normal velocity of $\Gamma_{h}$. We are thus required to solve the new system
				\begin{align*}
					(\mathbf{M}^{m + 1} + \mathbf{B}^{m + 1} + \Delta t \mathbf{S}^{m + 1}) a^{m + 1} = M^{m}(\Delta t \mathbf{F}^{m} + \mathbf{a}^{m}).
				\end{align*}
				This allows us to adapt the activator values accordingly when the remeshing is performed. Figure~\ref{fig:WithRemesh} shows the effects of the remeshing on the progression of the neutrophil under the full system at various times (with corresponding activator levels on the right), where we have chosen the same discretisation/parameter values and initial conditions as before.
				\begin{figure}[h!]
					\centering
					\includegraphics[height=110mm]{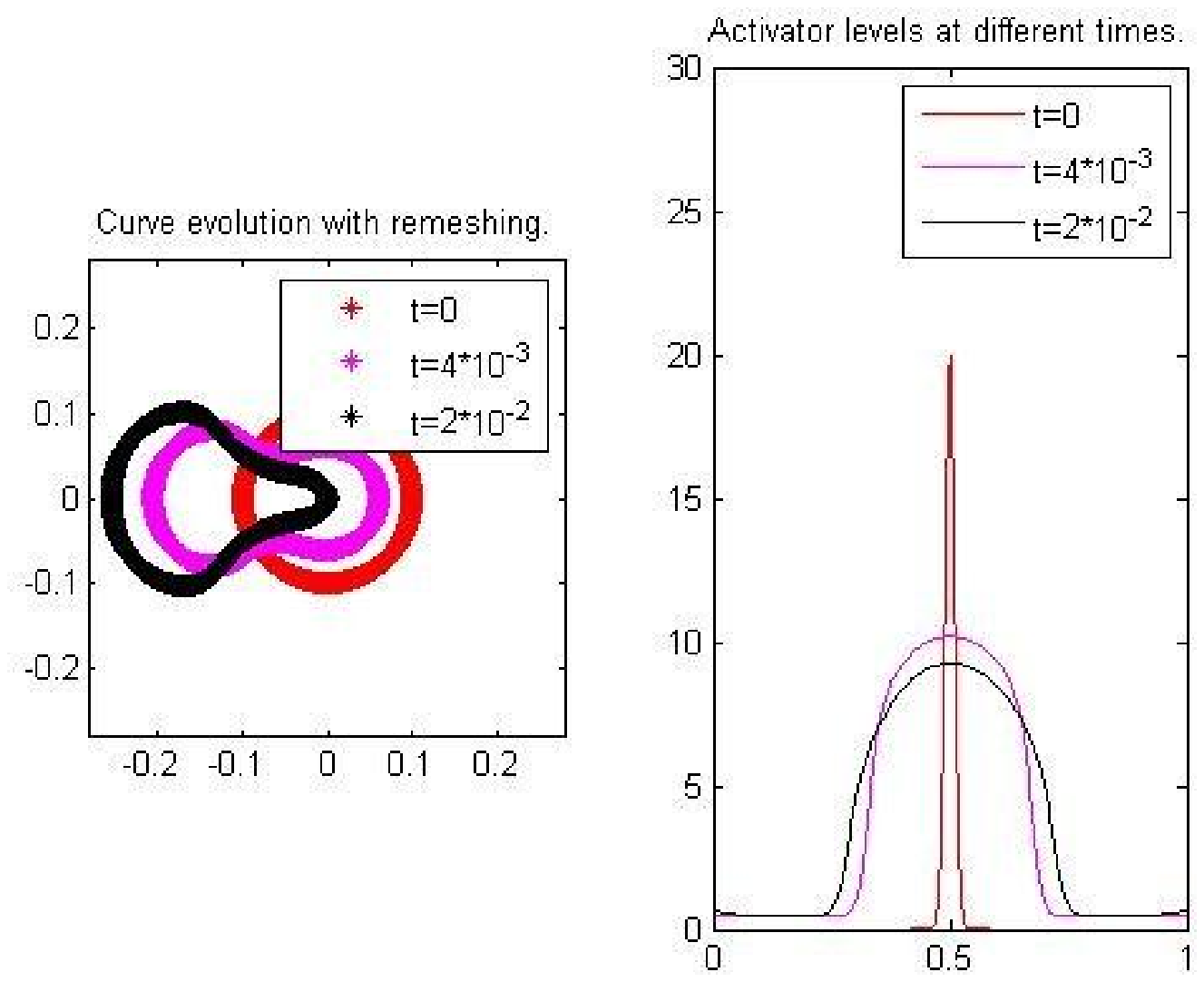}
					\caption{Membrane progression and local activator concentration levels simulated by the full system with remeshing.}
					\label{fig:WithRemesh}
				\end{figure}

			\subsubsection{Remarks}
				The results based on the original model of \cite{Neil10} suggest that a ``parent" pseudopod would split to give rise to two new ``child" pseudopods, a process that can be observed in the activator profile. Our results do not appear to show any such splitting: the one pseudopod that was created by the initial impulse does not result in any further pseudopods being formed.

				The reason for this may in fact be due to our reduction of the original model; in particular, neglecting the diffusion term in \eqref{eq:Neilson-c}. The Diffusion coefficient $D_{c}$ is greater than $D_{a}$, indicating that the local inhibitor may have an effect over a slightly larger part of the pseudopod than the local activator. We argue that this is because the activator's influence is over the tip of the pseudopod, whereas the local inhibitor affects both the pseudopod tip and the parts of the membrane connecting the pseudopod to the rest of the neutrophil. This suggests that if we want to observe pseudopod splitting then we should adapt our model to use \eqref{eq:NormNeilson-c} rather than \eqref{eq:NormNeilsonReduced-c}.

			\subsubsection{Extension}
				Motivated by the remarks above, the original local inhibitor equation \eqref{eq:NormNeilson-c} is incorporated back into our model. Figure~\ref{fig:NewModel} shows the resulting simulation of the neutrophil at various times (with corresponding activator and local inhibitor levels on the right), where we have chosen $h = 4 \times 10^{-2}, \Delta t = O(h^{2})$, $\delta = 312.5$ and $\varepsilon = 10^{-6}$ with final time $T = 10$. The initial curve is chosen to be the unit circle and the initial condition for the activator is given by $a_{0} = 20 \exp \left( \frac{-(p - 0.5)^{2}}{0.0002} \right)$. The pseudopod splitting discussed in the previous section is now clearly visible, confirming our intuition that the local inhibitor $c$ is crucial for this phenomena to arise.
				\begin{figure}[h!]
					\centering
					\subfigure[$t = 0$]{\includegraphics[width=3.3in]{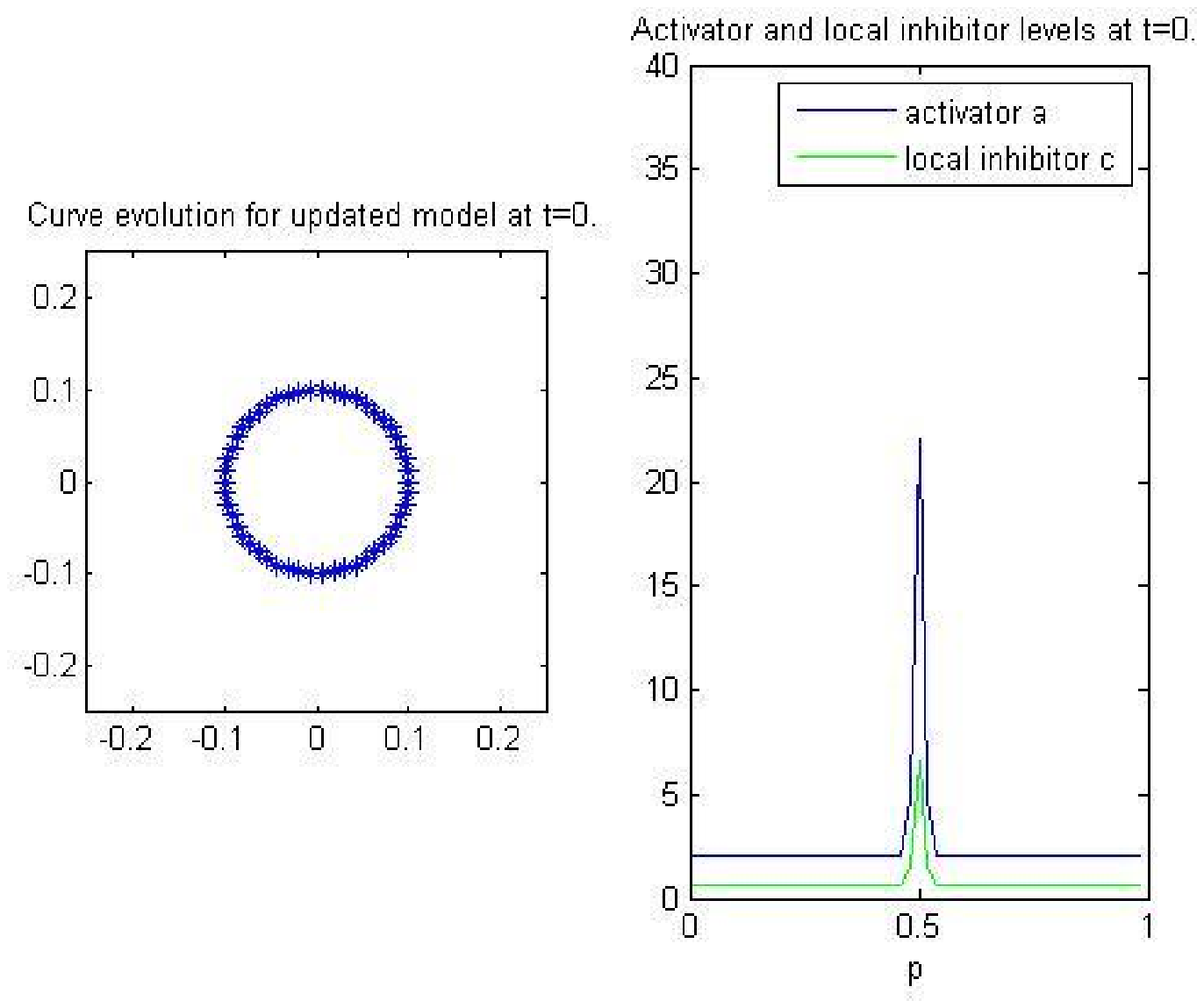}}
					\subfigure[$t = 3 \times 10^{-3}$]{\includegraphics[width=3.3in]{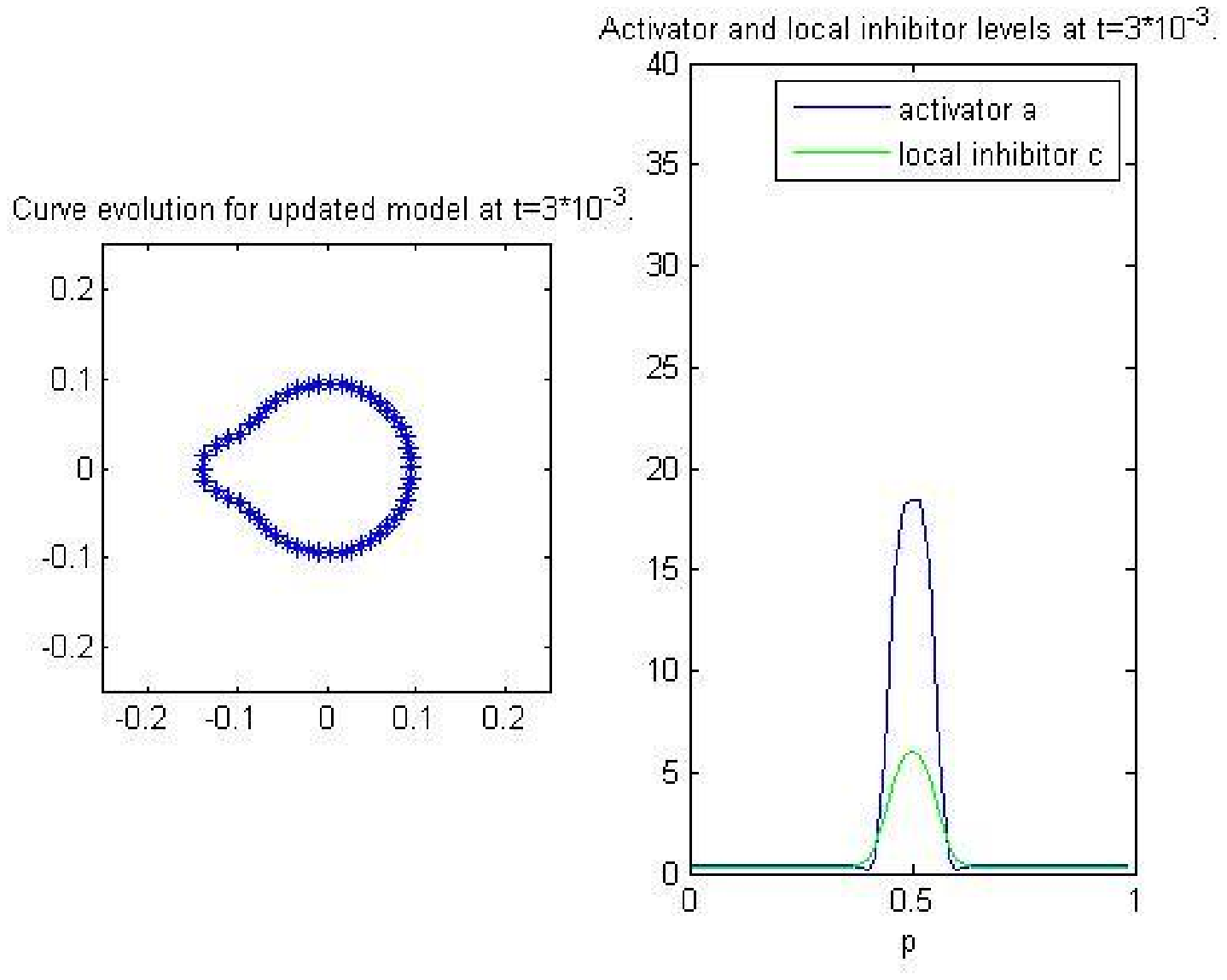}}\\
					\subfigure[$t = 5 \times 10^{-3}$]{\includegraphics[width=3.3in]{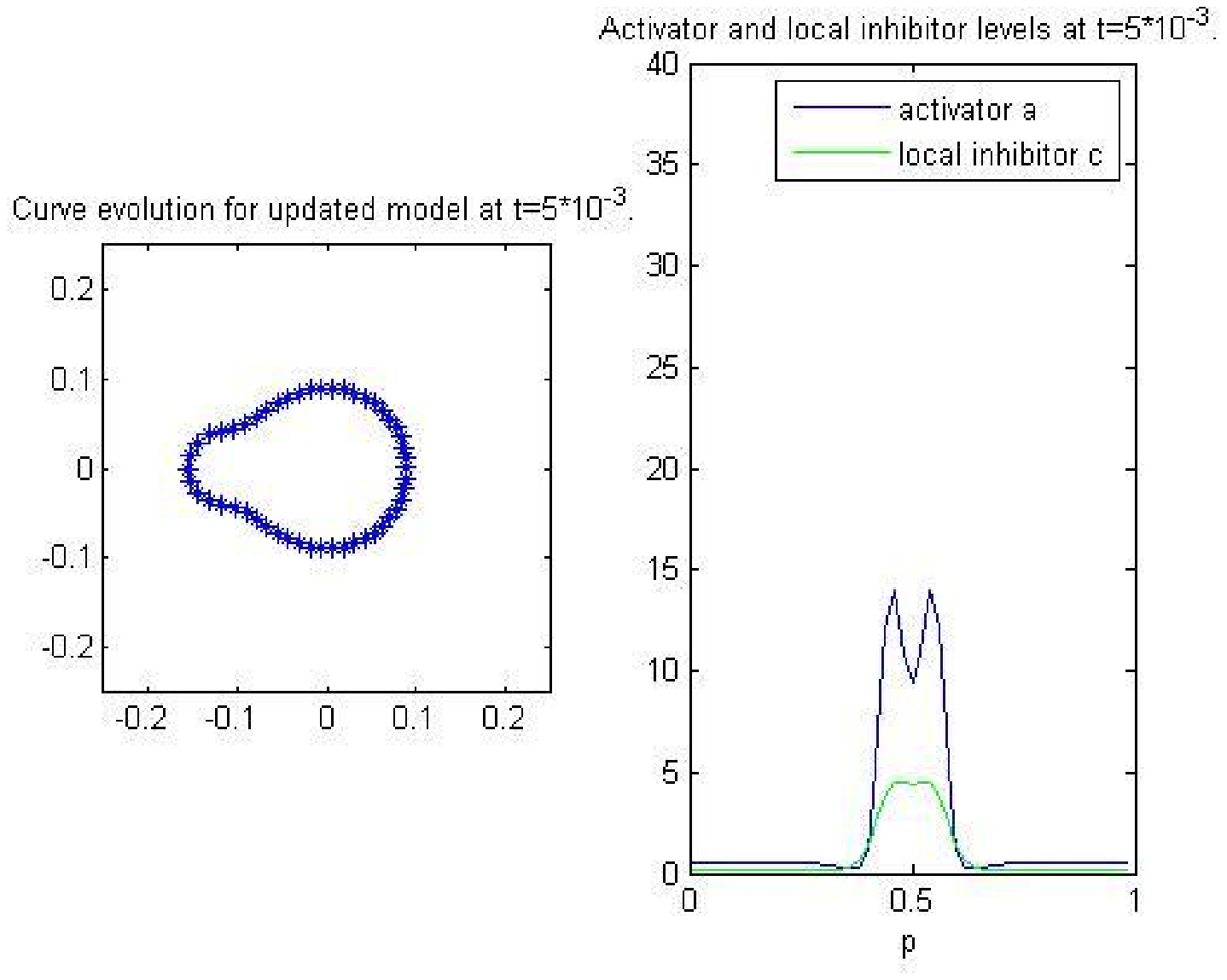}}
					\subfigure[$t = 1.2 \times 10^{-2}$]{\includegraphics[width=3.3in]{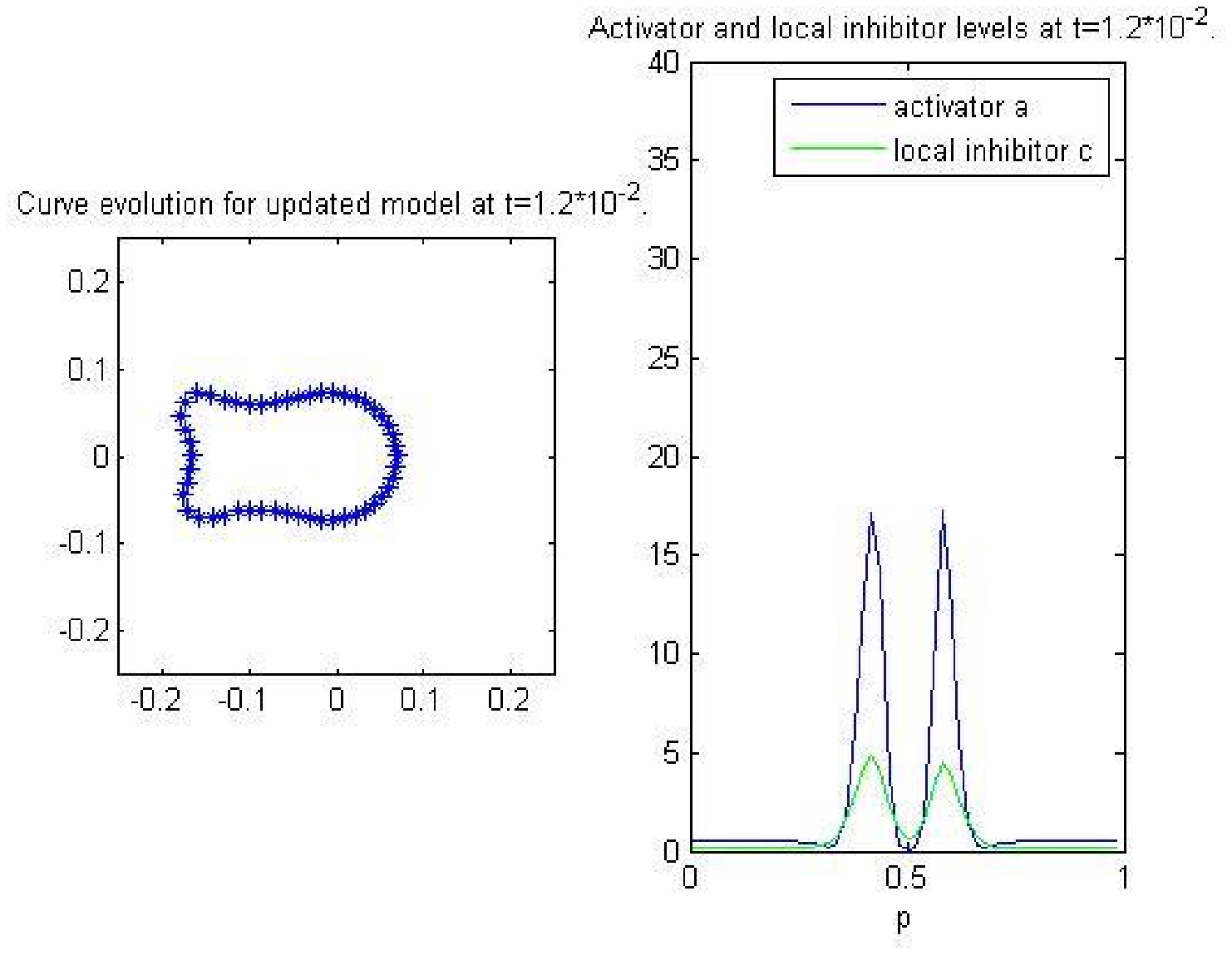}}\\
					\caption{Membrane progression and local activator and inhibitor concentration levels simulated by the full system including the original equation for local inhibitor $c$.}
					\label{fig:NewModel}
				\end{figure}

	\newpage

	\section{Probabalistic Approach to Bacterium and Neutrophil Movement} \label{sec:Prob}
		We now address the effects of cell movement and chemotaxis on the probability of a neutrophil catching and neutralising a pathogenic bacterium within the bloodstream. We initially model this game of `cat-and-mouse' on an infinite plane before considering the movement of the neutrophil with respect to the chemotaxical effects of the bacterium and then restricting the problem to within a capillary with negligible net blood flow. Finally, we consider the same capillary problem, but with an additional obstacle (perhaps a neighbouring blood cell) which either engulfs or deflects the bacterium.

		It seems reasonable to model the bacterium's movement as a Brownian Motion since we assume it is suspended in a fluid (blood plasma). Thus, pendent in this dynamic liquid, it is consistently undergoing collisions with the molecules of that liquid and if the number of collisions is large enough then the Central Limit Theorem gives rise to a Brownian Motion. We give this Brownian Motion a variance parameter, $\sigma$, which will depend on the dynamics of the fluid and the mass and dynamics of the bacterium itself.

		However, when modelling the movement of the neutrophil, the aforementioned chemotaxical effects allow us to be far more deterministic in our approach. For our purposes we first assume that the neutrophil is centred at the origin, and any of its movement is imparted onto the movement of the bacterium --- i.e. we observe the process from the neutrophil's perspective. For example, in two-dimensions, if we suppose that the neutrophil undergoes movement with constant velocity and moves with speed and direction $(1,1)^{\mathrm{T}}$ (for every unit of time) then, from its perspective, this is equivalent to an inclusion of a drift term in the movement of the bacterium equal to $(-1,-1)^{\mathrm{T}}$. Thus, in combination with the motion of the lone bacterium given above, the stochastic differential equation (SDE) that represents the path of the bacterium in our new frame is now given by
		\begin{align*}
			\rd X_{t} = \left( \begin{array}{c} -1 \\ -1  \end{array} \right) \rd t + \left( \begin{array}{cc} \sigma & 0 \\ 0 & \sigma \end{array} \right) \rd B_{t},
		\end{align*}
		or, more generally,
		\begin{align*}
			\rd X_{t} = \mathbf{b} \rd t + \bm{\sigma} \rd B_{t}.
		\end{align*}
		Owing to the fact that chemotaxis is a process based on the detection of a concentration of particles, we suggest that it is a diffusive process; thus its effects decay exponentially as one moves away from the source. Therefore, in two dimensions, we suggest the following model equation for the movement of the neutrophil with respect to the position of the bacterium, $(x_{1},x_{2})^{\mathrm{T}}$:
		\begin{align*}
			\left( \begin{array}{c} \dot{x_{1}} \\ \dot{x_{2}} \end{array} \right) = \left( \begin{array}{c} \theta x_{1} e^{-D} \\ \theta x_{2} e^{-D} \end{array} \right)
		\end{align*}
		where $\theta$ is some scale factor dependent on the cell's detection of chemoattractants and its subsequent movement in that direction. This suggests that the SDE describing the path of the bacterium from the perspective of the neutrophil is given by
		\begin{align*}
			\rd X_{t} = \left( \begin{array}{c} -kx_{1} e^{-D} \\ -kx_{2} e^{-D} \end{array} \right) \rd t + \left( \begin{array}{cc} \sigma & 0 \\ 0 & \sigma \end{array} \right) \rd B_{t}.
		\end{align*}
		In Sections~\ref{subsec:ModelMembraneMove} and~\ref{subsec:MembraneMoveNumerics}, we considered how the neutrophil membrane varies over time in response to the bacterium's presence with a view to coupling the equations derived here with the reaction-diffusion and membrane movement PDEs. For this section, however, we work with the unlikely, but hopefully not too inaccurate, assumption that the neutrophil membrane maintains a perfectly circular shape and is unaffected by the movement of the neutrophil or the dynamics of the surrounding fluid.

		We now turn our attention to the various domains that may be considered. Initially, we consider a bacterium trying to escape a neutrophil on an infinite, two-dimensional plane. Using $R$, the radius of the neutrophil, as a convenient scale factor, we assume that the bacterium succeeds in its escape if it reaches a distance of $kR$ from the neutrophil, for some $k > 1$. This provides us with a domain in the shape of an annulus,
		\begin{align*}
			\Omega = \{ x \in \R^{2} \vert R \leq x \leq kR \}.
		\end{align*}
		From here, we can then extend this to a bacterium within a capillary. Finally, we shall attempt to add obstacles such as other blood cells within the capillary.

		As mentioned previously, the path of the bacterium can be given by an SDE with a drift term $b(X_{t},t)$ (examples of which were given earlier) and a constant variance, $\sigma$. Using the annulus as our domain with $kR$ as the escape radius, we propose a method for computing the probability that a bacterium is engulfed by the neutrophil given that the bacterium starts at a point $\beta \in \Omega^{\circ}$.

		Suppose we also have a function $f(x,t) \in C^{2} \left( [0, \infty) \times \R \right)$, which describes the probability that the bacterium will escape given that it is in position $x$ at time $t$. Consequently, using Ito's Lemma \cite{Oks03},
		\begin{equation} \label{eq:ProbEscIto}
			\begin{aligned}
				\rd f(X_{t},t)
				&= \pd_{t} f(X_{t},t) \rd t + \nabla f(X_{t}, t) \cdot \left( b(X_{t}, t) \rd t + \left( \begin{array}{c} \sigma \\ \sigma \end{array} \right) \rd W_{t} \right) + \frac{1}{2} \sigma ^{2} \Delta f \rd t
				\\
				&= Af(X_{t},t) \rd t + \nabla f(X_{t},t) \cdot \left( \begin{array}{c} \sigma \\ \sigma \end{array} \right) \rd W_{t} 
			\end{aligned}
		\end{equation}
		where $Af(X_{t},t)$ is the generator associated with our SDE \cite{Oks03}. Integrating \eqref{eq:ProbEscIto} from 0 to $T$ and taking expectations gives
		\begin{align} \label{eq:ProbEscItoIntExp}
			\Exp \left[ f(X_{t},t) \right] - f(X_{0},0) = \Exp \left[ \int_{0}^{T} Af(X_{s},s) \rd s \right] + \Exp \left[ \int_{0}^{T} \nabla f(X_{t},t) \cdot \left( \begin{array}{c} \sigma \\ \sigma \end{array} \right)\rd W_{t} \right].
		\end{align}
		By the martingale property of the Ito Integral,
		\begin{align*}
			\Exp \left[ f(X_{t},t) \right] = f(X_{0},0) + \Exp \left[ \int_{0}^{T} Af(X_{s},s) \rd s \right].
		\end{align*}
		Equivalently, by the Optional Stopping Theorem, for stopping time $\tau = \inf \{ t \geq 0 \vert X_{t} \notin \Omega \}$ (the time that the bacterium exits our domain),
		\begin{align*}
			\Exp \left[ f(X_{\tau},\tau) \right] = f(\beta,0) + \Exp \left[ \int_{0}^{\tau} Af(X_{s},s) \rd s \right]
		\end{align*}
		where $\beta$ is the starting point of the bacterium. However, we also note that
		\begin{align*}
			\Exp \left[ f(X_{\tau},\tau) \right] = \Prob(X_{\tau} = R) f(R,\tau) + \Prob(X_{\tau} = kR) f(kR, \tau).
		\end{align*}
		Thus,
		\begin{align*}
			\Prob(X_{\tau} = R) f(R,\tau) + \Prob(X_{\tau} = kR) f(kR,\tau) = f(\beta,0) + \Exp \left[ \int_{0}^{\tau} Af(X_{s},s) \rd s \right]
		\end{align*}
		for every $f(x,t) \in C^{2} \left( [0, \infty) \times \R \right)$. So if we can solve the partial differential equation (PDE) $Af = 0$ with boundary conditions $f(X) = 1$, when $\abs{X} = kR$, and $f(X) = 0$, when $\abs{X} = R$, this provides us with
		\begin{align} \label{eq:XTauandfb}
			\Prob(X_{\tau} = kR) = f(\beta,0)
		\end{align}
		since $f(kR,\tau) = \Exp \left[ \int_{0}^{\tau} Af(X_{s},s) \rd s \right] = 0$.

		However, when solving $Af = 0$ for a parabolic generator, we need a third, time-related boundary condition. To this point, for a time dependent drift term $b(X_{t},t)$ we have the following problem:
		\begin{equation} \label{eq:DriftTermPDE}
			\begin{aligned}
				&\pd_{t} f + b(x,t) \cdot \nabla f + \frac{1}{2} \sigma^{2} \Delta f = 0
				\\
				&f = 0 \qquad \mbox{in} \quad \pd B_{R}(0) \times (0,\infty)
				\\
				&f = 1 \qquad \mbox{in} \quad \pd B_{kR}(0) \times (0,\infty)
			\end{aligned}
		\end{equation}
		For the third boundary condition, we make the assumption that the time dependence of the drift term diminishes over time --- i.e. the neutrophil's strategy for chasing the bacterium still depends on the position of the bacterium, but becomes less volatile over time. Equivalently,
		\begin{align*}
b(X_{t},t) \rightarrow d(X_{t}) \quad \mbox{as} \quad t \rightarrow \infty.
		\end{align*}
		for some function $d$ that depends only on $X_{t}$. Furthermore, we suppose that there exists a time $T < \infty$ where
		\begin{align*}
			\forall t > T, \quad b(X_{t},t) = d(X_{t}),
		\end{align*}
		Thus our third boundary condition for our parabolic problem is the solution, $u(x)$, to the elliptic problem
		\begin{equation} \label{eq:ParbolicBCEllipticPDE}
			\begin{aligned}
				&d(x) \cdot \nabla u + \frac{1}{2} \sigma^2 \Delta u = 0
				\\
				&u = 0 \qquad \mbox{in} \quad \pd B_{R}(0) \times (0,\infty)
				\\
				&u = 1 \qquad \mbox{in} \quad \pd B_{kR}(0) \times (0,\infty)
			\end{aligned}
		\end{equation}
		Finally, letting $s = -t$ to obtain a forward heat equation, we derive the parabolic PDE
		\begin{equation} \label{eq:BacMovePDE}
			\begin{aligned}
				&\pd_{s}f - b(x,s) \cdot \nabla f - \frac{1}{2} \sigma^2 \Delta f = 0
				\\
				&f = 0 \qquad \mbox{in} \quad \pd B_{R}(0) \times (-T,0)
				\\
				&f = 1 \qquad \mbox{in} \quad \pd B_{kR}(0) \times (-T,0)
				\\
				&f = u(x) \qquad \mbox{in} \quad \Omega \times \{ -T \}
			\end{aligned}
		\end{equation}
		The completion of this formulation of the problem involves showing that the values of solution to this PDE, $f(x,t)$, take values between 0 and 1, as required of a probability function, since our solution $f$ represents the probability of the bacterium escaping.
		\begin{thm} \label{thm:Boundedf}
			If $f(x,t)$ satisfies \eqref{eq:BacMovePDE} with the given boundary conditions on the domain $\Omega$, which consists of an annulus with inner radius $R$ and outer radius $kR$, then $\forall x \in \Omega, t \in [-T,0]$,
			\begin{align*}
				0 \leq f(x,t) \leq 1
			\end{align*}
		\end{thm}
		\begin{proof}
			The partial differential equation in question is equivalent to 
			\begin{align*}
				\pd_{t} f + Lf = 0 \qquad \mbox{where} \quad Lf = b(\cdot,t) \cdot \nabla f - \frac{1}{2} \sigma^2 \Delta f
			\end{align*}
			Now, since $\pd_{t} f + Lf \leq 0$, by the Weak Maximum Principle \cite{Evans10},
			\begin{align*}
				\max_{\Omega \times [-T,0]} f = \max_{\pd \Omega \times [-T,0]} f = 1 \qquad \mbox{and} \qquad \min_{\Omega \times [-T,0]} f = \min_{\pd \Omega \times [-T,0]} f = 0.
			\end{align*}
			Thus, $\forall x \in \Omega, t \in [-T,0]$, $0 \leq f(x,t) \leq 1$.
		\end{proof}
		With this general framework in place we attempt to find a solution for several special cases.

		\subsection{No Neutrophil Movement} \label{subsec:NoNeuMove}
				The first and simplest case that we consider is the case on the two-dimensional annulus when there are no chemotaxical effects and therefore negligible neutrophil movement.
				\begin{thm} \label{thm:ProbEngulf2d}
					For a stationary neutrophil and a bacterium that is modelled by a Brownian Motion with variance 1 beginning at a point $\beta$ such that $R < \abs{\beta} < kR$, where $R$ and $kR$ are the cell radius and escape radius respectively, and that is almost surely engulfed if it makes contact with the neutrophil,
					\begin{align*}
						\Prob(\mbox{Bacterium engulfed}) = 1 - \frac{1}{\ln k} \ln \left( \frac{\beta}{R} \right).
					\end{align*}
				\end{thm}
				\begin{proof}
					In this case the bacterium follows the path 
					\begin{align*}
						\rd Z_{t} = \rd B_{t}.
					\end{align*}
					So the drift term is zero and there is no time dependence. By \eqref{eq:XTauandfb}, we observe that the solution to this problem is given by the solution $f(x,t)$ to \eqref{eq:DriftTermPDE} with no time dependence and the drift term $b(x,t)$ equal to zero. Therefore, the probability of escape is given by the solution to the PDE:
					\begin{equation*}
						\begin{aligned}
							&\frac{1}{2} \Delta f = 0
							\\
							f = 0& \quad \mbox{in} \quad \pd B_{R}(0)
							\\
							f = 1& \quad \mbox{in} \quad \pd B_{kR}(0)
						\end{aligned}
					\end{equation*}
					Since we are on the annulus it will be easier to change to radial coordinates and we note that we choose a radially symmetric solution.
					\begin{align*}
						\frac{1}{r} \frac{\pd}{\pd r} \left( r \frac{\pd f}{\pd r} \right) = 0
					\end{align*}
					Solving this ODE gives
					\begin{align*}
						f(r) = \frac{1}{\ln k} \ln \left( \frac{r}{R} \right)
					\end{align*}
					This gives the probability that the bacterium will escape given that it starts at $r$. Hence the probability that the bacteria is engulfed is given by $1 - f(\beta)$.
				\end{proof}
				\begin{rem}
					We also note that since Brownian Motion is recurrent (in the sense of \cite[\S 4]{ChowTeich97}), to a ball with area greater than zero in two dimensions, as we increase the escape radius $kR$ we also increase the probability that the bacterium is caught. In fact, as $k \rightarrow \infty$, $\Prob(\mbox{Bacterium engulfed}) \rightarrow 1$. With infinite time and no obstructions the neutrophil will always engulf the bacterium.
				\end{rem}
				Similarly, in three dimensions:
				\begin{thm} \label{thm:ProbEngulf3d}
					For a stationary neutrophil and a bacterium that is modelled by a Brownian Motion in three dimensions beginning at a point $\beta$ such that $R < \abs{\beta} < kR$ where $R$ and $kR$ are the cell radius and escape radius respectively, and that is almost surely engulfed if it makes contact with the neutrophil,
					\begin{align*}
						\Prob(\mbox{Bacterium engulfed}) = 1 - \left( 1 - \frac{R}{\beta} \right) \left( \frac{1}{\left( 1 - \frac{1}{k} \right)} \right).
					\end{align*}
				\end{thm}
				\begin{proof}
					In the same way as the previous theorem
					\begin{align*}
						\Delta f = 0 \quad \Rightarrow \quad
						\frac{1}{r^{2}} \frac{\pd}{\pd r} \left( r^{2} \frac{\pd f}{\pd r} \right) = 0.
					\end{align*}
					Solving this ODE gives
					\begin{align*}
						f(r) = \left( 1 - \frac{R}{r} \right) \left( \frac{1}{\left( 1 - \frac{1}{k} \right)} \right).
					\end{align*}
					Again, the probability of being engulfed is $1 - f(\beta)$.
				\end{proof}

		\subsection{Movement of the Neutrophil}
			Now that we have obtained a rough outline of the probability function, we endeavour to find solutions in the more general case when $b \neq 0$ (movement of the neutrophil). Analytic solutions to the equations \eqref{eq:DriftTermPDE} are not obvious and none of the immediate methods for solving PDEs, such as the separation of variables technique, produce a solution. Therefore, we move to calculating these solutions computationally using linear finite element methods and linearly approximated domains on the software package DUNE. In Figures~\ref{fig:b55} to~\ref{fig:Obstacle}, we model the probability of a bacterium escaping a neutrophil, in a variety of situations, as a function of the starting position of the bacterium relative to the neutrophil (centred at the origin of each figure). i.e. a value close to one at a point in these figures implies that a bacterium, which starts at that point, is very likely to escape the neutrophil. \emph{Throughout these figures, the left image displays the PDE solution of the probability function that models the probability of a bacterium escaping and the right image displays the 'iso-probability lines' joining possible starting points of the bacterium that give an equal probability of the bacterium escaping}.
			
			Figures~\ref{fig:b55}, \ref{fig:b1010} and~\ref{fig:Drift} demonstrate how the choice of drift term affects the bacterium's chances of escape on an annulus domain and the PDE representing this, together with its associated computation, is discussed in detail in the captions below these figures. The unlikeliness of this particular environment is ameliorated by considering the more realistic setting of a capillary in Figure~\ref{fig:DriftCapillary}. This case is considered by solving equation \eqref{eq:DriftTermPDE} on a rectangular domain with Dirichlet boundary conditions, which are described, together with the domain approximations, more thoroughly in the figure caption. However, in figure~\ref{fig:DriftCapillary}, we have assumed that a bacterium can pass through capillary walls. We attempt to rectify this by using reflective boundaries to represent an impassable capillary wall in Figure~\ref{fig:DriftCapillaryReflect} (this too is discussed further in the associated caption). Finally, in Figure~\ref{fig:Obstacle}, we introduce an obstacle, which we suppose represents a passing blood cell. For this scenario, we assume that the bacterium can either enter or feed off the cell causing irreparable damage in both cases. Thus, we assume that if the bacterium reaches the cell the neutrophil has failed in its attempts to neutralise it. Again, for further discussion on how this model was formulated and computed the relevant caption should be consulted.
			\begin{figure}[h!]
				\centering
				\subfigure[]{\includegraphics[width=3.1in]{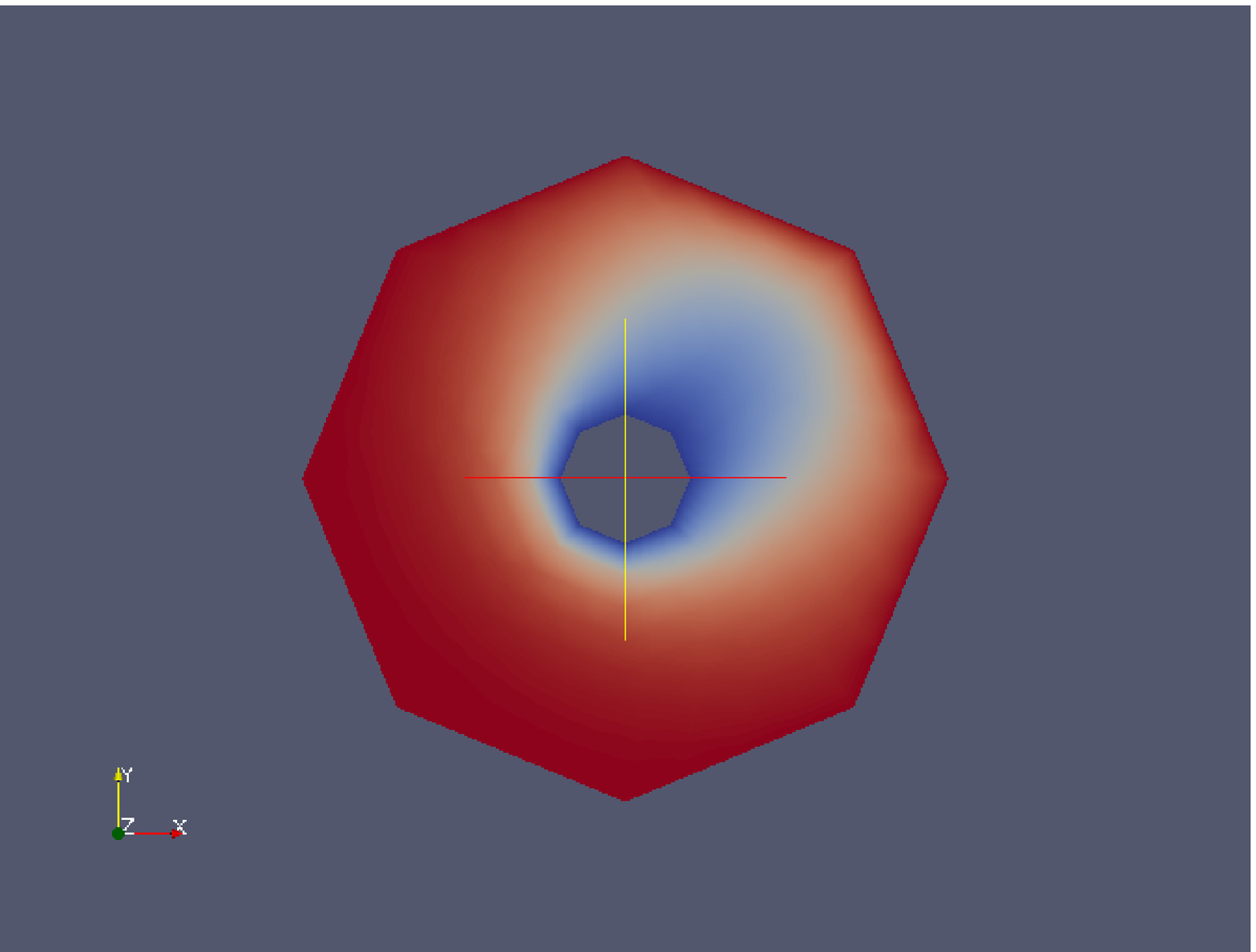}}
				\subfigure[]{\includegraphics[width=3.1in]{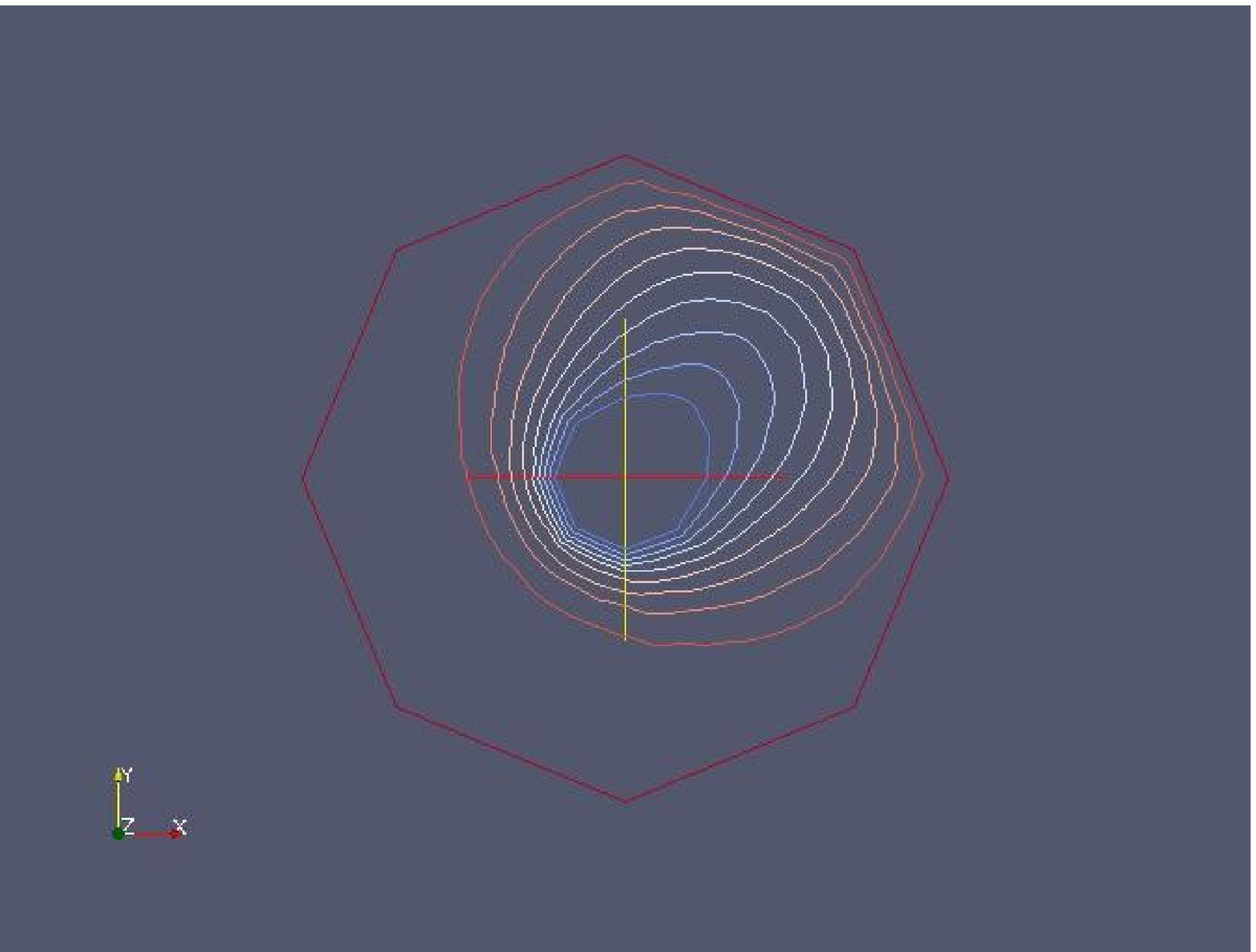}}\\
				\caption{We begin with the case when $b = (5, 5)^\mathrm{T}$ and $\sigma = 1 $. This is an example of a neutrophil moving fairly quickly in the $(1,1)^\mathrm{T}$, direction unaffected by chemotaxis. In an infinite domain it could be argued that this strategy (moving in a constant direction) would be the optimum strategy that could be employed by the neutrophil since this strategy would cover the most area in a limited amount of time. However, from the figures, on this restricted domain, it appears obvious that a bacterium is very likely to escape unless it starts in the direct path of the neutrophil or extremely close it. Despite the strong initial argument for this strategy it seems to be largely ineffective on more restricted domains. The figure itself shows a solution to a steady state case of a PDE derived previously (equation \eqref{eq:BacMovePDE}) to provide the probability of a bacterium escaping on an annulus. Thus, it is given by: $(5,5)^\mathrm{T} \cdot f(x) + \frac{1}{2} \Delta f(x) = 0$ in $\Omega$; $f = 0$ in $\pd B_{0.2}(0)$ and $f = 1$ in $\pd B_{1}(0)$, where $\Omega = \{ x \in \R^{2} \vert R \leq x \leq kR \}$.}
				\label{fig:b55}
			\end{figure}
			\begin{figure}[h!]
				\centering
				\subfigure[]{\includegraphics[width=3.1in]{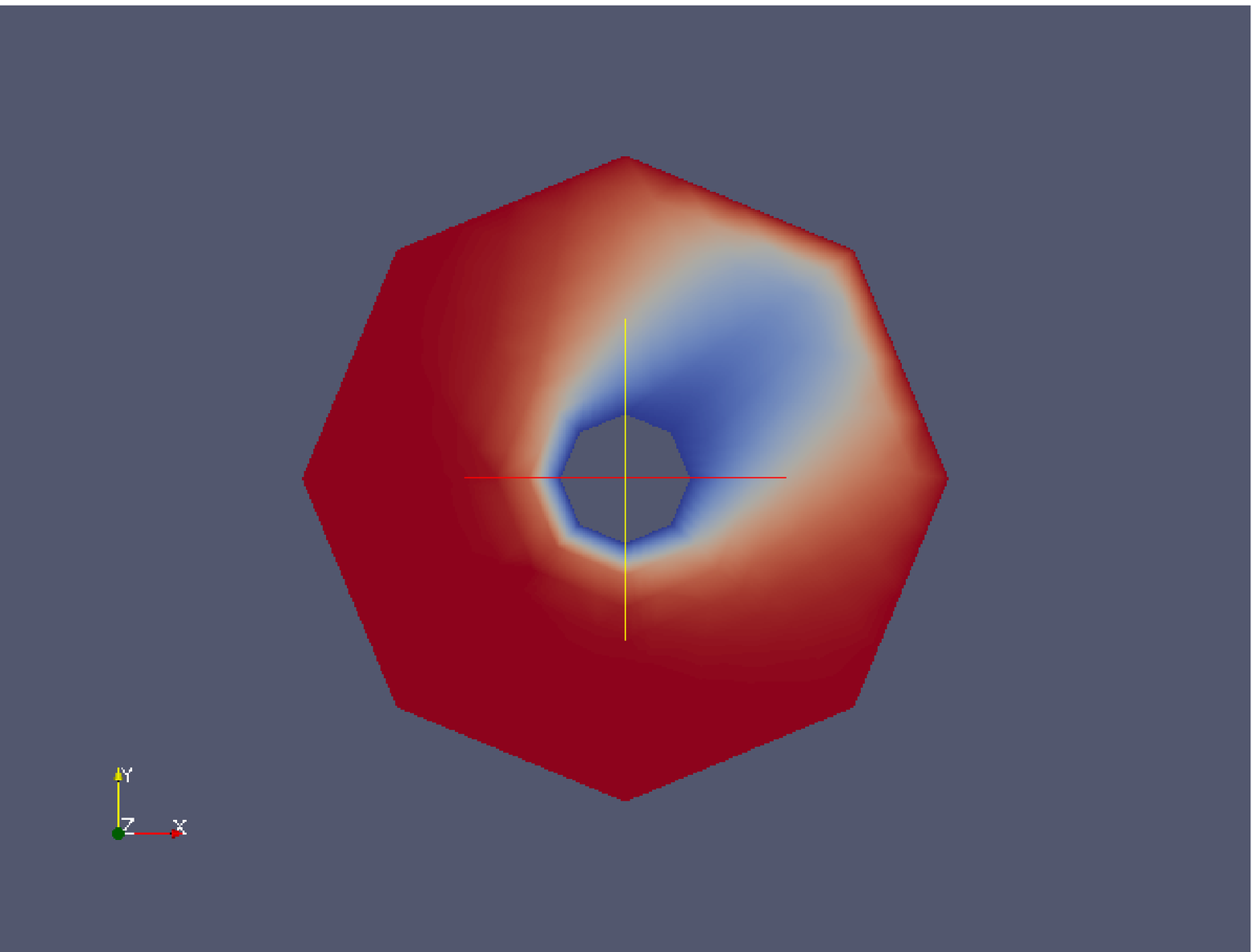}}
				\subfigure[]{\includegraphics[width=3.1in]{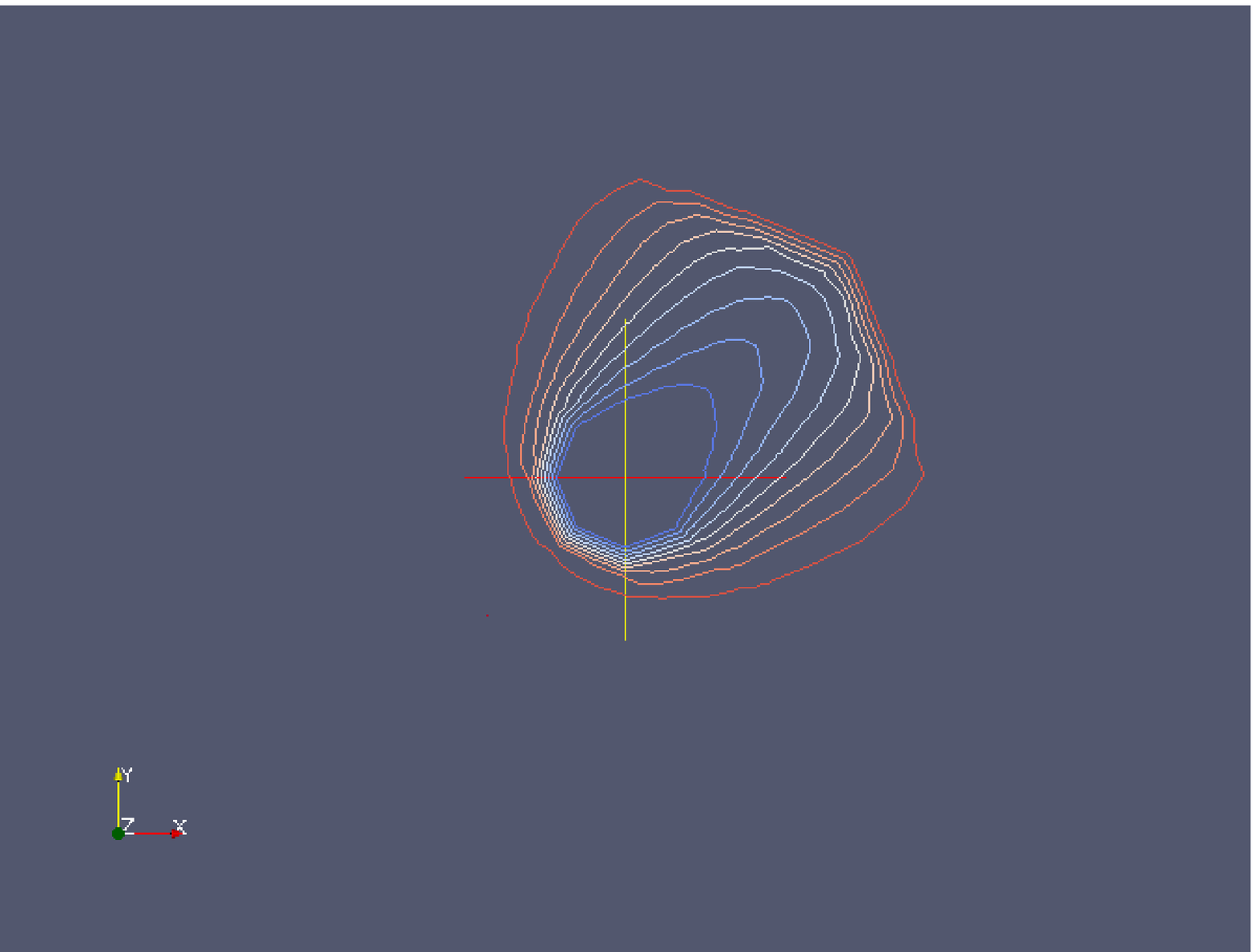}}\\
				\caption{For further comparison, We also observe the case when $b = (10, 10)^{\mathrm{T}}$ and $\sigma = 1 $. This is an example of a neutrophil moving very quickly in the $(1,1)^{\mathrm{T}}$ direction unaffected by chemotaxis. By the original argument for a neutrophil's strategy on an infinite domain this strategy should prove even more beneficial. However, on this restricted domain it only serves to strengthen the arguments against this strategy. The figure shows the solution to the same PDE as before, derived in Section\ref{sec:Prob} (equation \eqref{eq:BacMovePDE}), but with the noted change to the diffusion term: $(10,10)^{\mathrm{T}} \cdot f(x) + \frac{1}{2} \Delta f(x) = 0$ in $\Omega$; $f = 0$ in $\pd B_{0.2}(0)$ and $f = 1$ in $\pd B_{1}(0)$, where $\Omega = \{ x \in \R^{2} \vert R \leq x \leq kR \}$.}
				\label{fig:b1010}
			\end{figure}
			\begin{figure}[h!]
				\centering
				\subfigure[]{\includegraphics[width=3.1in]{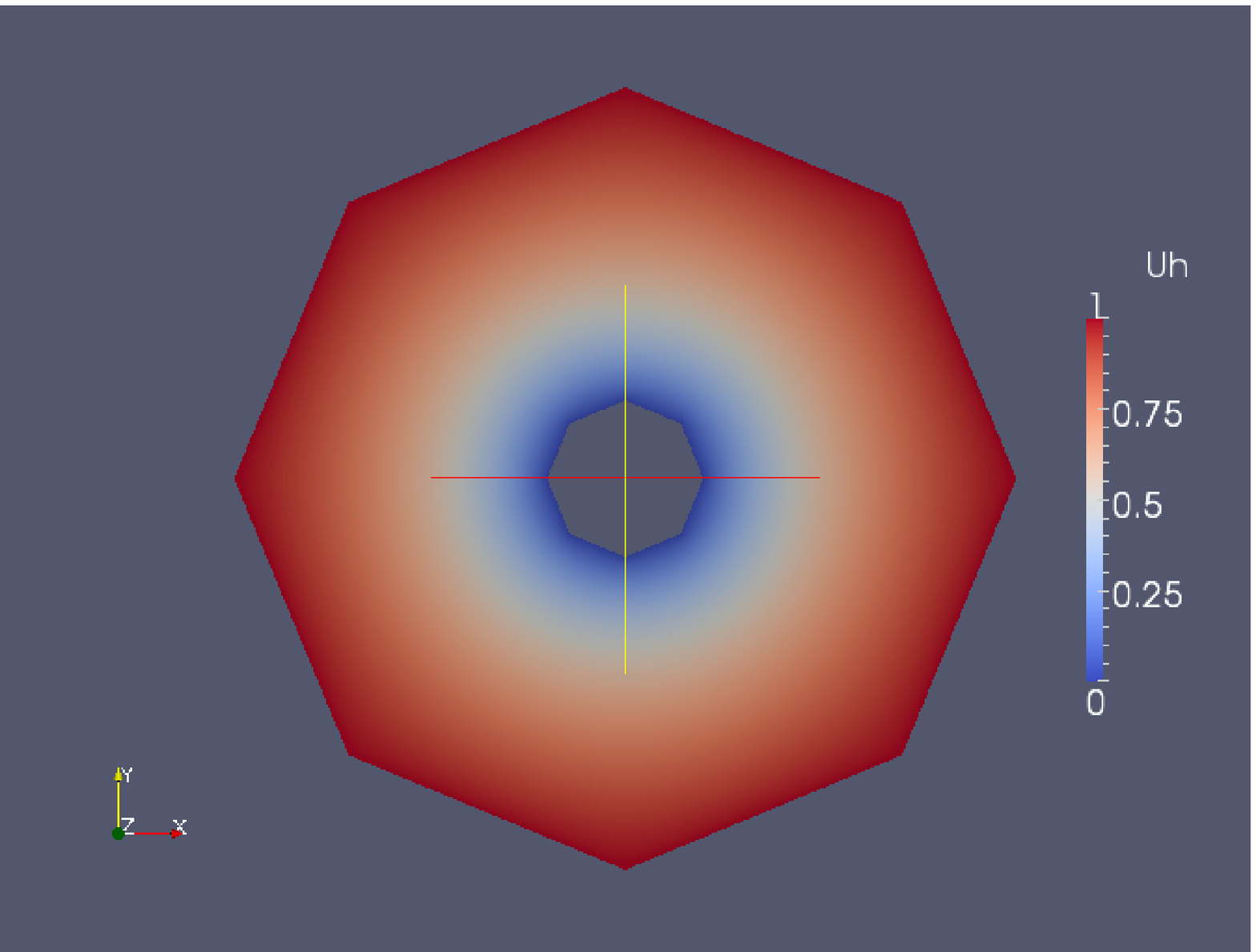}}
				\subfigure[]{\includegraphics[width=3.1in]{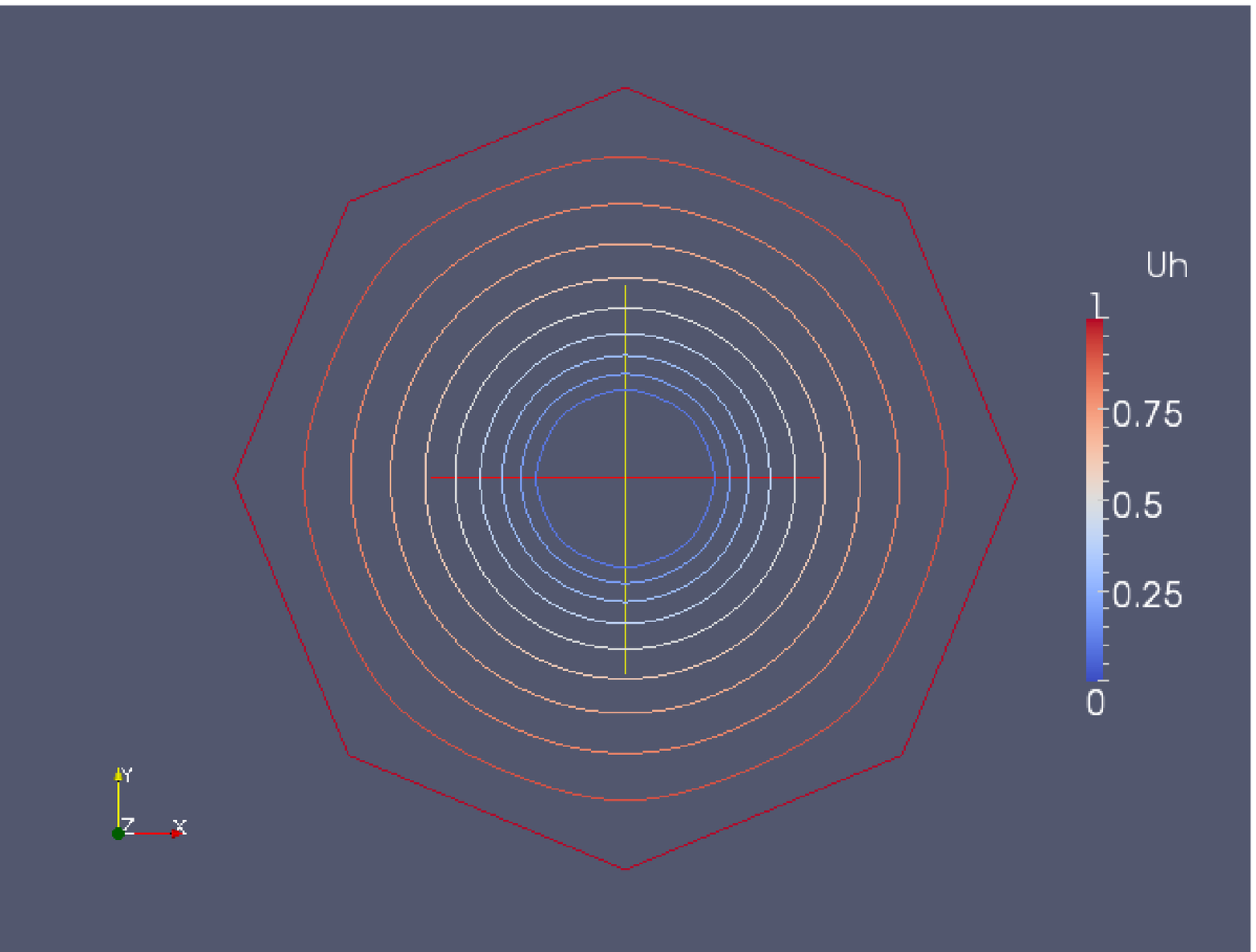}}\\
				\caption{We now consider the chemotaxical effects with the aforementioned suggestion of a drift term that induces these effects. We use the drift function suggested in Section~\ref{sec:Prob} with $\theta = 0.5$: $b(x) = (0.5x_1e^{-(|x| - R)}, 0.5x_2e^{-(|x|-R)})^T$. It should be noted that this function has no empirical or theoretical background and was only chosen because we expect this function to confer several of the properties we expect from chemotaxis: importantly, a greater effect as the bacterium gets closer to the cell; this effect should decay exponentially because it is based on the diffusion of particles. The most prominent thing we note is that chemotaxis is a process based upon concentrations of amino acids in a fluid and the diffusions of these concentrations so the effects of this process decay exponentially with distance. Due to this, in our model at least, the bacterium significantly increases its chances of escaping by starting only a relatively small distance further from the neutrophil. In addition to this we also note that despite the very limited movement we have put on the cell, with the $\theta = 0.5$, term this strategy seems at least as good as the previous strategies in this environment and requires far less movement and therefore energy on the part of the neutrophil. The PDE governing the probability distribution here is given by: $(0.5 x_{1} e^{-(\abs{x} - R)}, 0.5 x_{2} e^{-(\abs{x} - R)})^{\mathrm{T}} \cdot f(x) + \frac{1}{2} \Delta f(x) = 0$ in $\Omega$, $f = 0$ in $\pd B_{0.2}(0)$ and $f = 1$ in $\pd B_{1}(0)$, where $\Omega = \{ x \in \R^{2} \vert R \leq x \leq kR \}$.}
				\label{fig:Drift}
			\end{figure}
			\begin{figure}[h!]
				\centering
				\subfigure[]{\includegraphics[width=3.1in]{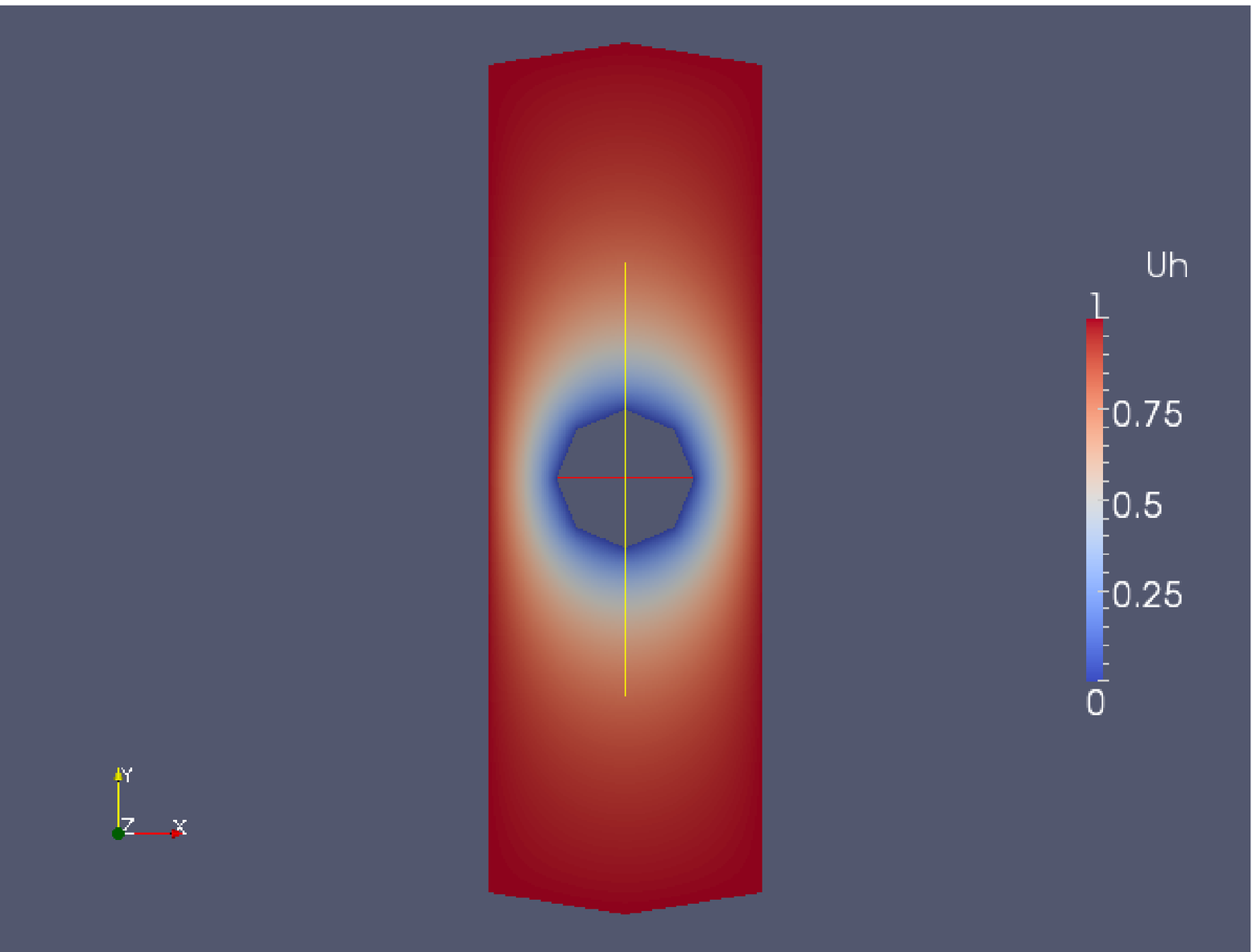}}
				\subfigure[]{\includegraphics[width=3.1in]{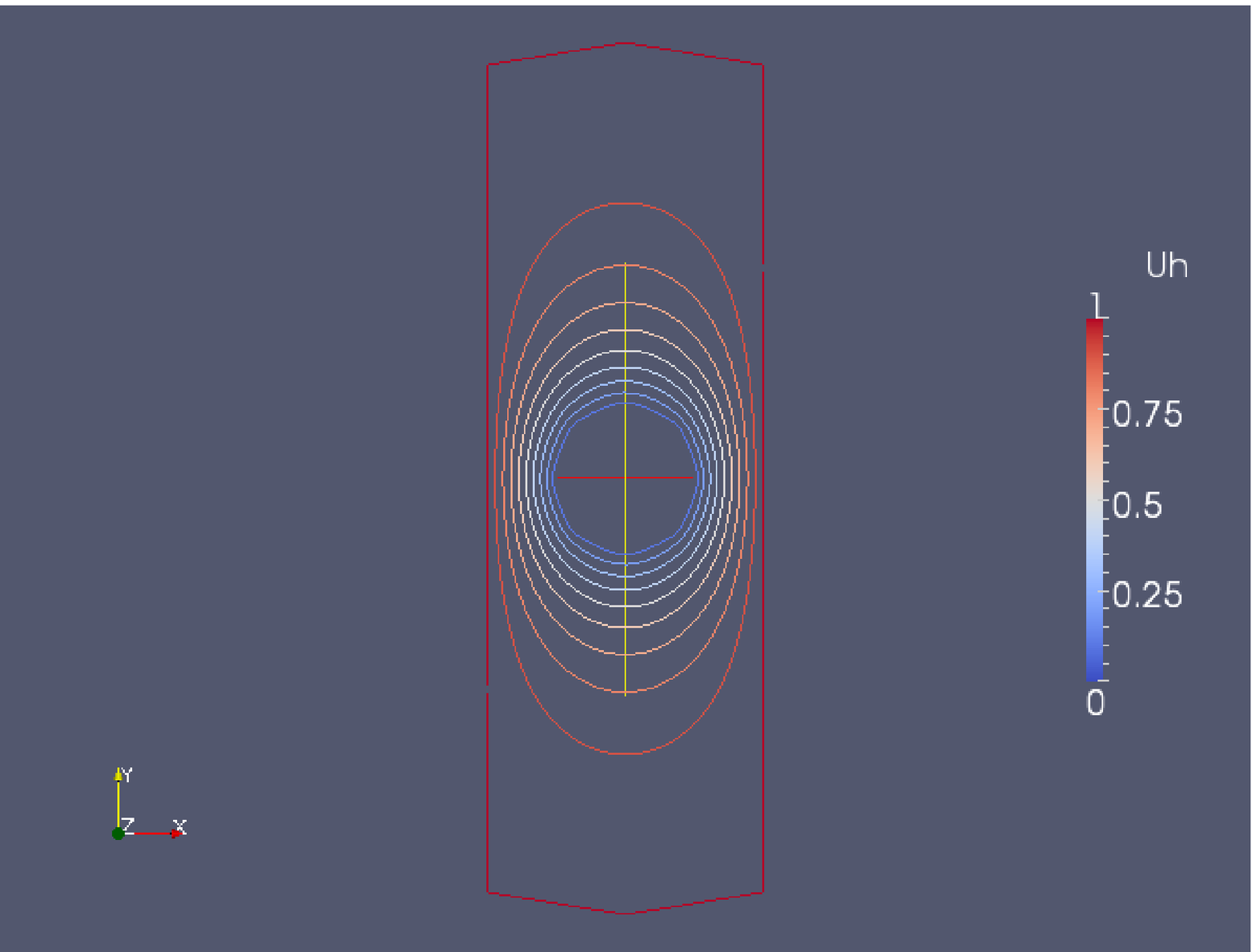}}\\
				\caption{We again consider the chemotaxical effects with the same drift term as Figure~\ref{fig:Drift}. However, here we consider the domain of a capillary, which we assume, for the moment, to be perfectly rectangular. Furthermore, we suppose that if the bacterium hits either capillary wall, or as before, it reaches an arbitrary distance away from the neutrophil, then it has escaped. This escaping can be viewed as a bacterium passing through the capillary wall. The results are similar to before, but in this case the relationship between probability of escape and distance is more pronounced on this longer domain. As we would expect, the bacterium also stands a greater chance of escape if it moves horizontally away from the neutrophil than vertically through the same distance. The equations for this situation are as given for Figure~\ref{fig:Drift}, but instead, $f = 1$ in $\Gamma = \{ x \in \R^{2} \vert \abs{x} =  \sqrt{10}, -1 \leq x_{1} \leq 1 \} \cup \{ x \in \R^{2} \vert x_{1} = -1 \} \cup \{ x \in \R^{2} | x_{1} = 1 \}$ and $\Omega = \{ x \in \R^{2} \vert R \leq x \leq kR \}$.}
				\label{fig:DriftCapillary}
			\end{figure}
			\begin{figure}[h!]
				\centering
				\subfigure[]{\includegraphics[width=3.1in]{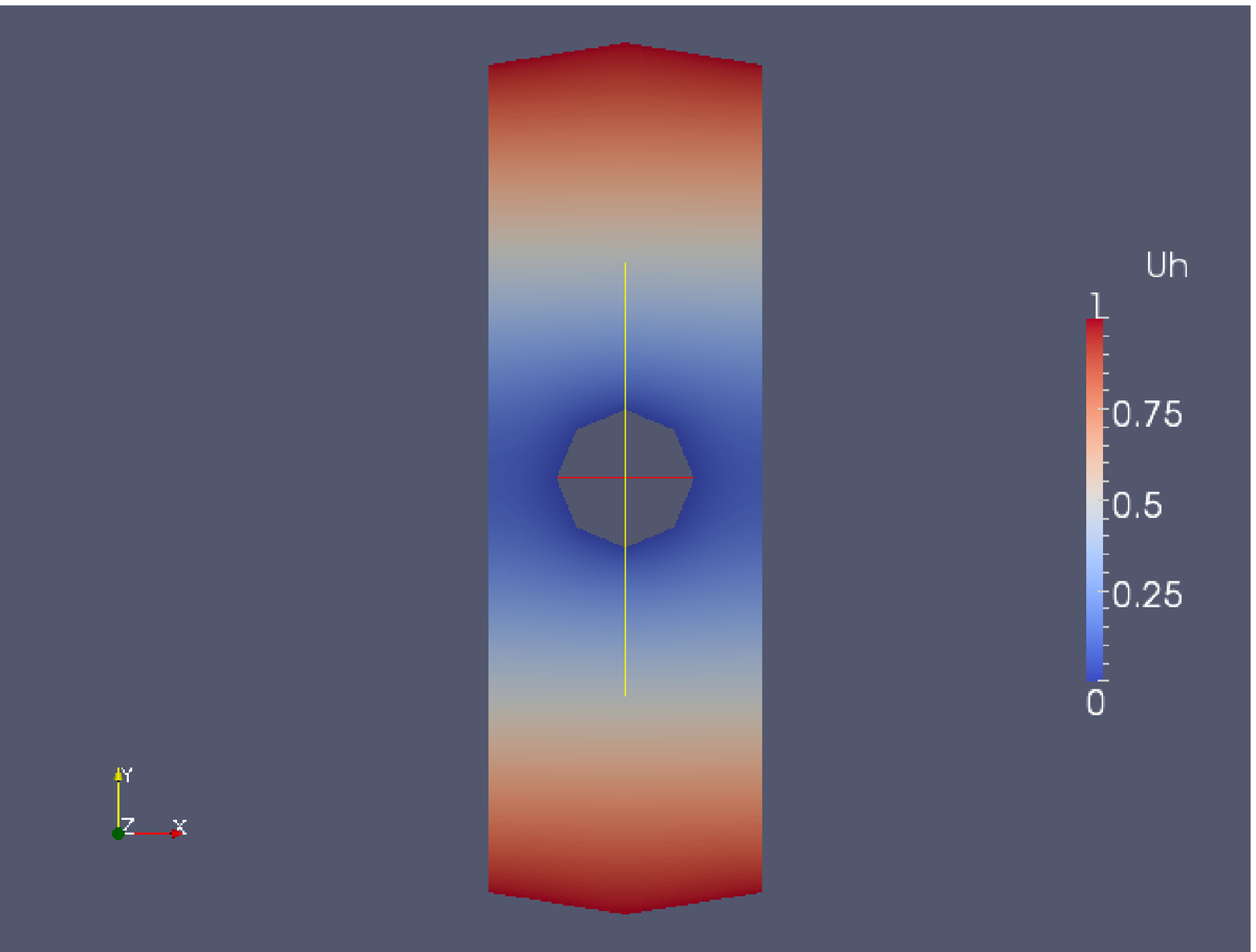}}
				\subfigure[]{\includegraphics[width=3.1in]{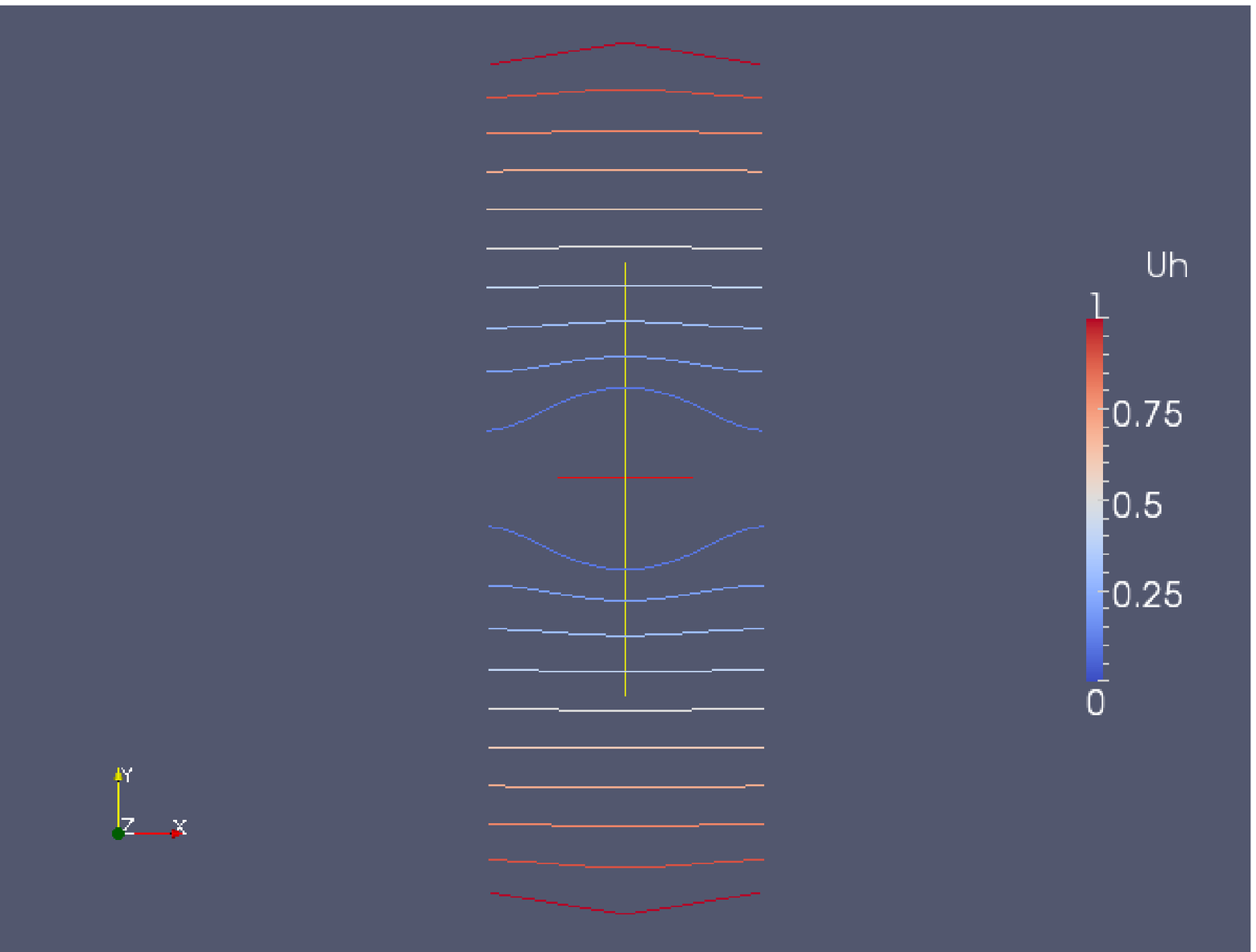}}\\
				\caption{We now consider the same situation as Figure~\ref{fig:DriftCapillary}, but in this case we assume that the capillary wall is a barrier that the bacterium cannot pass through. We model this as a reflective boundary. This is done through the use of Neumann boundary conditions on the reflective boundaries as in \cite{Br76} in addition to the already present Dirichlet conditions. In real-life we would expect the boundary to have attributes of both this figure and the previous figure. However, of the two we assume this to be the more likely and thus we continue with this approximation whilst keeping in mind the results of the previous figure as a distinct possibility that cannot be excluded in any extensive future modelling. In this figure, we also add the assumption that the bacterium will take up a larger proportion of the capillary and therefore, we limit its horizontal movement even further, but increase its size to compensate. The results are similar to before, but in this case horizontal movement of the bacterium away from the neutrophil is rewarded less than vertical movement in contradiction to the previous figure. The equations governing the escape probabilities in this case are: $\left( 0, 0.5 x_{2} e^{-(\abs{x} - R)} \right)^{\mathrm{T}} \cdot f(x) + \frac{1}{2} \Delta f(x) = 0$ in $\Omega$; $f = 0$ in $\pd B_{0.5}(0)$; $f = 1$ in $\Gamma_{1}$ and $\nabla \cdot f = 0$ in $\Gamma_{2}$, where $\Gamma_{1} = \{ x \in \R^{2} \vert \abs{x} =  \sqrt{10}, -1 \leq x_{1} \leq 1 \}$; $\Gamma_{2} = \{ x \in \R^{2} \vert x_{1} = -1 \} \cup \{ x \in \R^{2} | x_{1} = 1 \}$ and $\Omega = \{ x \in \R^{2} \vert R \leq x \leq kR \}$.}
				\label{fig:DriftCapillaryReflect}
			\end{figure}
			\begin{figure}[h!]
				\centering
				\subfigure[]{\includegraphics[width=3.1in]{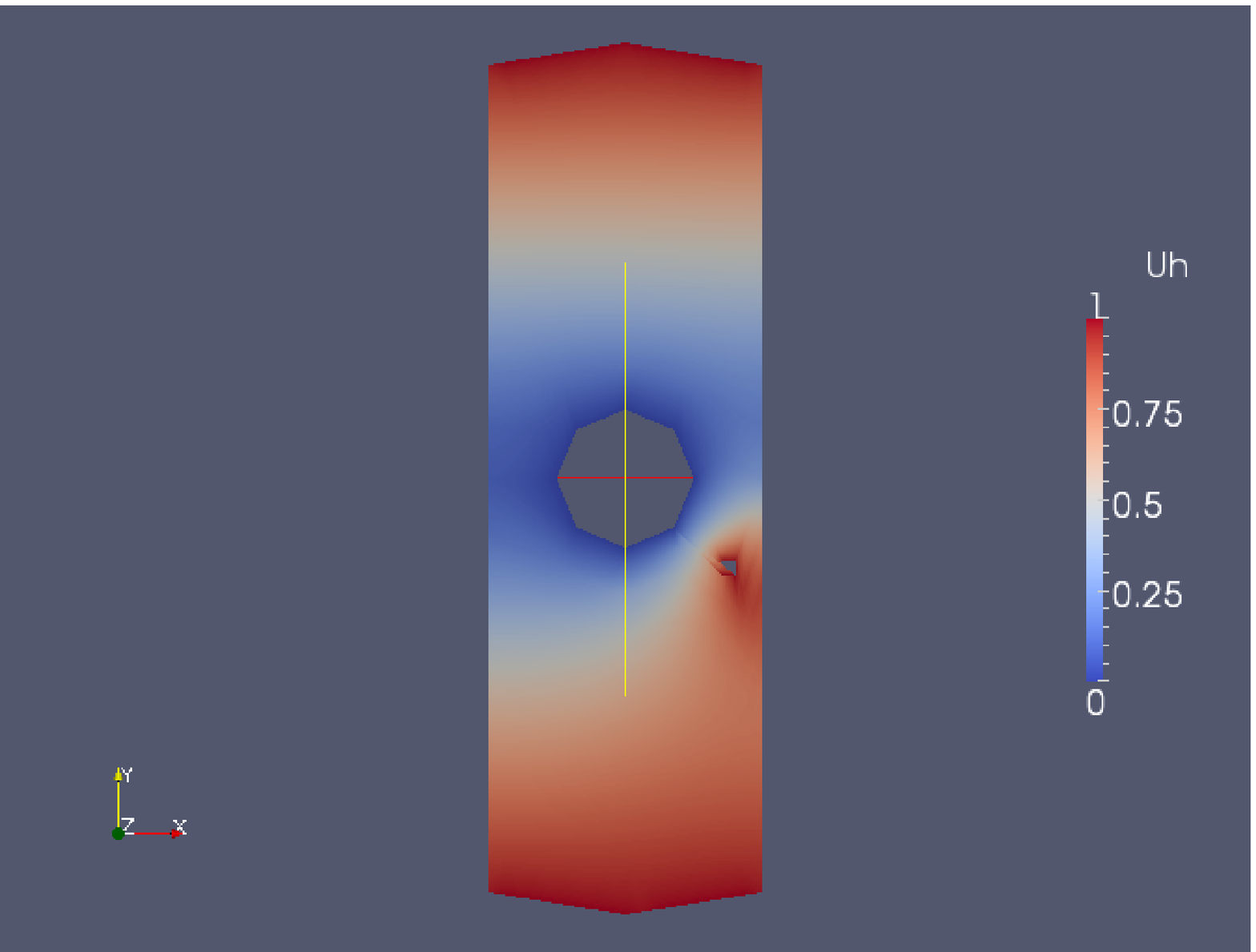}}
				\subfigure[]{\includegraphics[width=3.1in]{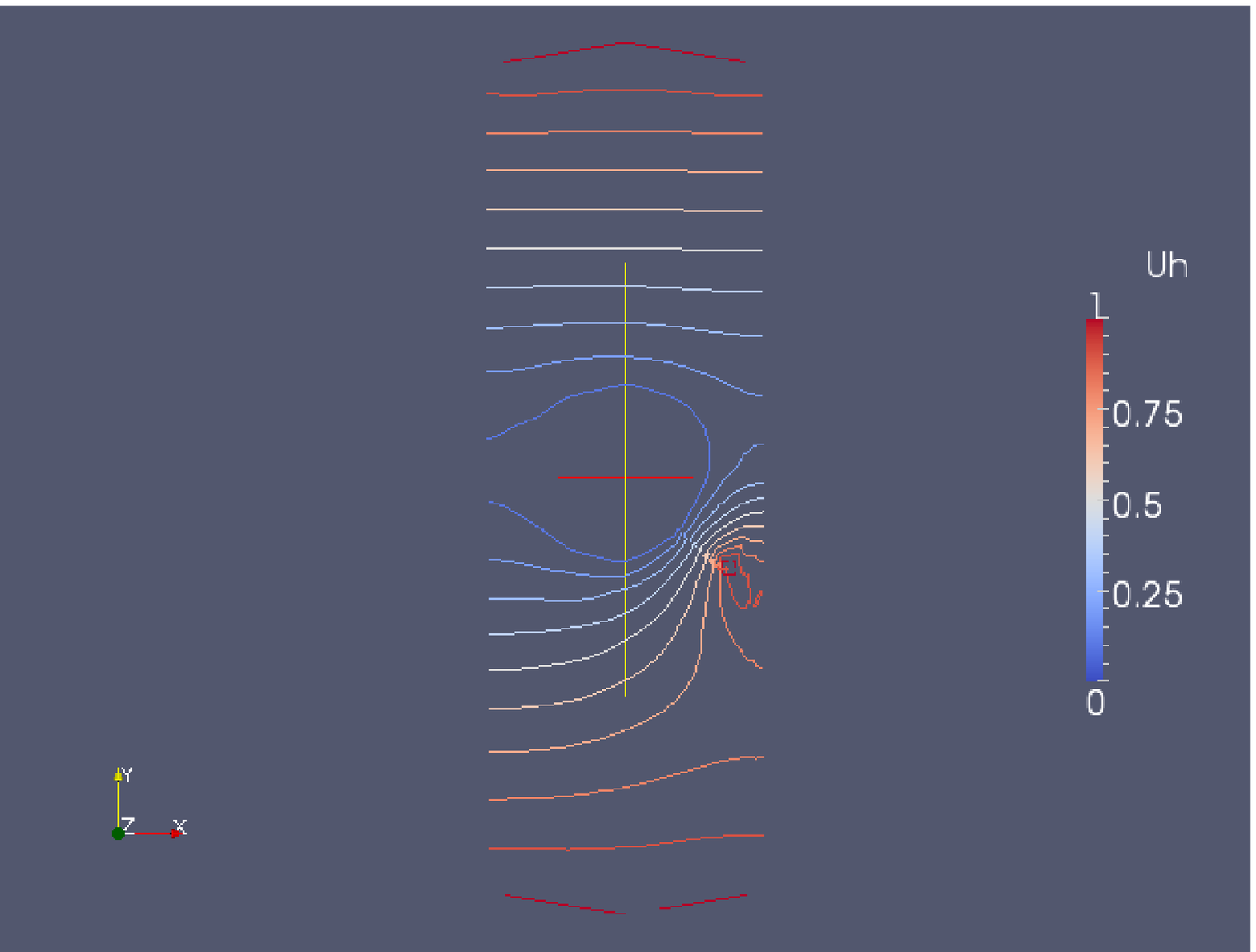}}\\
				\caption{This time we remove the drift term and, as before, we attempt to model neutrophil and bacterium movements in the domain of a capillary with impassable boundary walls; modelled as reflective boundaries. However, on this occasion we introduce an obstacle, which we view as equivalent to a passing red blood cell. We assume that the bacterium can enter or feed off the cell causing irreparable damage and thus we assume that if the bacterium reaches the cell then the neutrophil has failed in its attempts to neutralise it. It is perhaps surprising that the obstacle has such a great effect on the neutrophil's performance. In fact, it severely limits the neutrophil's ability to capture bacteria in the vicinity of the obstacle and, with several obstacles in place, one would expect the neutrophil's effectiveness to be drastically reduced. In this case we change both the domain and the drift function from the equation used in previous figures, as well as those derived in Section~\ref{sec:Prob}. We assume that the capillary is perfectly rectangular, as in previous figures, and we represent the obstacle as a square removed from the domain. Thus, the equations characterising this situation are: $\frac{1}{2} \Delta f(x) = 0$ in $\Omega$; $f = 0$ in $\pd B_{0.5}(0)$; $f = 1$ in $\Gamma_{1}$; $f = 1$ in $\Gamma_{obstacle}$ and $\nabla \cdot f = 0$ in $\Gamma_{2}$ where $\Gamma_{1} = \{ x \in \R^{2} \vert \abs{x} =  \sqrt{10}, -1 \leq x_{1} \leq 1 \}$; $\Gamma_{2} = \{ x \in \R^{2} \vert x_{1} = -1 \} \cup \{ x \in \R^{2} | x_{1} = 1 \}$; $\Gamma_{obstacle} = \mbox{the perimeter of the square spanned by the points} \; (0.8,0.8); (0.9,0.8); (0.8,0.7) \; \mbox{and} \; (0.9,0.7)$ and $\Omega = \{ x \in \R^{2} \vert R \leq x \leq kR \}$}
				\label{fig:Obstacle}
			\end{figure}

	\section{Empirical Model for Neutrophil Motion}  \label{sec:ModelComparison}
		\subsection{Experimental Data}
			Empirical data from \cite{Li2008} suggest that in the absence of the chemoattractant signal from a bacterium, neutrophils move in a roughly zig-zag manner. It is suggested that this could be a strategy for improving the efficiency with which the cell searches an area of space for bacteria. We now compare the PDE models developed earlier to the experimental data using a model based on its statistical analysis, and postulate some ways of determining how efficient the search strategy is compared to, say, a simple random walk.
			The data suggest that the cell moves in an essentially straight line for a random time with an Exponential(0.67) distribution, measured in seconds.  After each such time period, the cell makes a turn with magnitude determined by another Exponential(0.67) distribution (in radians). The history of left/right turns can be modelled as a discrete Markov chain in which the probability of turning the opposite direction to previously (that is, a left turn followed by a right turn or \textit{vice versa}) is approximately 2.1 times the probability of turning in the same direction as before: the analysis of 4822 turns from 12 observed trajectories reveals that opposite-type turning pairs outnumbered same-type pairs 3623 to 1559. The cell is assumed to move at a constant speed of $7.46\times10^{-6} \mbox{ms}^{-1}$ \cite{Li2008}.

		\subsection{Comparison with PDE Model}
			We should expect our reaction-diffusion model for chemical activators on the membrane to produce results similar to those of this second, less complicated model, backed up by experimental data.  It is suggested by Li \emph{et al.} that the cell's movement is controlled by the formation of pseudopods, which cause the cell to move in the direction they are extended. The zig-zag path indicates that the production of pseudopods should therefore occur in an alternating left-right sequence and this is behaviour we can examine in the results of the PDE model.

			The PDE model shows that the formation of a pseudopod $A$ is followed by a period in which the pseudopod splits in two. The stochastic forcing term in the reaction-diffusion equation \eqref{eq:Neilson-a} causes one of these peaks to eventually dominate and grow into a new pseudopod.  This is most likely to be the peak closest to $A$'s parent. New pseudopods therefore form more often between the two most recent extensions, and this accounts for the observed alternating turn behaviour. We postulate that this final aspect of the pseudopods' behaviour may be predicted by the PDE model: see, for example, the extensions section.

		\subsection{Analysis of the Search Strategy}
			The cell's aim should be to quickly locate the chemoattractant signal of a bacterium, so we expect it to have evolved a search strategy for doing this in an efficient way.  A reasonable requirement is for the cell to explore the space rapidly without covering the same region too often or getting stuck in one place.  For comparison we may consider how the strategy improves upon a typical random walk without directional bias.
				\begin{figure}[h!]
					\centering
					\includegraphics[height=110mm]{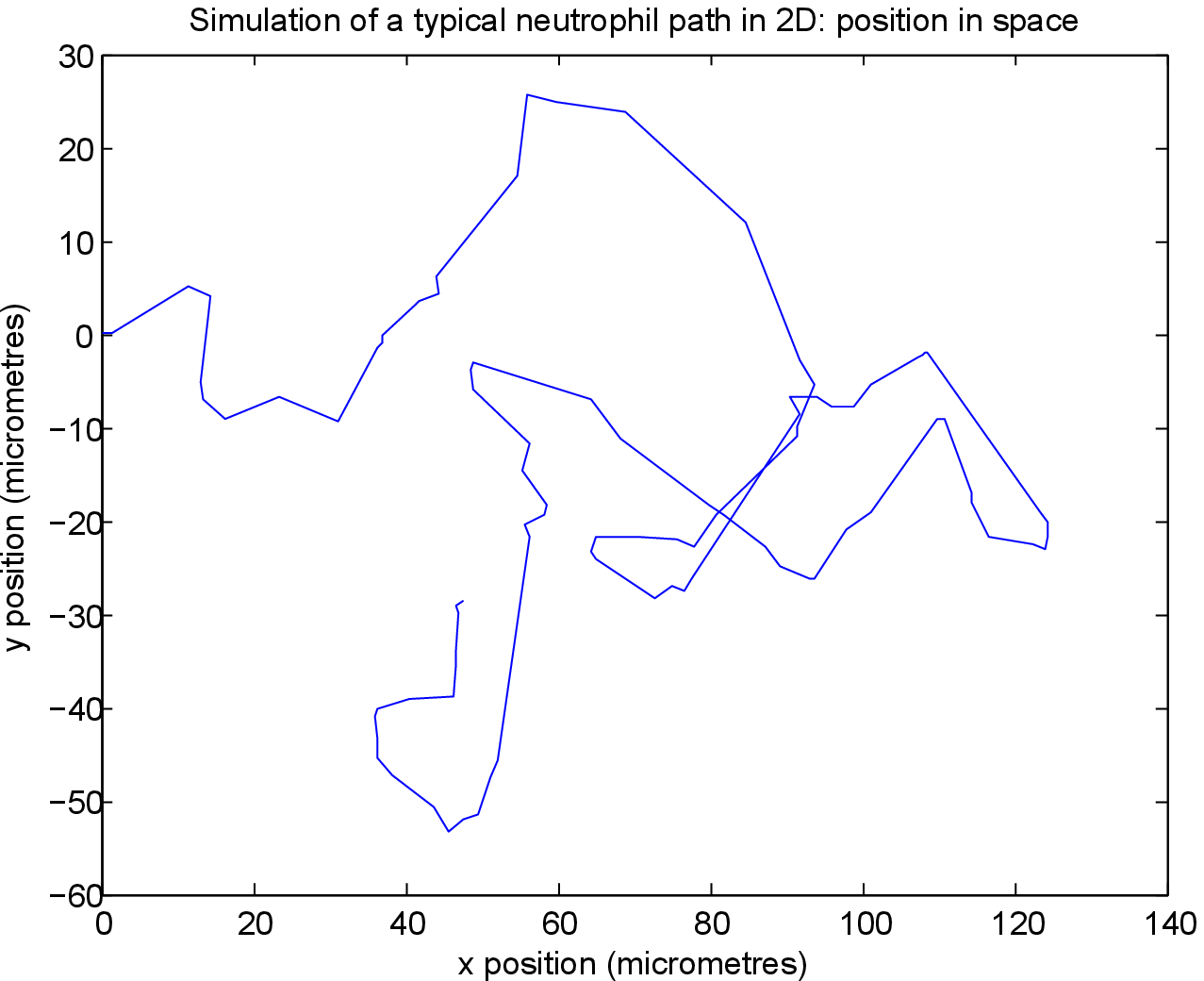}
					\caption{Position plots of a simulated cell movement, starting at the origin.}
					\label{fig:cellposition}
				\end{figure}
				\begin{figure}[h!]
					\centering
					\includegraphics[height=110mm]{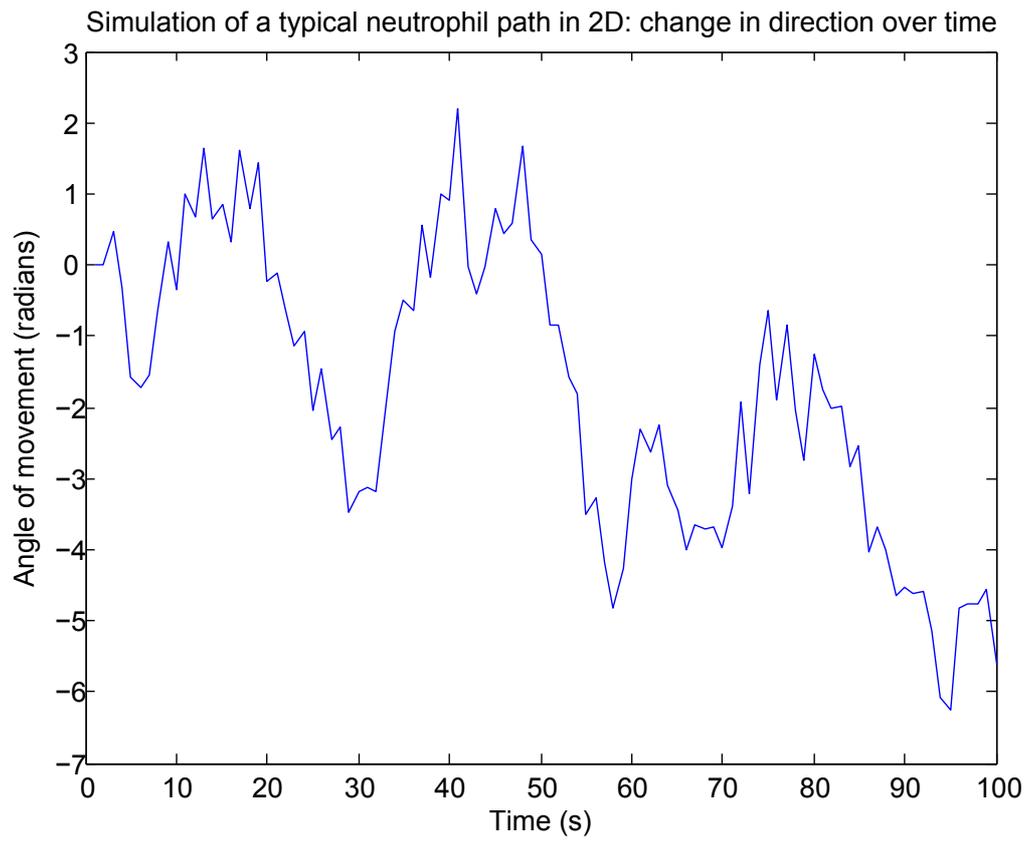}
					\caption{Direction of movement over time for the same simulation.}
					\label{fig:cellangle}
				\end{figure}
				\begin{figure}[h!]
					\centering
					\includegraphics[height=110mm]{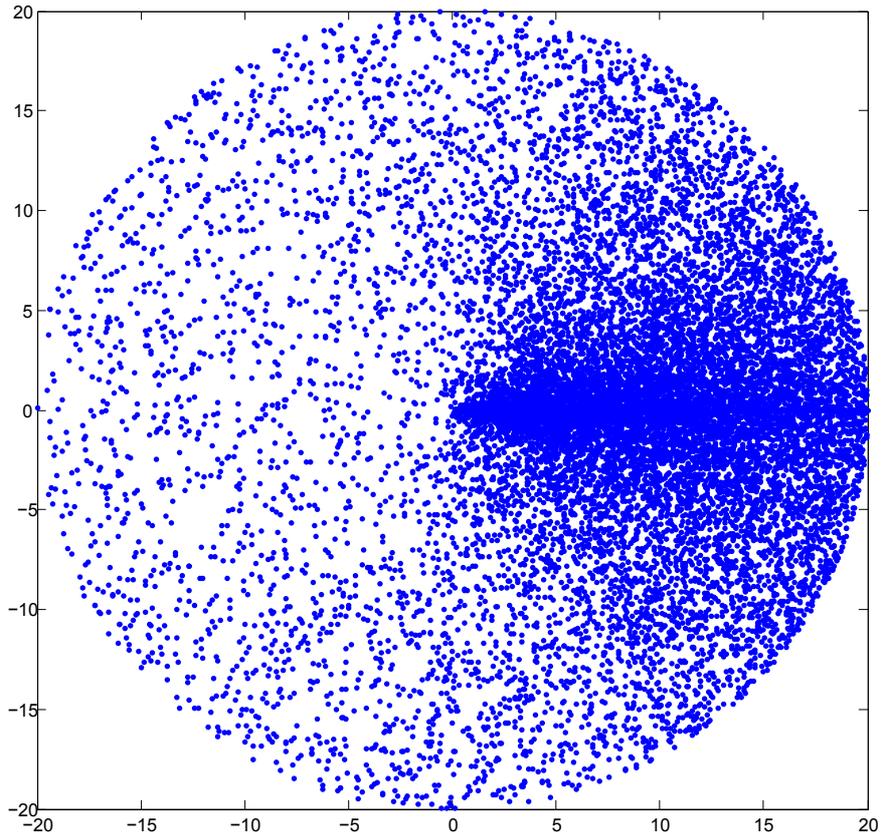}
					\caption{Plots of escape position for 200 uniformly spaced choices of circle radius from 0.1 - 20, 100 independent simulations per radius.}
					\label{fig:escapepoints}
				\end{figure}
				\begin{figure}[h!]
					\centering
					\subfigure[$r = 20$]{\includegraphics[width=3.3in]{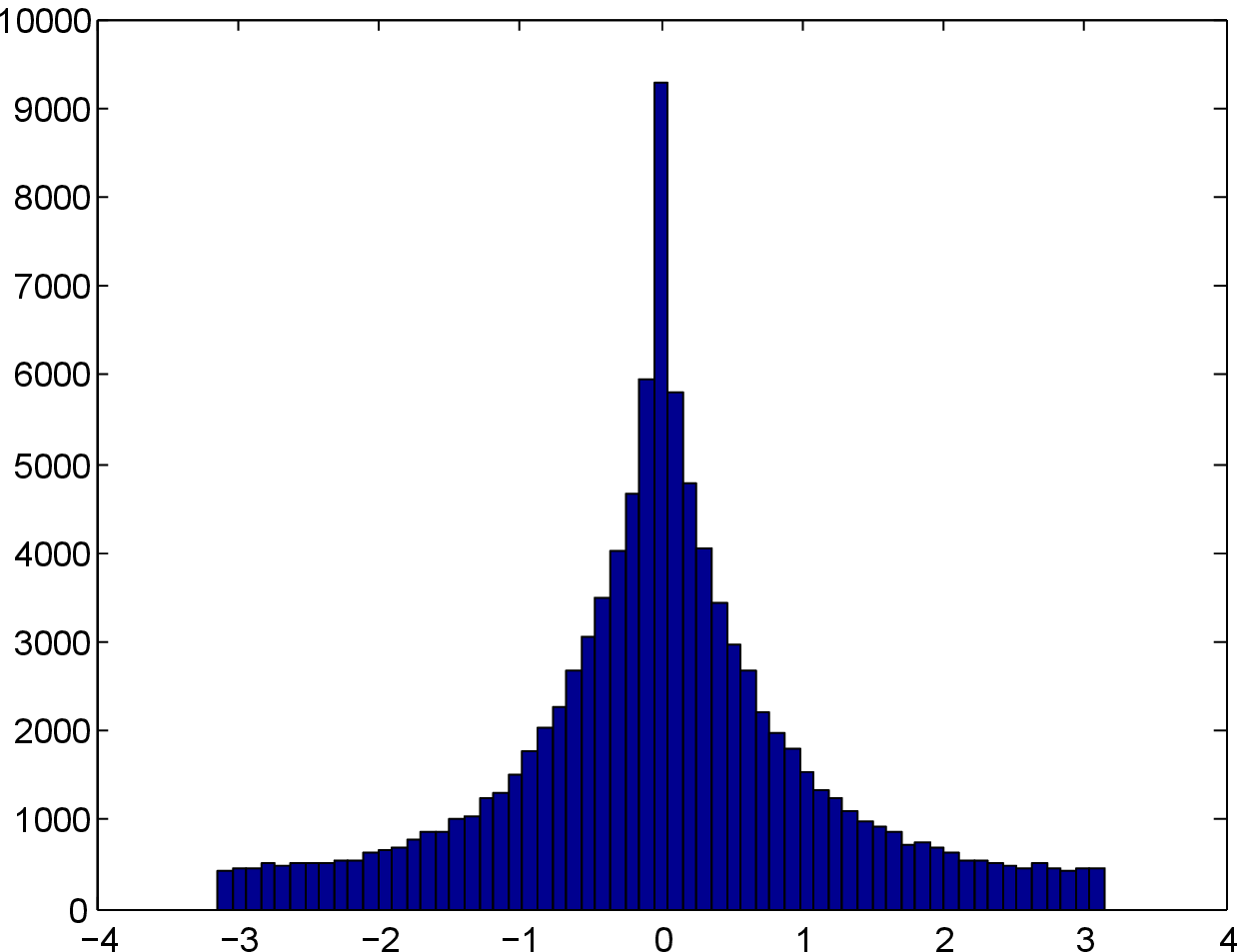}}
					\subfigure[$r = 50$]{\includegraphics[width=3.3in]{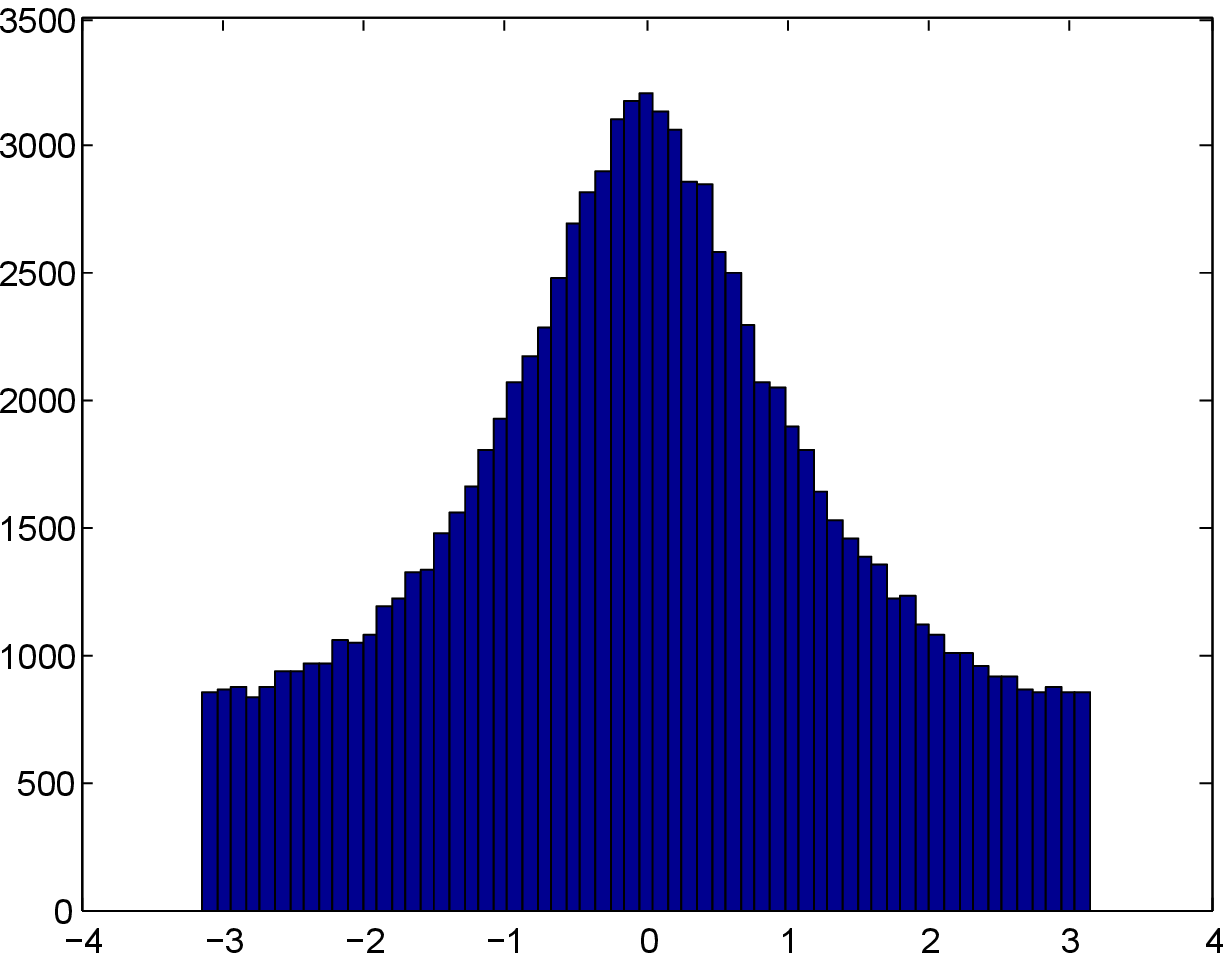}}\\
					\subfigure[$r = 100$]{\includegraphics[width=3.3in]{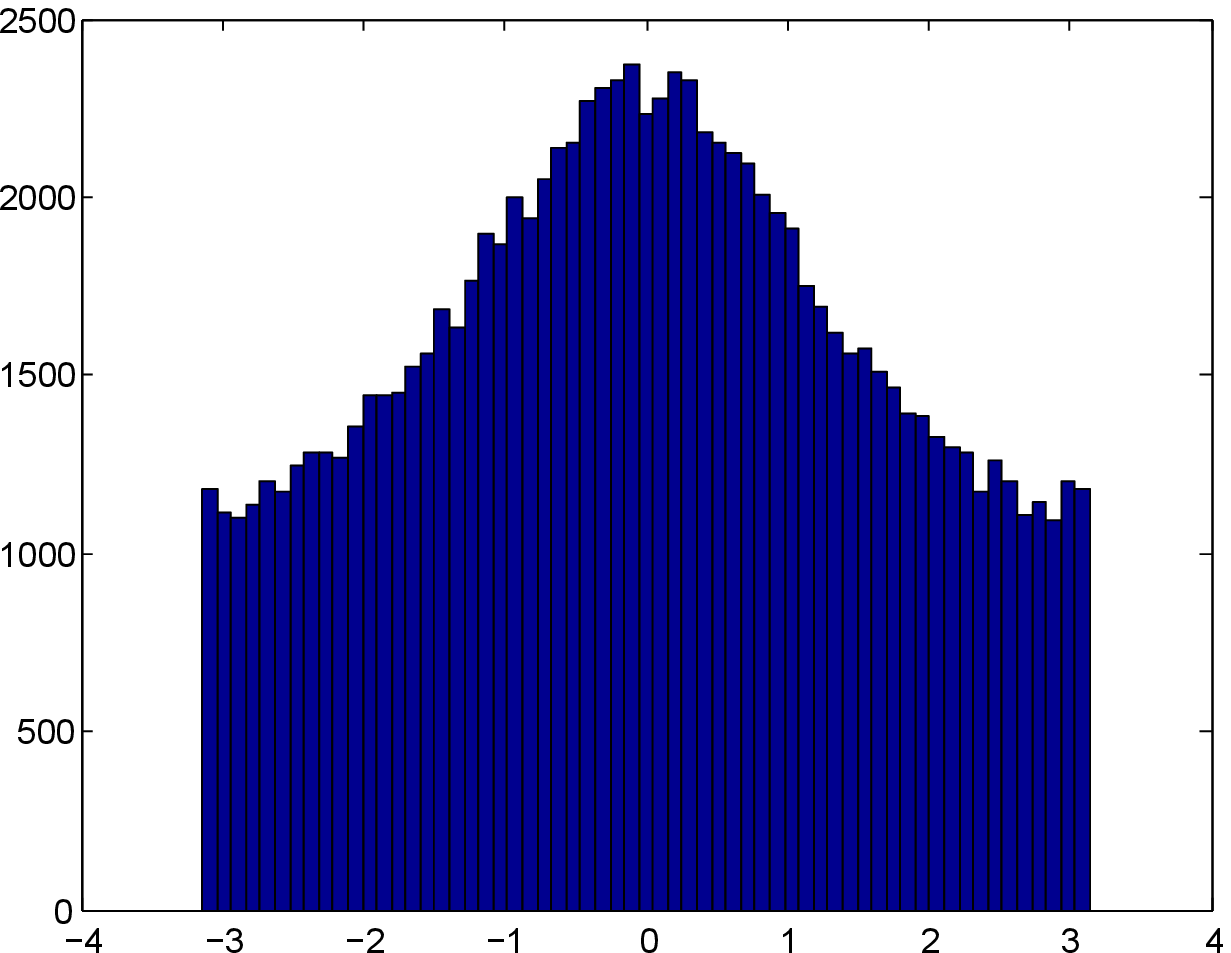}}
					\subfigure[$r = 200$]{\includegraphics[width=3.3in]{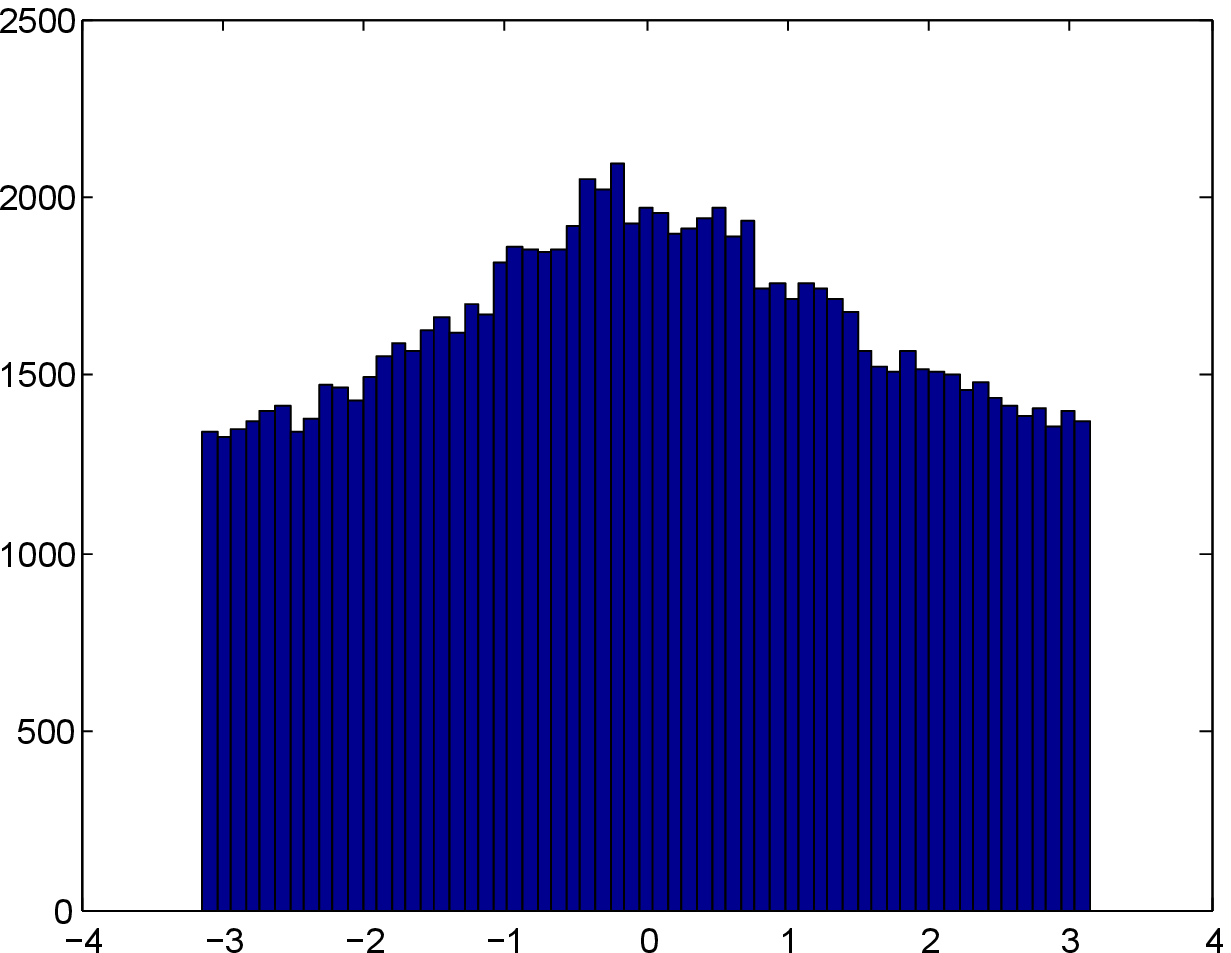}}\\
					\caption{Histograms comparing the distribution of exit angles for radii r = 20, 50, 100 and 200, based on 100,000 simulations each.}
					\label{fig:angles}
				\end{figure}
				\begin{figure}[h!]
					\centering
					\includegraphics[scale=0.75]{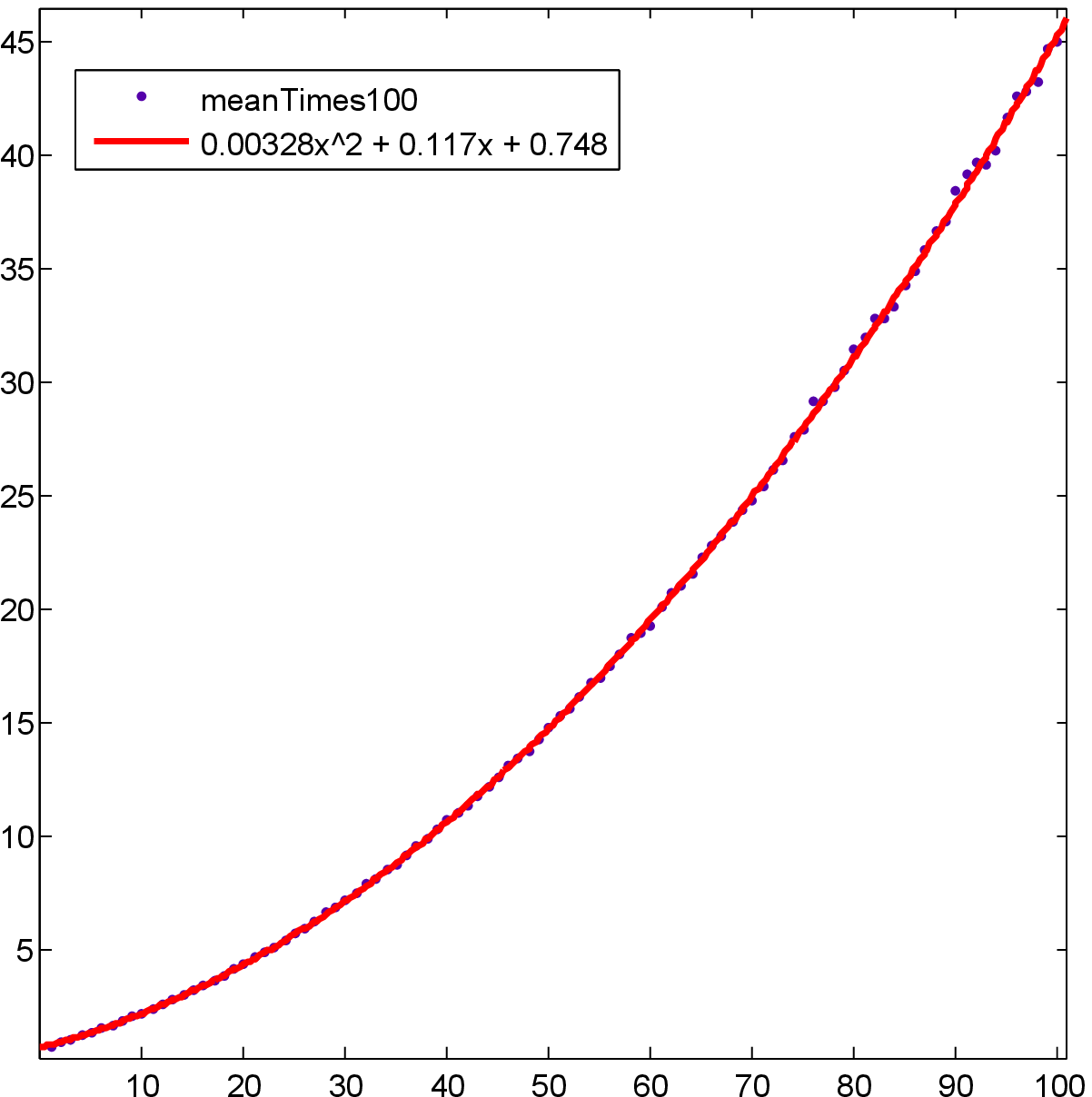}
					\caption{Mean escape times plotted against radius of the circle, based on 10,000 simulations.}
					\label{fig:escapetimes}
				\end{figure}

			An example cell path is given by Figure~\ref{fig:cellposition}; the turn history for the same path is illustrated in Figure~\ref{fig:cellangle}. The type of motion modelled here is a \textit{persistent} random walk, since the direction of movement is dominated by a slowly-changing parameter $\theta$ around which the frequent turns impose a highly variable $\phi$ term. We can observe this tendency of the cell to keep moving in roughly the same direction by conditioning its motion to begin in the positive $x$ direction and observing the points at which it exits a circle of a given radius $r$ (centred at the origin). Figure~\ref{fig:escapepoints} plots the exit points for 200 choices of radius at regular intervals from 0.1 to 20. The simulation for each radius was run 100 times and all exit points correspond to independent paths. The high concentration of points surrounding the positive $x$-axis clearly indicates the tendency of the cell to keep moving in the direction it started on average; however we have also seen how the high frequency of turns within this overall trajectory causes the zig-zag motion mentioned earlier. The result is a `sweeping' procedure across the plane that may help the cell to explore a large area efficiently and without going back on itself too often.

			Histograms for the cases $r = 20, 50, 100, 200$, as shown in Figure~\ref{fig:angles} (100,000 simulations each) further press the point, and we note for these higher radii that the influence of the starting direction on the exit point diminishes as the radius increases. This shows that while the cell trajectory has a high level of local persistence, globally it does not confine itself to a particular area of the plane. We conclude that it is unlikely for the cell to continue moving in the same average direction for an excessive amount of time; this is an important consideration in the real-world application of the model where the space is not infinite and therefore not all areas are equally worth exploring. We see that even though the cell in our model is not confined to a finite space, its tendency is not to continue moving in the same direction for an excessive time period.

			Figure~\ref{fig:escapetimes} shows the mean exit times for each integer radius 10-200. The curve is well-fitted by a quadratic $y = 0.00328x^2 + 0.117x + 0.748$. Long-term behaviour is therefore that movement from the origin occurs on the timescale $r^2$, like a random walk, whereas short-term behaviour is on the much faster timescale $r$. This further backs the conclusion that the cell explores local regions quickly, but dwells for a longer time in larger areas.

	\section{Conclusion \& Further Work} \label{sec:Conclusion}
		In this report, we have taken the model of \cite{Neil10} and reduced it to a single reaction-diffusion PDE based on simulation data provided in the paper. In addition we have also proposed a new model for the movement of the neutrophil membrane using a mean curvature approach. Numerical methods for the reaction-diffusion PDE and the neutrophil movement have been developed and implemented to simulate the system. These simulations do not support the observations made in \cite{Neil10}, namely the pseudopod splitting, suggesting that our reduced model was oversimplified. On reintroduction of the local inhibitor PDE, resulting simulations show this pseudopod splitting behaviour. Therefore, our model gives the same results as in \cite{Neil10}, with the additional benefit of having a simpler model for the neutrophil movement, which does not involve solving a non-linear ODE.

		In addition to the PDE approach, we also modelled the neutrophil in its pursuit of bacteria through the use of SDEs. On an infinite plane it was suggested that movement along a constant vector would be the most efficient strategy for the neutrophil, but on the restricted environments commonly found in practice our model has found this particular strategy to be ineffective. Moving on from this conclusion, we then found models that included chemotaxical effects to be far more fruitful for the neutrophil in these situations, as one would suspect. However, the model showed a neutrophil solely based on a response to chemoattractants to be lacking over longer distances; suggesting that a compromise between these two strategies would be an optimum evolutionary choice for a neutrophil that may find itself in a variety of surroundings and would need to perform accordingly. Results, at least partially, supported by our other approaches in the relevant sections. Finally, our model portrayed a very different picture when obstacles were added. Though our model is a very rough approximation, this perhaps suggests that the neutrophil will be at a slight disadvantage in densely populated areas of cells.
		
		Finally, our simulations based on the experimental data indicated that the cell's search strategy is to zig-zag frequently back and forth over an area while travelling in one dominant direction for longer periods of time. The high persistence of the cell's motion means that it explores its local area much more quickly than a standard random walk while dwelling for longer in larger areas, as we saw from the analysis of exit times and exit points for circles of varying radii. The PDE model reflects some of the observations regarding pseudopod behaviour.

		One way to extend the work done in this project is to consider an alternative model for the neutrophil membrane movement. The natural extension would be to consider a Willmore flow model since this minimises an energy functional related to the curvature of the membrane. We can also incorporate the bacterium via the full expression of the stochastic term $(R_{0} \neq 0)$, and the SDE model developed in Section~\ref{sec:Prob} can be used as a benchmark for comparison. In particular, it would be interesting to consider the effect (caused by morphing of the membrane) of chemoattractants ($\theta$) on the neutrophil movement. Finally, we could compare simulated paths produced by both the PDE model and the empirical model using statistical comparison methods. Properties such as the expected exit times and exit position probabilities could be used for these comparisons. The PDE model could also be extended to better reflect the observed positioning of new pseudopods.

	\newpage

	\addcontentsline{toc}{section}{References}
	\bibliography{BiomembranesReport}

	\begin{table}
		\centering
		\begin{tabular}{|c|l|c|}
			\hline
			Quantity & Description & Value
			\\
			\hline
			$r_{a}$ & Decay rate of activator & $2 \times 10^{-2}$
			\\
			$r_{b}$ & Decay rate of global inhibitor & $3 \times 10^{-2}$
			\\
			$r_{c}$ & Decay rate of local inhibitor & $1.3 \times 10^{-2}$
			\\
			$D_{a}$ & Diffusion Coefficient of activator & $4 \times 10^{-7}$
			\\
			$D_{b}$ & Diffusion Coefficient of global inhibitor & $4$
			\\
			$D_{c}$ & Diffusion Coefficient of local inhibitor & $2.8 \times 10^{-6}$
			\\
			$b_{a}$ & Basal activator production rate & $1 \times 10^{-1}$
			\\
			$b_{c}$ & Local inhibitor production rate & $5 \times 10^{-3}$
			\\
			$s_{a}$ & Saturation Coefficient & $5 \times 10^{-4}$
			\\
			$s_{c}$ & Michaelis-Menten constant & $2 \times 10^{-1}$
			\\
			\hline
		\end{tabular}
		\caption{Table of Reaction-Diffusion equation constants}
		\label{tab:EqnConstants}
	\end{table}

\end{document}